\documentclass[a4paper,UKenglish]{lipics}

\usepackage{xcolor,pgf,tikz,pgflibraryarrows,pgffor,pgflibrarysnakes}
\usepackage{stmaryrd}
\usepackage{amssymb,amsmath}
\usepackage{url,xspace}
\usepackage{dsfont}
\usepackage{yhmath,wasysym,longtable,array,enumerate}
\usepackage{new-decisive}
\usepackage{thm-restate}
\usepackage{fourier-orns}
\usetikzlibrary{snakes,fit}
\usetikzlibrary{arrows.meta}

\usepackage{todonotes}

\title{When are Stochastic Transition Systems Tameable?}
\titlerunning{When are Stochastic Transition Systems Tameable?}

\author[1]{Nathalie Bertrand}
\author[2]{Patricia Bouyer\thanks{Supported by ERC project EQualIS}}
\author[3]{Thomas Brihaye}
\author[2,3]{Pierre Carlier$^*$}
\affil[1]{Inria Rennes, France}
\affil[2]{LSV, CNRS \&  ENS Cachan, France}
\affil[3]{Universit\'e de Mons, Belgium}
\authorrunning{N. Bertrand, P. Bouyer, T. Brihaye, and P. Carlier}
\begin{document}

\maketitle

\begin{abstract}
A decade ago, Abdulla, Ben Henda and Mayr
introduced the elegant concept of decisiveness for denumerable Markov
chains~\cite{ABM07}.  Roughly speaking, decisiveness
allows one to lift most good properties from finite Markov chains to
denumerable ones, and therefore to adapt existing verification
algorithms to infinite-state models.  Decisive Markov chains however
do not encompass stochastic real-time systems, and general stochastic
transition systems (STSs for short) are needed. In this article, we
provide a framework to perform both the qualitative and the
quantitative analysis of STSs. 
First, we define various notions of decisiveness (inherited
from~\cite{ABM07}), notions of fairness and of attractors for STSs,
and make explicit the relationships between them.  Then, we define a
notion of abstraction, together with natural concepts of soundness and
completeness, and we give general transfer properties, which will be
central to several verification algorithms on STSs.  We further design
a generic construction which will be useful for the analysis of
$\omega$-regular properties, when a finite attractor exists, either in
the system (if it is denumerable), or in a sound denumerable
abstraction of the system. We next provide algorithms for qualitative
model-checking, and generic approximation procedures for quantitative
model-checking.
%
%
Finally, we instantiate 
our framework with stochastic timed automata (STA), generalized
semi-Markov processes (GSMPs) and stochastic time Petri nets (STPNs),
three models combining dense-time and probabilities.  This allows us
to derive decidability and approximability results for the
verification of these models. Some of these results were known from
the literature, but our generic approach permits to view them in a
unified framework, and to obtain them with less effort. We also derive
interesting new approximability results for STA, GSMPs and STPNs.

\end{abstract}

\newpage
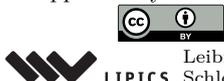
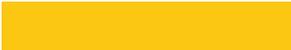
\tableofcontents
\newpage



\section{Introduction}
Given its success for finite-state systems, the model checking
approach to verification has been extended to various models based on
automata, and including features such as time, probability and
infinite data structures. These models allow one to represent software
systems more faithfully, by representing timing constraints,
randomization, and \emph{e.g.} unbounded call stacks.  At the same
time, they often offer the possibility to consider \emph{quantitative}
verification questions, such as whether the best execution time meets
a requirement, or whether the system is reliable with high
probability.  Quantitative verification is notably hard for
infinite-state systems, and often requires the development of
techniques dedicated to each class of models.

A decade ago, Abdulla, Ben Henda and Mayr introduced the concept of
decisiveness for denumerable Markov chains~\cite{ABM07}.  Formally, a
Markov chain is decisive w.r.t. a set of states $F$ if runs
almost-surely reach $F$ or a state from which $F$ can no longer be
reached. The concept of decisiveness thus forbids some weird
behaviours in denumerable Markov chains, and allows one to lift most
good properties from finite Markov chains to denumerable ones, and
therefore to adapt existing verification algorithms to infinite-state
models. In particular, assuming decisiveness enables the quantitative
model checking of (repeated) reachability properties, by providing an
approximation scheme, which is guaranteed to terminate for any given
precision for decisive Markov chains. Decisiveness also elegantly
subsumes other concepts such as the existence of finite attractors, or
coarseness~\cite{ABM07}.

\medskip Decisive Markov chains however are not general enough to
represent stochastic real-time systems. Indeed, to faithfully model
time in real-time systems, it is adequate to use dense
time~\cite{alur91}, that is, timestamps of events are taken from a
dense domain (like the set of rational or of real numbers). This
source of infinity for the state-space of the system is particularly
difficult to handle: the state-space is non-denumerable (even
continuous), the branching in the transition system is also
non-denumerable, \textit{etc}. For those reasons, stochastic real-time
systems do not fit in the framework of decisive Markov chains
of~\cite{ABM07}.

Also, standard analysis techniques for non-stochastic real-time
systems (when they exist) cannot be easily adapted to stochastic
real-time systems. Traditionally, these techniques rely on the design
of appropriate finite abstractions, which preserve good properties of
the original model. A prominent example of such an abstraction is that
of the region automaton for timed automata~\cite{AD94}. However, these
abstractions 
do not preserve all quantitative properties and, in the context of
stochastic systems they may be too coarse already for the evaluation
of the probability of properties as simple as reachability properties.


\medskip A general framework to analyse a large class of stochastic
real-time systems, or more generally continuous stochastic systems, is
thus lacking. In this article, we face this issue and provide a
framework to perform the 
\emph{stochastic transition systems} (STSs for short). To do so, we
generalize the main concepts of~\cite{ABM07} (such as decisiveness,
attractors), and standard notions for Markov chains (like
fairness). STSs are purely stochastic (i.e. without non-determinism)
Markov processes~\cite{Pan01,Pan09},
that is, Markov chains with a continuous state-space. Note that, while
this journal version builds on the conference paper~\cite{BBBC16}, we
choose here to phrase our results for time-homogeneous and Markovian
models. As mentioned in~\cite{Pan01}, 
the Markovian assumption
is not a severe restriction since many apparently non Markovian
processes can be recast to Markovian models by changing the state
space. In our opinion, this choice furthermore enables the design of a
richer and more elegant theory (compared to~\cite{BBBC16}).



\medskip Our first contribution is to define various notions of
decisiveness (inherited from~\cite{ABM07}), notions of fairness and of
attractors in the general context of STSs.  To complete the semantical
picture, we make explicit the relationships between these notions, in the
general case of STSs, and also when restricting to denumerable Markov
chains. Decisiveness or the existence of attractors will be later
exploited to analyze properties for STSs.

As mentioned earlier, the analysis of real-time systems often requires
the development of abstractions. As a second contribution, we define a
notion of abstraction, which makes sense for STSs. Concepts of
soundness and completeness are naturally defined for those
abstractions, and general transfer properties are given, which will be
central to several verification algorithms on STSs. The special case
of denumerable abstractions is discussed, since it allows one to
transfer more properties from the abstract system to the concrete one.

We then focus on denumerable Markov chains with a finite attractor, or
more generally STSs admitting a sound abstraction satisfying this
property, and an $\omega$-regular property represented by a
deterministic Muller automaton. Our third contribution consists in
building a graph for the attractor, which contains enough information
to analyze the probability that the STS satisfies the property. This
is is completely new compared to the original results of~\cite{ABM07}
and our conference paper~\cite{BBBC16}. It is inspired by a procedure
of~\cite{ABRS05} for probabilistic lossy channel systems, a special
class of denumerable Markov chains with a finite attractor.

Our fourth contribution concerns the qualitative model checking
problem for various properties. In particular, we extend the results
of~\cite{ABM07} and show that, under some decisiveness assumptions,
the almost-sure model checking of (repeated) reachability properties
reduces to a simpler problem, namely to a reachability problem with
probability $0$. We advocate that this reduction simplifies the
problem: in countable models, the $0$-reachability amounts to the non
existence of a path, in the underlying non-probabilistic system;
beyond countable models, checking that a reachability property is
satisfied with probability $0$ amounts to exhibiting a somehow regular
set of executions with positive measure. Beyond (repeated)
reachability properties, we apply our above-mentioned approach via the
graph of an attractor for the qualitative analysis of $\omega$-regular
properties.

Our fifth contribution is the design of generic approximation
procedures for the quantitative model-checking problem, inspired by
the path enumeration algorithm of Purushothoman Iyer and
Narashima~\cite{IN97}. Under some decisiveness assumptions, we prove
that these approximation schemes are guaranteed to terminate.
Assuming the STSs can be represented finitely and enjoy some smooth
effectiveness properties, one derives approximation algorithms: one
can approximate, up to a desired (arbitrary) precision, the
probability of (repeated) reachability properties. Note that without
these effectiveness properties, one cannot hope for algorithms, and
this motivates our above formulation of ``procedures''.  Further, once
again via the use of the graph of an attractor, we design an
approximation algorithm for $\omega$-regular properties; this
algorithm reduces the quantitative analysis of an $\omega$-regular
property to the quantitative verification of a reachability property
in the concrete model. Up to our knowledge, this approach is
completely new, and provides an interesting framework for quantitative
verification of stochastic systems.

Our last contribution consists in instantiating our framework with
high-level stochastic models, stochastic timed automata (STA),
generalized semi-Markov processes (GSMP) and stochastic time Petri
Nets (STPN), which are three classes of models combining dense-time
and probabilities.  This allows us to derive decidability and
approximability results for their verification. Some of these results
were known from the literature, \emph{e.g.} the ones
from~\cite{BBB+14}, but our generic approach permits to view them in a
unified framework, and to obtain them with less effort. We also derive
interesting new approximability results for STA and GSMPs. In
particular, the approximability results implied by this paper for STA
are far more general than those obtained using an \emph{ad-hoc}
approach in~\cite{BBBM08}. In the case of STPNs, we also interestingly
embed the framework of~\cite{HPRV-peva12,PHV16} into our setting,
which allows to show that we can relax some assumptions while
preserving approximability results.

The paper concludes with an overview of our main results, organized as
a guided tour of the STSs: it summarizes the relationships between all
notions, and provides the reader recipes to analyze STSs.

In the interest of readability, most technical proofs are postponed to
the appendix, with clear pointers.  The emphasis is put on our new
approach to the  analysis of $\omega$-regular properties, which
remains in the core of the paper (Section~\ref{sec:main}).

\subparagraph*{Other related works.}  Apart from direct related works
that we have already mentioned (like~\cite{ABM07}), let us review
related work from the literature. First, in~\cite{DDP03}, approximants
are given, which rely on refinements of a partitioning of the
state-space of an STS (via conditional expectations).  However, there
is no stopping criteria if we want to turn these approximants to a
proper approximation scheme. And the approach is also very different.

Then, for specific classes of stochastic systems, approximation
algorithms exist, which do however focus more on expressing
mathematical properties of (integral) equations that one should solve,
not really on convergence of the schemes. Sometimes, strong conditions
are put on the system, so that convergence is obvious. This is for
instance the case of \cite{AB06,BA07,PHV16}.

The literature on stochastic hybrid systems is very rich, and since
there is little hope to have some decidability results, approximation
methods are very much developed. We give here some examples of works
that have been done, but this is obviously not
exhaustive. In~\cite{FHH+11}, an over-approximation method based on a
discrete abstraction is proposed for stochastic hybrid automata, but
no converging approximation scheme is provided. In~\cite{SMA16}, an
approximation (with some guarantee on the error made) is provided,
which can be used for time-bounded verification of safety
properties. Some other papers focus on discrete-time, allowing the use
of constraint-solving methods, see e.g.~\cite{FHT08}.

Continuous stochastic systems as mentioned above are hard to analyze:
first, it is difficult (and sometimes even impossible) to compute
  the exact value of the probability of some property (as simple as a
  reachability property) in such a system; and, for such complex
  systems, %
  there is no generic proofs of convergence for approximation schemes.
The key contribution of
the current paper is to identify conditions to have correct decision
procedures and approximation schemes, and to provide full proofs of
convergence and correctness.

\section{Preliminaries}
\label{sec:prelim}

In this section, we define the general model of stochastic transition
systems, which are 
Markov chains with a continuous state-space. This model corresponds to
labelled Markov processes of~\cite{Pan01} with a single action (hence
removing non-determinism). We then define several probability
measures, on infinite paths, but also on the state-space, which give
different point of views over the behaviour of the systems.  We
continue by defining regular measurable events, and end up with
defining deterministic Muller automata, and technical material for
handling properties specified by these automata.  In the interest of
space, for basics on probability and measure theory, we refer the
reader to~\cite{Pan01}.

\subsection{Stochastic transition systems}

Given $(S,\Sigma)$ a measurable space ($\Sigma$ is a $\sigma$-algebra
over $S$), we write $\Dist(S,\Sigma)$ for the set of probability
distributions over $(S,\Sigma)$. In the sequel, when the context is
clear, we will omit the $\sigma$-algebra and simply write this set as
$\Dist(S)$.

\begin{definition}
  A \emph{stochastic transition system} (STS) is a tuple $\calT =
  (S,\Sigma,\kappa)$ consisting of a measurable  space
  $(S,\Sigma)$, and $\kappa : S \times\Sigma \to [0,1]$ such that for
  every fixed $s \in S$, $\kappa(s,\cdot)$ is a probability measure
  and for each fixed $A \in \Sigma$, $\kappa(\cdot, A)$ is a
  measurable function.  Function $\kappa$ is the \emph{Markov kernel}
  of $\calT$.
\end{definition}
Note that it is sufficient to define $\kappa(s,\cdot)$ (for every $s
\in S$) over a subset which generates the $\sigma$-algebra $\Sigma$.

Observe that if $S$ is a denumerable set and $\Sigma=2^S$, then
$\calT$ is a denumerable Markov chain (DMC for short). If $S$ is
finite, the kernel $\kappa$ then coincides with the standard
probability matrix of the Markov chain. We now give two examples of
STSs.

\begin{example}[Denumerable Markov chain]\label{Example:DMCRandomWalk}
  The first example is the DMC depicted in
  Figure~\ref{Figure:RandomWalk}. We consider here
  $\calT_{\rw}=(S_\rw,\Sigma_\rw,\kappa_\rw)$ where
  \begin{itemize}
  \item $S_\rw=\IN$,
  \item $\Sigma_\rw=2^{S_\rw}$,
  \item for each $i\ge 1$, $\kappa_\rw(i, \lbrace i+1\rbrace)=p$ and
    $\kappa_\rw(i, \lbrace i-1\rbrace)=1-p$ with $p\in\intervaloo{0,1}$,
    and
  \item $\kappa_\rw(0,\lbrace 1\rbrace)=1$.
  \end{itemize}
  This represents a random walk~--~hence the index $\rw$~--~over the
  natural numbers.
\begin{figure}[!h]
  \centering
  \begin{tikzpicture}
    \tikzstyle{ptt}=[scale=1]
    \tikzstyle{loc}=[ptt,draw,circle,minimum size =1cm, thick];
    \tikzstyle{inv}=[ptt,circle,minimum size =1cm, thick];
    \tikzstyle{fleche}=[->,>=stealth', thick, rounded corners=1pt];
    \node[loc] (lzero) at (0,0) {$0$};
    \node[loc] (lun) at (2.2,0) {$1$};
    \node[loc] (ldeux) at (4.4,0) {$2$};
    \node[inv] (lint) at (6.6, 0) {$\cdots$};
    \draw[fleche] (lzero) to[bend left=30] node[ptt,midway, above] {$1$} (lun);
    \draw[fleche] (lun) to[bend left=30] node[ptt,midway, below] {$1-p$} (lzero);
    \draw[fleche] (lun) to[bend left=30] node[ptt,midway, above] {$p$} (ldeux);
    \draw[fleche] (ldeux) to[bend left=30] node[ptt,midway, below] {$1-p$} (lun);
    \draw[fleche] (ldeux) to[bend left=30] node[ptt,midway, above] {$p$} (lint);
    \draw[fleche] (lint) to[bend left=30] node[ptt,midway, below] {$1-p$} (ldeux);
  \end{tikzpicture} \caption{Random walk over $\IN$.}
  \label{Figure:RandomWalk}
\end{figure}
\end{example}

In the sequel, given a DMC $\calT=(S,\Sigma,\kappa)$ and two states
$s,s'\in S$, we will often write $\kappa(s,s')$ instead of
$\kappa(s,\lbrace s'\rbrace)$.

\begin{example}[Continuous-time Markov chain]
  \label{Example:Continuous} 
  We now give a continuous variant of the previous random walk which
  models a simple queueing system. Precisely, we consider a queueing
  system with a single queue, a parameter $\lambda$ for arrivals and
  $\nu$ for serving times. Each state $i\in\IN$ is equipped with a
  non-negative real number that corresponds to the time that has
  elapsed since the beginning. Formally, we consider
  $\calT_\qs=(S_\qs, \Sigma_\qs, \kappa_\qs)$ with
  $S_\qs=\IN\times\Rpos$. We equip $S_\qs$ with the $\sigma$-algebra
  generated by $2^{\IN}\times\calB(\Rpos)$ where $\calB(\Rpos)$ is the
  Borel $\sigma$-algebra on $\Rpos$. Then intuitively, $\kappa_\qs$
  describes how the length of the queue evolves with time. Formally,
  for each $t, d\in \Rpos$, 
  \[
  \kappa_\qs((0,t),\{1\} \times \intervalcc{0,t+d}) =
  \kappa_\qs((0,t),\{1\} \times \intervalcc{t,t+d})
  =\int_0^{d}\lambda e^{-\lambda x}\ud x
  \]
  and for every $i \ge 1$,
  \[
  \begin{array}{c}
    {\displaystyle  \kappa_\qs((i,t),\{i+1\}
      \times \intervalcc{0,t+d}) =
      \kappa_\qs((i,t),\{i+1\} \times \intervalcc{t,t+d}) =
      \int_0^{d}\lambda\,
      e^{-(\lambda+\nu) x}\ud x} \\
    {\displaystyle  \kappa_\qs((i,t),\{i-1\} \times \intervalcc{0,t+d}) =
      \kappa_\qs((i,t),\{i-1\} \times \intervalcc{t,t+d}) =
      \int_0^{d} \nu\,
      e^{-(\lambda+\nu) x}\ud x}
  \end{array}
  \]
\end{example}
There will be more examples of STSs with a continuous set of states in
Section~\ref{sec:appli}.

\medskip In the sequel, we fix an STS $\calT=(S,\Sigma,\kappa)$. We
will give two semantical views on the behaviour of $\calT$: the first
one is operational, in the sense that $\calT$ generates executions,
with a measure over these executions; the second one observes how the
state-space evolves over time.  The first point-of-view is the
standard semantics of probabilistic systems and is widely used in the
model-checking community, where the temporal aspects are
important. From a state, a probabilistic transition is performed
according to a fixed distribution, and the system resumes from one of
the successor states. Among others, \emph{e.g.} \cite{BHHK03}, uses
this semantics for the standard model of continuous-time Markov
chains. The second point-of-view bloomed more recently. It proposes to
view probabilistic systems as transformers of probability
distributions. Compared to the previous point-of-view, here one is
interested not in states the system can be in, but rather in how the
probability mass evolves along steps. This semantics was motivated by
the ability to express different properties than the previous
one~\cite{BRS-csl02}. It has been considered for both discrete-time
Markov chains~\cite{KA-icfem04,AAGT-lics12} and continuous-time
models, \emph{e.g.} those induced by stochastic Petri
nets~\cite{HPRV-peva12}.  These two point-of-views are two sides of
the same coin, and we will use both in the following, though we are
ultimately interested in properties related to the operational
semantics.

\subsection{A $\sigma$-algebra for measuring sets of infinite paths}
\label{section:PrelimMeasure}

The objective is now to interpret $\calT$ in an operational manner,
and to define a probability measure over the set of infinite paths of
$\calT$.  We follow the lines of~\cite{DP03}. A \emph{finite
  (resp. infinite) path} of $\calT$ is a finite (resp. infinite)
sequence of states.  We write $\Paths(\calT)$ for the set of infinite
paths of $\calT$. In order to get a probability measure over
$\Paths(\calT)$, we need to equip this set with a $\sigma$-algebra. We
therefore define for each finite sequence of measurable sets
$(A_i)_{0 \le i \le n} \in \Sigma^{n+1}$ the following set of infinite
paths:
\[ \Cyl(A_0,A_1,\dots,A_n) = \{\rho = s_0 s_1 \dots s_n\dots \in
\Paths(\calT) \mid \forall 0 \le i \le n,\ s_i \in A_i\} \enspace. \]
This set is called a \emph{cylinder}. We then equip $\Paths(\calT)$
with the $\sigma$-algebra generated by the cylinders. We denote it by
$\calF_{\calT}$.

Let $\mu$ be a(n initial) probability measure over $\Sigma$, that is,
$\mu \in \Dist(S)$.  We define a probability measure
$\Prob^{\calT}_\mu$ as follows. First we inductively define a
probability measure over the cylinders. For every finite sequence of
measurable subsets $(A_i)_{0 \le i \le n} \in \Sigma^{n+1}$, we set:
\[
\Prob^{\calT}_\mu(\Cyl(A_0,A_1,\dots,A_n)) = \int_{s_0 \in A_0}
\Prob^{\calT}_{\kappa(s_0,\cdot)}(\Cyl(A_1,\ldots,A_n))\ud\mu(s_0)
\enspace,
\]
and we initialize with $\Prob_{\mu}^{\calT}(\Cyl(A_0))=\mu(A_0)$.
From now on, we will use the classical notation $\mu(\ud
s_0)=\ud\mu(s_0)$. It should be noted that the value
$\Prob^{\calT}_\mu(\Cyl(A_0,A_1,\dots,A_n))$ is the result of $n$
successive integrals and can be expressed as follows:
\begin{multline*}
  \Prob^{\calT}_\mu(\Cyl(A_0,A_1,\dots,A_n)) = \\
  \int_{s_0 \in A_0} \int_{s_1 \in A_1} \dots \int_{s_{n-1} \in
    A_{n-1}} \kappa(s_0,\ud s_1) \cdot \kappa(s_1, \ud s_2) \cdots
  \kappa(s_{n-2},\ud s_{n-1}) \cdot \kappa(s_{n-1},A_n) \cdot \mu(\ud
  s_0).
\end{multline*}
Finally, using the classical Caratheodory's extension theorem,
$\Prob_{\mu}^{\calT}$ can be extended in a unique way to the
$\sigma$-algebra $\calF_{\calT}$.

\begin{restatable}{lemma}{probpaths}\label{lem:probpaths}
  $\Prob_\mu^{\calT}$ defines a probability measure over 
  $(\Paths(\calT), \calF_{\calT})$.
\end{restatable}

The proof of Lemma~\ref{lem:probpaths} is classical and we omit it
here. The interested reader may \emph{e.g.} refer to the proof of
\cite[Proposition~3.2]{BBB+14}, which can easily be adapted to our
context. 

\subsection{STSs as transformers of probability measures}
One can also interpret the dynamics of $\calT$ as a transformer of
probability measures over $(S,\Sigma)$. For $\mu$ a probability
measure over $\Sigma$, its transformation through $\calT$ can be
defined as the measure $\Omega_\calT(\mu)$ defined for every $A \in
\Sigma$ by:
\[
\Omega_\calT(\mu)(A) = \int_{s_0 \in S} \kappa(s_0,A) \cdot \mu(\ud
s_0) \enspace.
\]
It can be shown that $\Omega_\calT(\mu)$ is also a probability measure
over $(S,\Sigma)$.

This interpretation offers a dual view on the STS $\calT$. Indeed,
roughly speaking, $\Omega_\calT(\mu)(A)$ is the probability of being
in $A$ after one step, when $\mu$ is the initial distribution on
$\calT$. Given a distribution $\mu\in\Dist(S)$ and given $A\in\Sigma$
such that $\mu(A)>0$, we write $\mu_A$ for the conditional probability
of $\mu$ given $A$, that is for each $B\in\Sigma$,
$\mu_A(B)=\frac{\mu(A\cap B)}{\mu(A)}$. It should be observed that
$\mu_A\in\Dist(S)$. There is a strong relation between the transformer
$\Omega_{\calT}(\mu)$ and the probability measure $\Prob^\calT_\mu$
over paths defined previously, which we formalize below:
\begin{restatable}{lemma}{lemmaintegration}
\label{lemma:integration}
  Let $\mu\in\Dist(S)$ be an initial distribution and let $(A_i)_{0
    \le i \le n}$ be a sequence of measurable sets. Write $\nu_0 =
  \mu_{A_0}$, the conditional probability of $\mu$ given $A_0$, and for every $1 \le
  j \le n-1$, write $\nu_j = (\Omega_\calT(\nu_{j-1}))_{A_j}$. Then,
  for every $0 \le j \le n$:
\begin{multline*}
\Prob^{\calT}_\mu(\Cyl(A_0,A_1,\dots,A_n)) = \\
\mu(A_0)\cdot\prod_{i=1}^{j}(\Omega_{\calT}(\nu_{i-1}))(A_i)\cdot
\Prob^\calT_{\Omega_\calT(\nu_{j})}(\Cyl(A_{j+1},\dots,A_n))
\enspace. 
\end{multline*}
\end{restatable}
\noindent The proof of this result is postponed to the technical
appendix (page~\pageref{app:lemma_integration}).

From this result, we can express the probability to reach $A$ in $n$
steps from the initial distribution $\mu$:
\[
(\Omega^{(n)}_{\calT}(\mu))(A)=\Prob_{\mu}^{\calT}(\Cyl(\overbrace{S,\ldots,S}^{n\text{
    times}}, A)) \enspace.
\]

This alternative view of stochastic processes as transformers of
  probability measures is heavily used by Paolieri \emph{et al.} in
  their time-bounded analysis of stochastic Petri nets~\cite{PHV16}:
  the evolution of the probability distributions is tracked through
  stochastic state classes. We advocate that the two views
  (behavioural and probability transformers) need to be used at the
  same time. One of the first results we establish (see
  Lemma~\ref{lemma:attractorGF} later) is a witness of their
  interplay. Also, the measure transformer view will prove quite
  useful when it comes to abstraction in
  Section~\ref{sec:abstractions}. Beyond that, it has been observed by
  Kwon \emph{et al.}  that the classes of properties one can express
  in both views are incomparable~\cite{KA-icfem04, KVAK10}, and
  depending on the application, one or the other can be more
  appropriate.

\subsection{Basic properties of paths in STSs}\label{subsec:Events}
To define properties on the STS $\calT$, we use \LTL-like notations,
that will be interpreted as measurable subsets of $\Paths(\calT)$.
Let $\mathcal{L}_{S,\Sigma}$ be the set of formulas 
defined by the following grammar:
\[
\varphi ::= B \mid \varphi_1 \U[\bowtie k] \varphi_2
\mid \varphi_1 \vee \varphi_2 \mid \varphi_1 \wedge \varphi_2 \mid
\neg \varphi \enspace,
\]
where $B \in \Sigma$, $\mathord{\bowtie} \in \{\geq,\leq,=\}$ is a
comparison operator and $k \in \IN$ is a natural number. Given
$\rho=(s_n)_{n\ge 0}$ we write $\rho_{\ge i}=(s_n)_{n\ge
  i}\in\Paths(\calT)$ for each $i\ge 0$. Then the satisfaction
relation of paths formulas is given as follows:
\begin{alignat}{7}
  {} & \rho \models B \ & {} & \Leftrightarrow \ & {} & s_0\in B\notag\\
  {} & \rho \models \varphi_1 \U[\bowtie k] \varphi_2 \ & {} &
  \Leftrightarrow \ & {} & \exists i\ge 0, \ i \bowtie k,\ \text{s.t.}\
  \rho_{\ge i}\models \varphi_2 \text{ and }
  \forall 0\leq j <i, \ \rho_{\ge j}\models \varphi_1\notag\\
  {} & \rho\models\varphi_1\vee\varphi_2 \ & {} & \Leftrightarrow \ & {} & \rho\models\varphi_1 \text{ or } \rho\models\varphi_2\notag\\
  {} & \rho\models\varphi_1\wedge\varphi_2 \ & {} & \Leftrightarrow \ & {} & \rho\models\varphi_1 \text{ and } \rho\models\varphi_2\notag\\
  {} & \rho\models\neg\varphi \ & {} & \Leftrightarrow \ & {} &
  \rho\nvDash\varphi. \notag
\end{alignat}
We write $\ev{\calT}{\varphi}$ for the set of infinite paths $\rho$ in
${\cal T}$ such that $\rho \models \varphi$. It is standard to show
that the event $\ev{\calT}{\varphi}$ is a measurable subset of
$(\Paths(\calT),\calF_{\calT})$ (see \emph{e.g.}~\cite{vardi85}). In
particular, for every initial probability measure $\mu$,
$\Prob_\mu^\calT(\ev{\calT}{\varphi})$ is well-defined. In the sequel,
for simplicity, we often write $\Prob_\mu^\calT(\varphi)$ instead of
$\Prob_\mu^\calT(\ev{\calT}{\varphi})$.

We will also use classical notations like $\top = S$; $\bot =
\emptyset$; $\varphi_1 \U \varphi_2 = \varphi_1 \U[\geq 0]
\varphi_2$; $\F \varphi = \top \U \varphi$; $\F[\bowtie p] \varphi =
\top \U[\bowtie p] \varphi$; $\G \varphi = \neg \F (\neg \varphi)$;
$\G[\bowtie p] \varphi = \neg \F[\bowtie p] (\neg \varphi)$.
\begin{example}
  We illustrate some properties on STSs $\calT_\rw$ and $\calT_\qs$ of
  Examples~\ref{Example:DMCRandomWalk}
  and~\ref{Example:Continuous}. Consider first $\calT_\rw$ and
  $B=\lbrace 0\rbrace$. The event $\ev{\calT_\rw}{\F B}$ corresponds
  to the set of runs that eventually reach state $0$, \emph{i.e.}
  $\ev{\calT_\rw}{\F B}=\lbrace \rho=(s_n)_{n\ge
    0}\in\Paths(\calT_\rw)\mid \exists n\ge 0,\ s_n=0\rbrace$, while
  the event $\ev{\calT_\rw}{\G\F B}$ corresponds to the set of runs
  that infinitely often visit $0$, \textit{i.e.} $\ev{\calT_\rw}{\G\F
    B}=\lbrace \rho=(s_n)_{n\ge 0}\in\Paths(\calT_\rw)\mid \forall
  n\ge 0, \ \exists k\ge n,\ s_k=0\rbrace$. Now in $\calT_\qs$, we could
  consider the set of states $B'=\lbrace 0\rbrace \times
  \intervalcc{0, T}$ with $T>0$. Then the event $\ev{\calT_\qs}{\F
    B'}$ corresponds to the set of runs that eventually reach $0$
  within $T$ time units, \textit{i.e.} $\ev{\calT_\qs}{\F B'}=\lbrace
  \rho=((s_n,t_n))_{n\ge 0}\in\Paths(\calT_\qs)\mid \exists n\ge 0,\
  (s_n=0 \wedge t_n\leq T)\rbrace$.
\end{example}

\subsection{Labelled STSs and their properties}
\label{subsec:LSTS}
To ease the expression of rich properties over STSs, we extend the
model with a labelling with atomic propositions.
\begin{definition}
  A \emph{labelled STS} (LSTS for short) is a tuple $\calT =
  (S,\Sigma,\kappa,\AP,\calL)$, where $(S,\Sigma,\kappa)$ is an STS,
  $\AP$ is a finite set of atomic propositions, and $\calL:
  S\rightarrow 2^{\AP}$ is a measurable labelling function.
\end{definition}
Measures and other notions are extended in a straightforward way from
STSs to LSTSs. We fix a finite set $\AP$ of atomic propositions and an
LSTS $\calT = (S,\Sigma,\kappa,\AP,\calL)$.

A property over $\AP$ is a subset $P$ of
$\left(2^\AP\right)^\omega$. An infinite path $\rho = s_0 s_1 \ldots$
of $\calT$ satisfies the property $P$ whenever $\mathcal{L}(s_0)
\mathcal{L}(s_1) \mathcal{L}(s_2) \ldots \in P$, written $\rho \models
P$.
$\omega$-regularity is a standard notion in computer science to
characterise simple sets of infinite behaviours, and typical
$\omega$-regular properties are B\"uchi and Muller properties. In
order to express such properties, we introduce a new notation for the
set of atomic propositions that are true infinitely often along a
sequence of labels: for $\varpi = u_0 u_1 u_2 \ldots \in
\left(2^\AP\right)^\omega$, we define $\Inf(\varpi)= \lbrace
a\in\AP\mid |\lbrace j\in \IN\mid a\in u_j\rbrace| = \infty\rbrace$.
We extend this notation to paths in a natural way: if $\rho = s_0 s_1
s_2 \ldots \in S^\omega$, writing $\varpi = \calL(s_0) \calL(s_1)
\calL(s_2) \ldots$, we define (with a slight abuse of notation)
$\Inf(\rho) = \Inf(\varpi)$.
A \emph{B\"uchi property} $P$
over $\AP$ can be specified by a subset of atomic propositions $F
\subseteq \AP$ as $P = \{\varpi\in \left(2^\AP\right)^\omega \mid
\Inf(\varpi)\cap F \neq \emptyset \}$.  A \emph{Muller property} over
$\AP$ is a property $P$ such that there exists $\mathcal{F} \subseteq
2^\AP$ with $P = \{\varpi\in \left(2^\AP\right)^\omega \mid
\Inf(\varpi)\in\calF\}$.
\begin{rk}\label{remark:DMAMes}
  It should be noted that the set of infinite paths satisfying B\"uchi
  or Muller properties can be expressed using events as in
  Section~\ref{subsec:Events}. Indeed, for $F\subseteq \AP$ 
  we write $2^{\AP}_F=\lbrace u\in 2^{\AP}\mid u\cap F\neq
  \emptyset\rbrace$ and given $a\in\AP$, $2^{\AP}_a=\lbrace u\in
  2^{\AP}\mid a\in u\rbrace$. Then,
  \begin{itemize}
  \item the set of paths satisfying the B\"uchi property with
    acceptance condition $F$ is 
    \[
    \mathsf{Ev}_{\calT}\Big(\G\F \big(\bigvee_{u\in 2^{\AP}_F}
    \calL^{-1}(u)\big)\Big) \enspace; \]
  \item the set of paths satisfying the Muller property with acceptance
    condition $\calF$ is
    \[
    \mathsf{Ev}_{\calT}\Big(\bigvee_{F\in\calF} \Big(\bigwedge_{a\in
      F} \big(\G\F \bigvee_{u\in 2^{\AP}_a} \calL^{-1}(u)\big) \wedge
    \bigwedge_{a\notin F} \bigwedge_{u\in 2^{\AP}_a} \F\G
    \big(\calL^{-1}(u)\big)^c\Big)\Big) \enspace. \]
\end{itemize}
\end{rk}

It is well known that automata equipped with B\"uchi or Muller
acceptance conditions capture all $\omega$-regular properties, and
this also holds for deterministic Muller automata.

\begin{definition}
A \emph{deterministic Muller automaton} (DMA) over $\AP$ is a tuple $\calM =
(Q,q_0,E,\mathcal{F})$ where:
\begin{itemize}
\item $Q$ is a finite set of locations, and $q_0\in Q$ is the initial location;
\item $E \subseteq Q \times 2^\AP \times Q$ is a finite set of edges;
\item $\mathcal{F}$ is a Muller condition over $Q$;
\end{itemize}
and such that
\begin{itemize}
\item $\calM$ is deterministic: for all pair of edges $(q,u,q_1)$ and
  $(q,u,q_2)$ in $E$, $q_1=q_2$;
\item $\calM$ is complete: for every $q \in Q$, for every $u \in 2^{\AP}$,
  there exists $(q,u,q') \in E$. 
\end{itemize}
\end{definition}
A deterministic Muller automaton $\calM$ naturally gives rise to a
property $P_{\calM}$ defined by the language (over $2^{\AP}$) accepted
by $\calM$. Its semantics over infinite paths of $\calT$ is derived
from that of property $P_\calM$: if $\rho \in \Paths(\calT)$, we write
$\rho \models \calM$ whenever $\rho \models P_\calM$.  Expanding Remark~\ref{remark:DMAMes}, one derives the standard fact that the set $\calT[\calM]
\stackrel{\text{def}}{=} \{\rho \in \Paths(\calT) \mid \rho \models
\calM\}$ is measurable, and we write $\Prob_\mu^\calT(\calM)$ for
$\Prob_\mu^\calT(\calT[\calM])$.

\begin{rk}
  It is well known (see~\cite{VW94} and \cite[Chapter~3]{lncs2500})
  that for any \LTL formula $\varphi$ (the syntax given in the
  previous subsection, where we replace sets $B$ by inverse images by
  $\calL$ of atomic propositions from $\AP$), there is a deterministic
  Muller automaton $\calM_\varphi$ that characterises $\varphi$, that
  is: for every run $\rho$, $\rho \models \varphi$ iff $\rho \models
  \calM_\varphi$.
\end{rk}

\paragraph*{Product STS}
To measure the probability of properties specified by a DMA $\calM =
(Q,q_0,E,\calF)$, it is standard to build a new STS consisting of the
product of $\calT$ with $\calM$~\cite{BK08}. To this aim, we consider
the discrete $\sigma$-algebra $2^Q$ on the finite set of locations $Q$
of $\calM$. The product $S\times Q$ can then be equipped
  with the product $\sigma$-algebra $\Sigma_p$ defined as the smallest
  $\sigma$-algebra generated by the rectangles, where a rectangle is set
  of the form $X \times Q'$ where $X\in \Sigma$ and $Q' \subseteq
  Q$; i.e. $X \times Q' = \{(s,q)~|~ s \in X, \, q \in Q'\}$.  Given $Y$ an element of $\Sigma_p$, for all $q \in Q$, one
  can define the set $\pi_q(Y) \stackrel{\text{def}}{=} \{s \in S \mid
  (s,q) \in Y\}$. We therefore write (abusively) $Y = \bigcup_{q \in
    Q} \pi_q(Y) \times \{q\}$. Then, one can show that $\Sigma_p$
  coincides with $\Sigma'$, the set of all subsets of $S \times Q$ of
  the form $\bigcup_{q\in Q} C_q \times \{q\}$, where $C_q \in \Sigma$
  for every $q \in Q$ (see the proof in the appendix,
  page~\pageref{app:product-sigma-algebra}).
\label{product-sigma-algebra}

Note that in the sequel, we will sometimes write $(C_q, q)$ instead of
$C_q\times \lbrace q\rbrace$.

We now define the product of $\calT$ with $\calM$ as follows.
\begin{definition}
  Given $\calT = (S,\Sigma,\kappa,\AP,\calL)$ an LSTS and $\calM =
  (Q,q_0,E,\calF)$ a DMA over $\AP$, we define the product of $\calT$
  with $\calM$ as the LSTS $\calT \ltimes
  \calM=(S',\Sigma',\kappa',\AP',\calL')$ such that:
  \begin{itemize}
  \item $S' = S \times Q$;
  \item $\Sigma'$ is the product $\sigma$-algebra $\Sigma\times 2^Q$;
  \item $\kappa'((s,q),(A,q')) = \begin{cases}\kappa(s,A) & \textrm{
        if } (q,\calL(s),q') \in E, \textrm{
        and}\\ 
      0 & \textrm{ otherwise;\footnotemark}\end{cases}$
    \item $\AP' = Q$;
    \item $\calL'(s,q) = q$.
  \end{itemize}  
\end{definition}
\footnotetext{{Note that the above definition of $\kappa'$ extends
    naturally to all elements of the $\sigma$-algebra $\Sigma'$:
          for each pair $(q, u)$ with $q\in Q$ and $u\in 2^\AP$, there is a
    unique $q'\in Q$ such that $(q,u,q')\in E$. Fix $(s,q)\in S\times
    Q$, write $q'$ for the unique location such that $(q, \calL(s),
    q')\in E$. Then for each $A=\bigcup_{q\in Q} C_q\times \lbrace q
    \rbrace$,
    $\kappa'((s,q),A)=\kappa'((s,q),(C_{q'},q'))=\kappa(s,C_{q'})$.
}}

\begin{example}
  We consider again $\calT_{\rw}$ the random walk over $\IN$ of
  Example~\ref{Example:DMCRandomWalk}. We assume that it is equipped
  with the simple set of atomic propositions $\AP=\lbrace a\rbrace$
  and we assume that each state of the STS is labelled with $a$. Let
  $\calM$ be the DMA depicted on the left-hand side of
  Figure~\ref{Figure:MullerAutomaton}. The winning condition is given
  by $\calF=\lbrace\lbrace q_1,q_2\rbrace\rbrace$. The product
  $\calT_\rw\ltimes\calM$ is then depicted on the right-hand side of
  Figure~\ref{Figure:MullerAutomaton}. It should be noted that we
  there assume that the system starts at $(0,q_0)$ however, there
  should be similar chains starting in $(i,q_0)$ for each $i\ge
  1$. Note also that we did not specify the labels on the states:
  according to the definition each state is labelled with its current
  position in $\calM$.
\begin{figure}[!h]
  \centering
  \begin{tikzpicture}
    \tikzstyle{ptt}=[scale=1]
    \tikzstyle{loc}=[ptt,draw,circle,minimum size =0.8cm];
    \tikzstyle{locRec}=[ptt,draw,rectangle,minimum height =0.7cm, minimum width =0.9cm, rounded corners = 3pt];
    \tikzstyle{locRecInv}=[ptt,rectangle,minimum height =0.7cm, minimum width =0.9cm, rounded corners = 3pt];
    \tikzstyle{inv}=[ptt,circle,minimum size =1cm];
    \tikzstyle{fleche}=[->,>=stealth', rounded corners=1pt];
    \node[loc] (qzero) at (0,0) {$q_0$};
    \node[loc] (qun) at (1.5,0) {$q_1$};
    \node[loc] (qdeux) at (3,0) {$q_2$};
    \draw[fleche] (qzero) -- (qun) node[midway, above] {$\{a\}$};
    \draw[fleche] (qun) to[bend left=30] node[midway, above] {$\{a\}$} (qdeux);
    \draw[fleche] (qdeux) to[bend left=30] node[midway, below] {$\{a\}$} (qun);
    \draw[fleche] (-1,0) -- (qzero);
    
    \node[locRec] (lzero) at (4.5,1) {$(0,q_0)$};
    \node[locRec] (lpzero) at (4.5, -1) {$(0, q_2)$};
    \node[locRec] (lun) at (6.5, 0) {$(1, q_1)$};
    \node[locRec] (ldeux) at (8.5, 0) {$(2, q_2)$};
    \node[locRec] (ltrois) at (10.5, 0) {$(3, q_1)$};
    \node[locRecInv] (linv) at (12.5, 0) {$\cdots$};
    
    \draw[fleche] (lzero) -- (lun) node[midway, above] {$1$};
    \draw[fleche] (lun) to[bend left=30] node[pos=.2, below=5pt] {$1-p$} (lpzero);
    \draw[fleche] (lun) to[bend left=30] node[midway, above] {$p$} (ldeux);
    \draw[fleche] (lpzero) to[bend left=30] node[midway, above] {$1$} (lun);
    \draw[fleche] (ldeux) to[bend left=30] node[midway, below] {$1-p$} (lun);
    \draw[fleche] (ldeux) to[bend left=30] node[midway, above] {$p$} (ltrois);
    \draw[fleche] (ltrois) to[bend left=30] node[midway, below]
    {$1-p$} (ldeux);
    \draw[fleche] (ltrois) to[bend left=30] node[midway, above]
    {$p$} (linv);
    \draw[fleche] (linv) to[bend left=30] node[midway, below] {$1-p$} (ltrois);
  \end{tikzpicture} \caption{A Muller automaton $\calM$ and the product $\calT_\rw\ltimes\calM$.}
  \label{Figure:MullerAutomaton}
\end{figure}
\end{example}

We define on $\calT \ltimes \calM$ a Muller condition which is
inherited from the one of $\calM$ via the new labelling function
$\calL'$: a run $\rho$ satisfies the Muller condition $\calF'$
whenever $\calL'(\rho)$ satisfies the Muller condition $\calF$. We
thus later simply use $\calF$ instead of $\calF'$.

We now explain how initial distributions for $\calT$ are lifted to the
product $\calT \ltimes \calM$. The idea is simple: the
$\calT$-component is inherited from $\calT$, and the $\calM$-component
is a Dirac\footnote{We recall that given a state $s$, the \emph{Dirac
    distribution over $s$} is defined by $\delta_s(A)=1$ if $s\in A$
  and $\delta_s(A)=0$ otherwise, for every $A\in\Sigma$.} distribution
over $q_0$, the initial state of $\calM$.  In other words, when an
initial distribution $\mu \in \Dist(S)$ is fixed for $\calT$, the
initial distribution of $\calT \ltimes \calM$ will be $\mu \times
\delta_{q_0}$.  We show that this allows to properly express the
probability of a property described by a DMA, with the following
correspondence (whose proof is given in the appendix,
page~\pageref{app:produit-technique}).
\begin{restatable}{proposition}{produit}
  Let $\mu \in \Dist(S)$ be an initial distribution for $\calT$, and
  $\calM = (Q,q_0,E,\calF)$ be a DMA. Then:
  \[
  \Prob_\mu^\calT(\calT[\calM]) = \Prob_{\mu \times
    \delta_{q_0}}^{\calT \ltimes \calM}(\{\rho \in \Paths(\calT
  \ltimes \calM) \mid \rho \models \calF\})\enspace.
  \]
\end{restatable}

\section{Nice properties of STSs}
\label{sec:prop}
In~\cite{ABM07}, Abdulla \emph{et al.} introduced the elegant concept
of decisive Markov chain. Intuitively, decisiveness allows one to lift
the good properties of finite Markov chains to infinite (but
denumerable) Markov chains. We explain here how to extend and refine
this concept and some related concepts to general STSs, and we
establish relationships between these properties.

\subsection{Several decisiveness notions}
Decisiveness has been defined in~\cite{ABM07} as a desirable property
of denumerable Markov chains, since it implies that they behave
essentially like finite Markov chains.

For $B \in \Sigma$ a measurable set, we define its \emph{avoid-set}
$\widetilde{B} = \{s \in S \mid \Prob_{\delta_s}^{\calT}(\F B) =
0\}$. It corresponds to the set of states from which the system will
always avoid the set $B$ with probability $1$ (or equivalently, reach
$B$ with probability $0$).  The set $\widetilde{B}$ enjoys the
following properties, that obviously hold also in the context of
denumerable Markov chains, but require proofs in our general context
of STSs (those proofs are postponed to the technical appendix,
page~\pageref{app-BtildeMes}).
\begin{restatable}{lemma}{BtildeMes}
  \label{lem:Btilde}\label{lemma:BTildeComplementaire}
  \label{lemma:BTildeEquivFGF}
  Given $B\in\Sigma$, it holds that: 
  \begin{enumerate}
  \item $\widetilde{B}$ belongs to the $\sigma$-algebra $\Sigma$;
  \item for every $\mu\in\Dist(\Btilde)$, $\Prob_{\mu}^{\calT}(\F
    B)=0$;
  \item for every $\mu \in \Dist(S)$, if $\mu((\widetilde{B})^c)>0$,
    then $\Prob_{\mu}^{\calT}(\F B)>0$;
  \item for every $\mu\in\Dist(S)$, $\Prob_{\mu}^{\calT}(\F
    \Btilde)=\Prob_\mu^\calT(\F \G \Btilde) =
    \Prob_{\mu}^{\calT}(\G\F \Btilde)$;
  \item for every $\mu\in\Dist(S)$, $\Prob_{\mu}^\calT(\F B\vee \F
    \Btilde)= \Prob_{\mu}^{\calT}(\F B \vee (\neg B \U \Btilde))$.
  \end{enumerate}
\end{restatable}

\smallskip
\noindent Let us comment on the third and fourth properties stated in
this lemma. The third item indicates that if we start from outside
$\widetilde{B}$, then we will always have a positive probability to
hit $B$. The fourth property says that $\widetilde{B}$ is some kind of
sink: once we hit $\widetilde{B}$, we cannot escape it.
The other properties are rather straightforward to understand (even
though proving the first property requires some technical
developments).

\medskip We are now ready to define different decisiveness
concepts. Two stem from~\cite{ABM07} (though no initial distribution
was fixed) while the third one was identified in~\cite{BBBC16} as a
useful alternative.
\begin{definition}
  Let $\mu$ be an initial probability distribution ($\mu \in
  \Dist(S)$). Then:
  \begin{itemize}
  \item $\calT$ is \emph{decisive w.r.t. $B$ from $\mu$} whenever
    $\Prob^{\calT}_\mu(\F B \vee \F \Btilde) = 1$; we then write that
    $\calT$ is $\D(\mu,B)$.
  \item $\calT$ is \emph{strongly decisive w.r.t. $B$ from $\mu$}
    whenever $\Prob^{\calT}_\mu(\G \F B \vee \F \Btilde) = 1$; we then
    write that $\calT$ is $\SD(\mu,B)$.
  \item $\calT$ is \emph{persistently decisive w.r.t. $B$ from
      $\mu$} whenever for every $p \ge 0$, $\Prob^{\calT}_\mu(\F[\ge
    p] B \vee \F[\ge p] \Btilde) = 1$; we then write that $\calT$ is
    $\PD(\mu,B)$.
  \end{itemize}
  Furthermore: $\calT$ is (\emph{strongly}, \emph{persistently}) \emph{decisive
  w.r.t. $B$} whenever it is (strongly, persistently) decisive
  w.r.t. $B$ from every initial distribution $\mu$; we then write that
  $\calT$ is $\D(B)$ (resp. $\SD(B)$, $\PD(B)$).  Also, given
  $\calB\subseteq\Sigma$, $\calT$ is (\emph{strongly},
  \emph{persistently}) \emph{decisive w.r.t. $\calB$ from $\mu$} if it is $\D(\mu,
  B)$ ($\SD(\mu, B)$, $\PD(\mu,B)$) for each $B\in\calB$. We write
  $\calT$ is $\D(\mu,\calB)$ ($\SD(\mu,\calB)$, $\PD(\mu,\calB)$). Similarly $\calT$ is (\emph{strongly}, \emph{persistently}) \emph{decisive
  w.r.t. $\calB$} if it is $\D(B)$ ($\SD(B)$, $\PD(B)$) for each
  $B\in\calB$. We write $\calT$ is $\D(\calB)$ ($\SD(\calB)$,
  $\PD(\calB)$). 
\end{definition}
Intuitively, the (simple) decisiveness property says that,
almost-surely, either $B$ will eventually be visited, or states from
which $B$ can no more be reached will eventually be visited. It
denotes a dichotomy between the behaviours of the STS $\calT$: there
are those behaviours that visit $B$, and those that do not visit $B$,
but then visit $\widetilde{B}$; other behaviours have probability $0$
to occur. Strong decisiveness imposes a similar dichotomy, but between
behaviours that visit $B$ infinitely often and behaviours that visit
$\widetilde{B}$. Persistent decisiveness refines simple decisiveness,
but by looking at an arbitrary horizon. It can also be seen as being
decisive from $\Omega^{(n)}_{\calT}(\mu)$ for every $n\ge 0$.

Note that if $\calT$ is finite, then it is decisive, strongly decisive
and persistently decisive.  Let us now illustrate the subtleties of
the various decisiveness notions on examples.

\begin{example}\label{example:BtildeDeci}
  Let us consider again the STS $\calT_\rw$ of
  Example~\ref{Example:DMCRandomWalk}, representing a discrete-time
  random walk, and assume $p>1/2$. As a classical result, the random
  walk is then diverging, see \emph{e.g.}~\cite{Bremaud}. Since the
  chain is strongly connected, for each $B\subseteq \IN$,
  $\Btilde=\emptyset$.  Let us first assume that the initial
  distribution is $\mu=\delta_0$, the Dirac distribution over state
  $0$. Then it can be shown that for each set of states $B$,
  $\Prob_{\mu}^{\calT_\rw}(\F B)=1$ and thus, $\calT_\rw$ is
  $\D(\mu, B)$.

  Assuming now that the initial distribution is $\mu'=\delta_1$, if
  $B=\lbrace 0\rbrace$, then
  $\Prob_{\mu'}^{\calT_\rw}(\F \lbrace 0\rbrace)<1$. But since
  $\Btilde=\emptyset$, we derive that $\calT_\rw$ is not
  $\D(\mu', B)$.

  Consider now for each $i\ge 0$, $B_i=\lbrace i \rbrace$. Since
  $p>1/2$, classical results on random walks imply that for each $i$,
  $\Prob_{\mu}^{\calT_\rw}(\G\F B_i)=0$. And since
  $\Btilde_i=\emptyset$, we obtain that $\calT_\rw$ is not $\SD(\mu,
  B_i)$.

  Consider now the STS $\calT_\qs$ of
  Example~\ref{Example:Continuous}. Assume that $\lambda>\nu$ and that
  $\mu=\delta_{(0,0)}$ and fix some $T>0$. We consider $B_1=\lbrace
  1\rbrace\times \intervalcc{0,T}$. Then one can compute
  $\Btilde=\IN\times\intervaloo{T,\infty}$.  Note that here, as time
  almost-surely always progresses, $\Prob_{\mu}^{\calT_\qs}(\F
  \Btilde)=1$. It thus follows that $\calT_\qs$ is $\D(\mu, B)$ and
  $\SD(\mu, B)$.
\end{example}

\subsection{Attractors}

The notion of finite attractor has been used in several contexts like
probabilistic lossy channel systems (see
\emph{e.g.}~\cite{ABRS05,rabinovitch06})
and abstracted in~\cite{ABM07} in the context of denumerable Markov
chains. A finite attractor is a finite set of states which is reached
almost-surely from every state of the system. We lift this definition
to our context, obviously relaxing the finiteness assumption, since it
is very unlikely that systems with a continuous state-space will have
finite attractors. Since the whole set of states is a trivial
attractor, this general definition will prove useful once we are able
to define attractors with some finiteness property, which will be done
through \emph{abstractions} in Section~\ref{sec:abstractions}.

\begin{definition}
  Let $\mu \in \Dist(S)$ be an initial distribution.  $B\in\Sigma$ is
  a \emph{$\mu$-attractor for $\calT$} if
  $\Prob^{\calT}_{\mu}(\F B)=1$ . Further, $B$ is an \emph{attractor
    for $\calT$} if it is a $\mu$-attractor for every
  $\mu\in\Dist(S)$.
\end{definition}

\begin{example}
  Consider again the random walk $\calT_\rw$ of
  Example~\ref{Example:DMCRandomWalk} and assume again that
  $p>1/2$. For $B=\lbrace 5\rbrace$, it can be shown that $B$ is a
  $\mu$-attractor for $\mu=\delta_0$. However, for any distribution
  $\mu'\in\Dist(\IN_{\ge 6})$ over natural numbers larger than or
  equal to $6$,
  $\Prob^{\calT_1}_{\mu'}(\F B)<1$ and thus $B$ is not a
  $\mu'$-attractor.

  On the other hand, if we assume $p \le 1/2$, it is a well-known
  property of random walks that $\{0\}$ is reached almost-surely from
  every state, hence we can infer that any bounded subset $A$ of $\IN$
  is an attractor (for every initial distribution).
\end{example}

The existence of an attractor is a rather strong property for an
STS. It will actually imply that attractors are almost-surely visited
infinitely often. While this result was already known for DMCs
(see~\cite[Lemmas 3.2 and 3.4]{ABM07}), the general case of STSs
requires a specific treatment. We choose to present the proof of this
result in the core of the paper, since it illustrates how to use at
the same time the two views on STSs.

\begin{restatable}{lemma}{attractorGF}
  \label{lemma:attractorGF}
  If $B$ is an attractor for $\calT$ then for every initial
  distribution $\mu\in \Dist(S)$,
  \[\Prob^{\calT}_{\mu}(\G \F B)=1. \]
\end{restatable}

\begin{proof}
  Let $B$ be an attractor for $\calT$, i.e. for each initial
  distribution $\mu\in\Dist(S)$, $\Prob_{\mu}^{\calT}(\F
  B)=1$. Towards a contradiction, assume that there is
  $\mu\in\Dist(S)$ such that $\Prob_{\mu}^{\calT}(\G\F B)<1$. Then,
  $\Prob_{\mu}^{\calT}(\F\G B^c)>0$. Now remember that from the
  definitions, we have that
  \[
  \ev{\calT}{\F\G B^c} = \bigcup_{n\ge 0} \bigcap_{m\ge 0}
  \Cyl(\overbrace{S,\ldots,S}^{n\text{
      times}},\overbrace{B^c,\ldots,B^c}^{m\text{ times}}).
  \]
  It follows that there is $n\in\IN$ such that
  \[
  \lim_{m\rightarrow\infty}\Prob_{\mu}^{\calT}(\Cyl(\overbrace{S,\ldots,S}^{n\text{
      times}},\overbrace{B^c,\ldots,B^c}^{m\text{ times}}))>0. 
  \] 
  From Lemma~\ref{lemma:integration}, if we write $\nu_0=\mu$ and
  $\nu_j=\Omega_{\calT}(\nu_{j-1})$ for each $1\leq j\leq n-1$, we get
  that for each $m\ge 1$,
  \[
  \Prob_{\mu}^{\calT}(\Cyl(\overbrace{S,\ldots,S}^{n\text{
      times}},\overbrace{B^c,\ldots,B^c}^{m\text{ times}})) =
  \Prob_{\Omega_{\calT}(\nu_{n-1})}^{\calT}(\Cyl(\overbrace{B^c,\ldots,B^c}^{m\text{
      times}})) 
  \] 
  since $\mu(S)=1$ and for each $0\leq j\leq n-2$,
  $(\Omega_{\calT}(\nu_j))(S)=1$. It can be seen that in this case,
  for each $0\leq j\leq n-1$, $\nu_j=\Omega_{\calT}^{(j+1)}(\mu)$. We
  write
  $\nu=\Omega_{\calT}(\nu_{n-1})=\Omega_{\calT}^{(n)}(\mu)\in\Dist(S)$. We
  thus get that
  \[
  \lim_{m\rightarrow\infty}\Prob_{\nu}^{\calT}(\Cyl(\overbrace{B^c,\ldots,B^c}^{m\text{
      times}}))=\Prob_{\nu}^{\calT}(\G B^c)>0,
  \] 
  which contradicts the fact that $B$ is an attractor, hence a
  $\nu$-attractor, for $\calT$.
\end{proof}


\subsection{Fairness}
Fairness is a standard notion in probabilistic
systems~\cite{Pnu83,PZ93,BK98}, saying that something which is allowed
infinitely often should happen infinitely often almost-surely. This
can for instance be instantiated in denumerable Markov chains as
follows: if a state $s$ is visited infinitely often, and the
probability to move from $s$ to $s'$ is positive, then, almost-surely,
infinitely often the state $s'$ is visited. Not all Markov chains are
fair, but finitely-branching Markov chains are. Note that other
notions of fairness are discussed in~\cite{BK98}.

The above notion of fairness cannot be lifted directly to continuous
state-space STSs (since for two states $s$ and $s'$, the probability
to move from $s$ to $s'$ is likely to be $0$). A more careful
definition of this notion must be provided for general STSs. For
$B\in\Sigma$, we define
\[
\PreProb^\calT(B)=\lbrace
B'\in\Sigma\mid\forall\mu'\in \Dist(B'), \
\Prob^{\calT}_{\mu'}(\Cyl(B',B))>0\rbrace \enspace,
\] as the set of measurable sets $B'$ ``from which'' $B$ can be
reached with positive probability.  Note that, ideally we would like
to define the maximal set that allows one to reach $B$, but the union
of all such sets may not be measurable in our general context.

\begin{definition}
  Let $\mu \in \Dist(S)$ be some initial distribution, and $B \in \Sigma$.
  The STS $\calT$ is \emph{fair w.r.t. $B$ from $\mu$}, written
  $\calT$ is $\fair(\mu,B)$, if for every 
  $B'\in\PreProb^\calT(B)$, $\Prob^{\calT}_{\mu}(\G\F B')>0$ implies
  \[
  \Prob^{\calT}_{\mu}(\G\F B\mid \G\F B')=1.
  \]

  {As for decisiveness, we extend this definition to sets $\calB
    \subseteq \Sigma$, and use similar notations when we relax the
    fixed initial measure $\mu$. Finally, we say that $\calT$ is
    \emph{strongly fair} whenever it is fair w.r.t. $B$ from $\mu$ for
    every $B \in \Sigma$ and every $\mu \in \Dist(S)$.}
\end{definition}

\begin{example}\label{example:fairness}
  Consider again the random walk of
  Example~\ref{Example:DMCRandomWalk}. In $\calT_\rw$, there is a
  positive lower bound on each single move in the chain: take
  $\varepsilon=\min(p,1-p)>0$, and realize then that for every
  $B\subseteq S$, for every $B'\in\PreProb^{\calT_\rw}(B)$ and for every
  $s\in B'$, $\kappa_1(s,B)\ge\varepsilon$. Then using a proof similar
  to~\cite[Theo. 2]{BK98}, we deduce that $\calT_\rw$ is fair
  w.r.t. $B$ from any initial distribution. Hence, $\calT_\rw$ is
  strongly fair.
\end{example}

\begin{example}[Counter-example]\label{counterexample:fairness}
  Consider now the DMC $\calT_{\textsf{\upshape unfair}}$ depicted in
  Figure~\ref{Figure:Fairness}. From each state $a_n$ ($n \ge 1$), the
  probability to move to $b$ is $\frac{1}{3^n}$ whereas the
  probability to move to $a_{n+1}$ is $1-\frac{1}{3^n}$. From $b$, the
  probability to move to $a_1$ is $1$. 
  
  Consider $B = \{b\}$, $\mu=\delta_{b}$ and
  $B'=\lbrace a_n\mid n\in\IN\rbrace$. It is easy to see that,
  $B'\in\PreProb^{\calT_{\textsf{\upshape unfair}}}(B)$ and that
  $\Prob_{\mu}^{\calT_{\textsf{\upshape unfair}}}(\G\F
  B')>0$. However, $\Prob_{\mu}^{\calT_{\textsf{\upshape
        unfair}}}(\G\F B\mid \G\F B')<1$ 
  and thus $\calT_{\textsf{\upshape unfair}}$ is not $\fair(\mu,B)$.


\begin{figure}[!h]
\centering
\begin{tikzpicture}
\tikzstyle{ptt}=[scale=1]
\tikzstyle{loc}=[ptt,draw,circle,minimum size =.9cm];
\tikzstyle{inv}=[ptt,circle,minimum size =.9cm];
\tikzstyle{fleche}=[->,>=stealth', rounded corners=1pt];

\node[loc] (A1) at (0,-4) {$a_1$};
\node[loc] (A2) at (2, -4) {$a_2$};
\node[loc] (A3) at (4, -4) {$a_3$};
\node[loc] (A4) at (6, -4) {$a_4$};
\node[loc] (B) at (0, -1.5) {$b$};
\draw[fleche] (A1) -- (A2) node[midway, below] {$\frac{2}{3}$};
\draw[fleche] (A2) -- (A3) node[midway, below] {$\frac{8}{9}$};
\draw[fleche] (A3) -- (A4) node[midway, below] {$\frac{26}{27}$};
\draw[fleche,dashed] (A4) -- (8, -4);
\draw[fleche] (A1) -- (B) node[pos=.4, right] {$\frac{1}{3}$};
\draw[fleche] (A2.150) -- (B.300) node[pos=.28, above] {$\frac{1}{9}$};
\draw[fleche] (A3.120) -- (B.330) node[pos=.2, above] {$\frac{1}{27}$};
\draw[fleche, dashed] (A4.north) -- (B.east);
\draw[fleche] (B) to[bend right=45] node[midway, left] {$1$} (A1);
\end{tikzpicture} \caption{A denumerable Markov chain $\calT_3$ that is not strongly fair.}
\label{Figure:Fairness}
\end{figure}
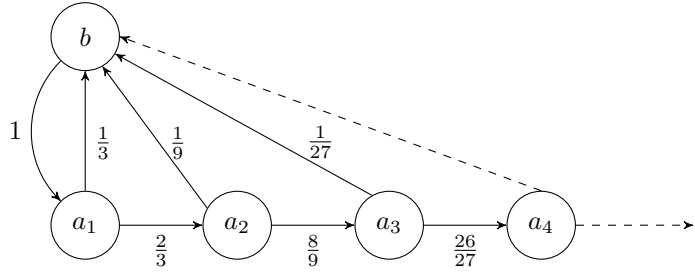
\end{example}


\subsection{Relationships between the various properties}
\label{subsec:link}


In this section, we compare all the notions, and give the precise
links between all of them. We first analyze the general case, and
reinforce the results in the case of DMCs.

\begin{restatable}{proposition}{proplink}
  \label{prop:links}
  Let $\calB\subseteq \Sigma$ and $\mu\in\Dist(S)$. The following
  implications 
hold:
  \begin{enumerate}
  \item $\calT$ is $\D(\mu,\calB)$ $\Longleftarrow$ $\calT$ is
    $\SD(\mu,\calB)$ $\iff$ $\calT$ is $\PD(\mu,\calB)$ $\implies$
    $\calT$ is $\fair(\mu,\calB)$
  \item $\calT$ is $\D(\calB)$ $\iff$ $\calT$ is $\SD(\calB)$ $\iff$
    $\calT$ is $\PD(\calB)$ $\implies$ $\calT$ is $\fair(\calB)$
  \end{enumerate}
      

\end{restatable}

The three missing implications in the above proposition do actually
not hold, as proved by the following example. We also illustrate the
fact that $\D(\mu,\calB)$ and $\fair(\mu,\calB)$ are incomparable.

\begin{example}[Counter-example]
  Consider again the random walk $\calT_\rw$ of
  Example~\ref{Example:DMCRandomWalk}. We have shown in
  Example~\ref{example:fairness} that $\calT_\rw$ is strongly fair,
  whatever the choice of $p$. Now let us assume that $p>1/2$ and let
  us consider the initial distribution $\mu=\delta_0$, the Dirac
  distribution over $0$. Then from Example~\ref{example:BtildeDeci},
  $\calT_\rw$ is decisive from $\mu$ w.r.t. any set of states.
  Again in this example, we have observed that it is not strongly
  decisive w.r.t. any set of the form $B=\lbrace i\rbrace$ with $i\ge
  0$. This shows that $\calT_\rw$ is neither
  $\D(\mu,\calB)\Rightarrow\SD(\mu,\calB)$, nor
  $\fair(\mu,\calB)\Rightarrow \SD(\mu,\calB)$, nor
  $\fair(\calB)\Rightarrow \SD(\calB)$. And since $\calT_\rw$ is not
  decisive from $\delta_1$ w.r.t. $\lbrace 0\rbrace$, this also proves
  that $\fair(\mu,\calB)$ does not imply $\D(\mu,\calB)$.

  In order to illustrate that $\D(\mu,\calB)$ does not imply
  $\fair(\mu,\calB)$ in general, we consider the DMC chain
  $\calT_{\textsf{\upshape unfair}}$ of
  Example~\ref{counterexample:fairness}. We consider $B=\{ b\}$ and
  $\mu=\delta_{b}$. It is easily observed that
  $\calT_{\textsf{\upshape unfair}}$ is $\D(\mu,B)$ as we start in $b$
  with probability $1$, but we have shown that
  $\calT_{\textsf{\upshape unfair}}$ is not $\fair(\mu,B)$.
\end{example}


If $\calT$ is a DMC, \emph{i.e}. if $S$ is at most denumerable and
$\Sigma=2^S$, we can complete the picture using furthermore the result
of~\cite[Lemma 3.4]{ABM07}, which says that any DMC with a finite
attractor is decisive w.r.t. any set of states.

\begin{proposition}
  \label{prop:linksDMC}
  Assume $\calT$ is a DMC. The following implications 
  hold:
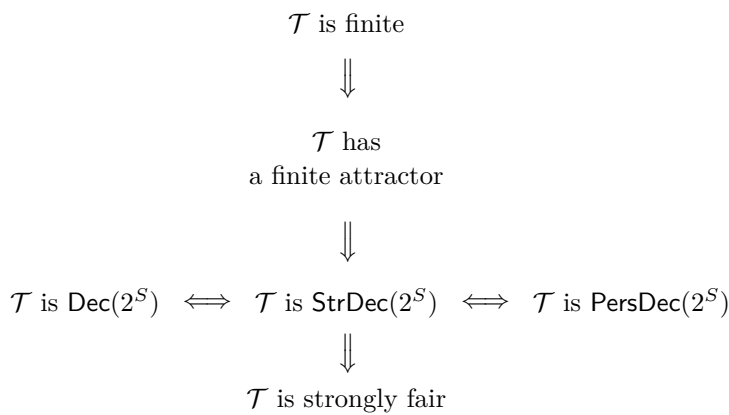
\begin{center}
  \begin{tikzpicture}
    \node (attractor) {\begin{tabular}{c} {\normalsize $\calT$ has} \\
        {\normalsize a finite  attractor} \end{tabular}};

    \node (finite) [above=of attractor] {$\calT$ is finite};

    \node (strongly_decisive) [below=of attractor]
    {$\calT$ is $\SD(2^S)$};

    \node (decisive) [left=of strongly_decisive] {$\calT$ is
      $\D(2^S)$};

    \node(persistently_decisive) [right=of strongly_decisive] {$\calT$
      is $\PD(2^S)$};
      
    \node (fair) [below=of strongly_decisive,yshift=.3cm] {$\calT$ is
      strongly fair};

    \path (finite) -- (attractor) node [midway] {\rotatebox{-90}{$\implies$}};
    \path (attractor) -- (strongly_decisive) node [midway] {\rotatebox{-90}{$\implies$}};
    \path (decisive) -- (strongly_decisive) node [midway] {$\iff$};
    \path (strongly_decisive) --  (persistently_decisive) node
    [midway] {$\iff$};
    \path (strongly_decisive) -- (fair)  node [midway] {\rotatebox{-90}{$\implies$}};
  \end{tikzpicture}
\end{center}
\end{proposition}




\section{Abstractions between STSs}
\label{sec:abstractions}
While decisiveness is well-defined for general STSs, proving that a
given STS $\calT$ is decisive might be technical in general. A
standard approach in model-checking to avoid such difficulties is to
abstract the system into a simpler one, that can be analyzed and
provides guarantees on the concrete system.  We thus propose a notion
of abstraction, which will help proving  properties of
general STSs. Also, through abstractions, we will be able to
characterize meaningful attractors.

Abstractions only preserve the positivity of probabilities, which is
rather weak. However, by strengthening the notion, we will see that
they will be very useful, not only for the qualitative analysis of
STSs, but also for their quantitative analysis.

\subsection{Abstraction}
\label{subsec:abst}
Let $\calT_1 = (S_1,\Sigma_1,\kappa_1)$ and $\calT_2 = (S_2,\Sigma_2,
\kappa_2)$ be two STSs. Let $\alpha : (S_1,\Sigma_1) \to
(S_2,\Sigma_2)$ be a measurable function.  A set $B \in \Sigma_1$ is
said \emph{$\alpha$-closed} whenever $B = \alpha^{-1}(\alpha(B))$: for
every $s,s' \in S_1$, if $s \in B$ and $\alpha(s) = \alpha(s')$, then
$s' \in B$.  Following~\cite{GBK16}, we define the pushforward of
$\alpha$ as $\alpha_\# : \Dist(S_1) \to \Dist(S_2)$ by
$\alpha_\#(\mu)(M_2) = \mu(\alpha^{-1}(M_2))$ for every $\mu \in
\Dist(S_1)$ and for every $M_2 \in \Sigma_2$. The role of the
pushforward $\alpha_{\#}$ is to transfer the measures from
$(S_1,\Sigma_1)$ to $(S_2,\Sigma_2)$.  In the following we will say
that two probability distributions $\mu$ and $\nu$ over some
probability space $(S,\Sigma)$ are \emph{qualitatively equivalent} if
for each $A\in\Sigma$, $\mu(A)=0\Leftrightarrow\nu(A)=0$.

\begin{definition}
  $\calT_2$ is an \emph{$\alpha$-abstraction} of $\calT_1$ if
  \begin{equation*}
    \forall \mu \in
    \Dist(S_1),\ \alpha_{\#}(\Omega_{\calT_1}(\mu))\ \text{is equivalent to} \ \Omega_{\calT_2}(\alpha_{\#}(\mu)) \enspace.
  \end{equation*}
\end{definition}
Later, one may speak of abstraction instead of $\alpha$-abstraction if
$\alpha$ is clear in the context.

From the definitions of $\Omega_{\calT}$, $\alpha_{\#}$ and equivalent
measures, the notion of $\alpha$-abstraction equivalently requires
that for every $\mu\in\Dist(S_1)$ and every $A\in\Sigma_2$,
\[
\Prob_{\mu}^{\calT_1}(\Cyl(S_1,\alpha^{-1}(A)))>0 \Longleftrightarrow
\Prob_{\alpha_{\#}(\mu)}^{\calT_2}(\Cyl(S_2, A))>0 \enspace. 
\]
Intuitively, the two STSs have the same ``qualitative'' single steps.


Let us now provide two examples of $\alpha$-abstraction.
  
  \begin{example}\label{ex:abstract1}
    Consider the two STSs $\calT_{\rw}=(S_\rw,\Sigma_\rw,\kappa_\rw)$
    --the discrete random walk on $\IN$ from
    Example~\ref{Example:DMCRandomWalk}-- and
    $\calT_\qs=(S_\qs, \Sigma_\qs, \kappa_\qs)$ --its continuous
    variant from Example~\ref{Example:Continuous}. Recall that
    $S_\rw=\IN$ and $S_\qs=\IN\times\Rpos$. Letting
    $\alpha: S_\qs \to S_\rw$ be the function defined by
    $\alpha((n,t))=n$, one can easily be convinced that $\calT_{\rw}$ is an
    $\alpha$-abstraction of $\calT_{\qs}$.
  \end{example}

    \begin{example}\label{ex:abstract2}
      Consider the random walk on $\IN$ from
      Example~\ref{Example:DMCRandomWalk}
      $\calT_{\rw}=(S_\rw,\Sigma_\rw,\kappa_\rw)$, and the finite
      Markov chain $\calT_f=(S_f, \Sigma_f,\kappa_f)$ depicted on
      Fig.~\ref{Figure:fMC}. We define a function
      $\alpha: S_\rw \to S_f$ by $\alpha(n)=s_n$ if $n\in \lbrace 0,1\rbrace$, and
      $\alpha(n)=2$ otherwise. Clearly enough $\calT_{f}$ is an
      $\alpha$-abstraction of $\calT_{\rw}$.

\begin{figure}[!h]
  \centering
  \begin{tikzpicture}
    \tikzstyle{ptt}=[scale=1]
    \tikzstyle{loc}=[ptt,draw,circle,minimum size =1cm, thick];
    \tikzstyle{inv}=[ptt,circle,minimum size =1cm, thick];
    \tikzstyle{fleche}=[->,>=stealth', thick, rounded corners=1pt];
    \node[loc] (lzero) at (0,0) {$s_0$};
    \node[loc] (lun) at (2.2,0) {$s_1$};
    \node[loc] (ldeux) at (4.4,0) {$s_2$};
    \draw[fleche] (lzero) to[bend left=30] node[ptt,midway, above] {$1$} (lun);
    \draw[fleche] (lun) to[bend left=30] node[ptt,midway, below] {$1-q$} (lzero);
    \draw[fleche] (lun) to[bend left=30] node[ptt,midway, above] {$q$} (ldeux);
    \draw[fleche] (ldeux) to[bend left=30] node[ptt,midway, below] {$1-q$} (lun);
 
    \draw[fleche] (ldeux)  to[in=-30,out=30,loop] node[ptt,midway,right] {$q$}(ldeux);
  \end{tikzpicture} \caption{A finite Markov chain.}
  \label{Figure:fMC}
\end{figure}
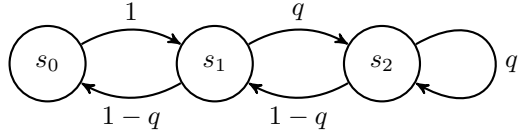
  \end{example}

  The notion of $\alpha$-abstraction naturally extends to labelled
  STSs.  A labelled STS $\calT_2=(S_2, \Sigma_2, \kappa_2, \AP_2,
  \calL_2)$ is an $\alpha$-abstraction of $\calT_1=(S_1, \Sigma_1,
  \kappa_1, \AP_1, \calL_1)$ whenever:
\begin{itemize}
\item $(S_2, \Sigma_2, \kappa_2)$ is an $\alpha$-abstraction of $(S_1, \Sigma_1, \kappa_1)$;
\item $\AP_1=\AP_2$;
\item for every $s_1$, $s'_1\in S_1$,
  $\alpha(s_1)=\alpha(s'_1)\Rightarrow \calL_1(s_1)=\calL_1(s'_1)$;
\item for every $s\in S_1$,
  $\calL_1(s)=a\Rightarrow\calL_2(\alpha(s))=a$.
\end{itemize}
The two last conditions imply that for each
$a\in 2^{\AP}$, $\calL_1^{-1}(\lbrace a\rbrace)$ is
$\alpha$-closed. Moreover, for each $a\in 2^{\AP}$,
$\alpha^{-1}(\calL_2^{-1}(\lbrace a\rbrace))=\calL_1^{-1}(\lbrace
a\rbrace)$. 

\medskip By definition, abstractions preserve the positivity of the
probability of single-step moves. More generally, one easily shows
(see the appendix page~\pageref{app:pos} for details) that the
positivity of reachability properties or more generally of properties
with bounded witnesses is also preserved through
$\alpha$-abstractions: assuming $\calT_2$ is an $\alpha$-abstraction
of $\calT_1$, for every $\mu \in \Dist(S_1)$, for every $A,B \in
\Sigma_2$:
\begin{equation}
\label{eq:PosReach}
\Prob_\mu^{\calT_1}(\ev{\calT_1}{\alpha^{-1}(A) \U \alpha^{-1}(B)}) >
0 \Longleftrightarrow
\Prob_{\alpha_{\#}(\mu)}^{\calT_2}(\ev{\calT_2}{A \U B}) > 0 
\end{equation}
However this does not apply to liveness properties, such as
$\ev{\calT}{\G \F A}$ with $A \in \Sigma_2$, or to other qualitative
questions like almost-sure reachability (reach $B$ with probability
$1$). To ensure that these more involved properties are preserved via
abstraction, we will strengthen the assumptions on the abstraction and
on the STSs. \label{PosReach}

\paragraph*{Soundness and completeness of abstractions}

We now define soundness and completeness of abstractions, that allow
one to lift properties of an abstraction $\calT_2$ to the concrete STS
$\calT_1$.  For the rest of this section, we fix an STS $\calT_1$ and
let $\calT_2$ be an $\alpha$-abstraction of $\calT_1$.

\begin{definition}
  Let $\mu\in\Dist(S_1)$. The $\alpha$-abstraction $\calT_2$ is
  \emph{$\mu$-sound} whenever for every $B \in \Sigma_2$:
  \[
  \Prob^{\calT_2}_{\alpha_{\#}(\mu)} (\F B) = 1 \Longrightarrow
  \Prob^{\calT_1}_{\mu} (\F \alpha^{-1}(B)) = 1\enspace.
  \]
  $\calT_2$ is a \emph{sound} $\alpha$-abstraction of $\calT_1$ if it
  is $\mu$-sound for every $\mu \in \Dist(S_1)$.
\end{definition}
Assuming soundness, it will thus be sufficient to prove that a
reachability property holds almost surely in the abstraction to derive
that the corresponding reachability property also holds almost surely
in the concrete STS.

\begin{definition}
  Let $\mu\in\Dist(S_1)$.  The $\alpha$-abstraction $\calT_2$ is
  \emph{$\mu$-complete} whenever for every $B \in \Sigma_2$,
  \[
  \Prob^{\calT_1}_{\mu} (\F \alpha^{-1}(B)) = 1 \Longrightarrow
  \Prob^{\calT_2}_{\alpha_{\#}(\mu)} (\F B) = 1
  \]
  $\calT_2$ is a \emph{complete} $\alpha$-abstraction of $\calT_1$ if
  it is $\mu$-complete for every $\mu \in \Dist(S_1)$.
\end{definition}
Assuming completeness, if a reachability property holds with
probability smaller than $1$ in the abstraction, then the
corresponding property will also have probability smaller than $1$ in
the concrete model.

Altogether, sound and complete abstractions will guarantee that, up to
$\alpha$, the same reachability properties are satisfied almost-surely
in $\calT_1$ and in $\calT_2$.


\begin{example}\label{example:abstr} We continue here
    Example~\ref{ex:abstract1} with $\calT_{\rw}$ (with parameter
    $0<p<1$) and $\calT_\qs$ (with parameters $\lambda>0$ and
    $\nu>0$).  One can show that $\calT_\rw$ is a sound and complete
    abstraction of $\calT_\qs$ whenever
    $p>1/2\Longleftrightarrow \lambda>\nu$.
    
\end{example}

\begin{example}\label{example:abstrbis} We continue here
  Example~\ref{ex:abstract2} with $\calT_{\rw}$ (with parameter
  $0<p<1$) and $\calT_f$ (with parameter $0<q<1$).  One can show that
  $\calT_f$ is always a complete abstraction of $\calT_\rw$; Moreover,
  $\calT_f$ is also sound if and only if $p \le 1/2$.
\end{example}

\subsection{Transfer of properties through abstractions}
\label{subsec:transfer}

In this section, we explain how and under which conditions one can
transfer interesting decisiveness, attractor and fairness properties
of STSs through abstractions.

\subsubsection{The case of sound abstractions}
\begin{restatable}{proposition}{MuDecisiveAbstr}
\label{thm:MuDecisiveAbstr}\label{coro:SoundDecisive}
  If $\calT_2$ is a $\mu$-sound $\alpha$-abstraction of $\calT_1$,
  then for every $B \in \Sigma_2$:
  \[
  \calT_2\ \text{is}\ \D(\alpha_{\#}(\mu),B)\ \implies\
  \calT_1\ \text{is}\ \D(\mu,\alpha^{-1}(B)) \enspace.
  \]
\end{restatable}


Using equivalences between the various properties
stated in Proposition~\ref{prop:links}, we can extend the above result
to other decisiveness properties: assuming $\calT_2$ is a sound
$\alpha$-abstraction of $\calT_1$, for every $B \in \Sigma_2$, if
$\calT_2$ is $\D(B)$ (or equiv. $\SD(B)$, $\PD(B)$) then $\calT_1$ is
$\D(\alpha^{-1}(B))$ (or equiv. $\SD(\alpha^{-1}(B))$,
$\PD(\alpha^{-1}(B))$).



The definitions of attractor and of sound $\alpha$-abstraction yield a
similar result:
\begin{proposition}
  \label{lem:attr-via-sound}
  If $\calT_2$ is a sound $\alpha$-abstraction of $\calT_1$ and if
  $A\in\Sigma_2$ is an attractor for $\calT_2$, then $\alpha^{-1}(A)$
  is an attractor for $\calT_1$.
\end{proposition}


Denumerable (and in particular finite) abstractions play an important
role, hence we summarize all interesting and useful implications and
equivalences for DMCs, which are direct consequences of
Propositions~\ref{prop:linksDMC} and~\ref{thm:MuDecisiveAbstr}.


\begin{proposition}
  \label{prop:DMCabs}
  Assume $\calT_2$ is a DMC, which is an $\alpha$-abstraction of
  $\calT_1$. Let $\calB=\lbrace \alpha^{-1}(B)\mid
  B\in\Sigma_2\rbrace$ be the set of $\alpha$-closed sets of
  $\Sigma_1$. The following implications and equivalences hold true:
  
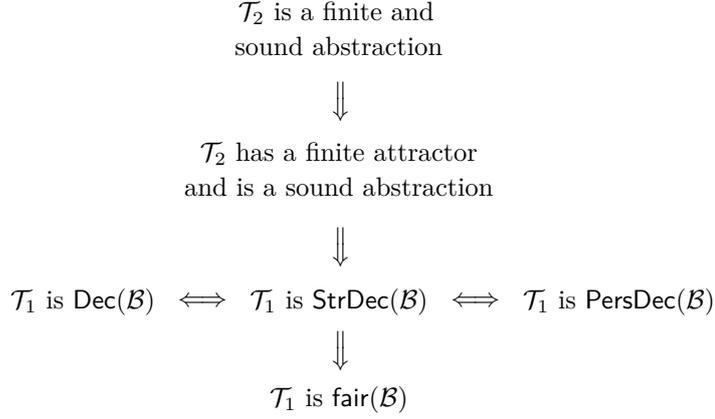
\begin{center}
    \begin{tikzpicture}
      \node (attractor) {\begin{tabular}{c} {\normalsize $\calT_2$ has a finite
          attractor} \\ {\normalsize and is a sound abstraction} \end{tabular}};
      \node (strongly_decisive) [below=of attractor,yshift=.15cm] {$\calT_1$ is $\SD(\calB)$};
      \node (decisive) [left=of strongly_decisive] {$\calT_1$ is $\D(\calB)$};
      \node (finite) [above=of attractor,yshift=-.3cm] {\begin{tabular}{r}
          {\normalsize $\calT_2$ is a finite and } \\ {\normalsize 
            sound abstraction} \end{tabular}};
      \node (persistently_decisive) [right=of strongly_decisive]
      {$\calT_1$ is $\PD(\calB)$};
      
      \node (fair) [below=of strongly_decisive,yshift=.3cm] {$\calT_1$ is $\fair(\calB)$};
      
      \path (finite) -- (attractor)  node [midway] {\rotatebox{-90}{$\implies$}};
      \path (attractor) -- (strongly_decisive) node [midway] {\rotatebox{-90}{$\implies$}};
      \path (decisive) -- (strongly_decisive) node [midway] {$\iff$};
      \path (strongly_decisive) --  (persistently_decisive) node
      [midway] {$\iff$};
      \path (strongly_decisive) -- (fair)  node [midway] {\rotatebox{-90}{$\implies$}};
    \end{tikzpicture}
  \end{center}
\end{proposition}

\subsubsection{Trickier transfers of properties}
We established that decisiveness properties could be transferred
through sound abstractions. However, proving soundness of an
abstraction is not easy in general, and one way to do it is by proving
some decisiveness properties. It is therefore relevant to explore
alternatives to prove decisiveness properties. We provide here two
such alternatives.

First, we assume a denumerable abstraction, and lower bounds on
probabilities of reachability properties. 
\begin{restatable}{proposition}{attractorSound}
  \label{prop:attractorSound}
  Let $\calT_2$ be a DMC such that $\calT_2$ is an
  $\alpha$-abstraction of $\calT_1$.
  \begin{enumerate}
    \item Assume that there is a finite set
  $A_2=\lbrace s_1, \ldots, s_n\rbrace\subseteq S_2$ such that $A_2$
  is an attractor for $\calT_2$ and $A_1=\bigcup_{i=1}^n
  \alpha^{-1}(s_i)=\alpha^{-1}(A_2)$ is an attractor for
  $\calT_1$.
\item Assume moreover that for every $1\leq i\leq n$, for every
  $\alpha$-closed set $B$ in $\Sigma_1$, there exist $p>0$ and
  $k\in\IN$ such that:
  \begin{itemize}
  \item for every $\mu\in\Dist(\alpha^{-1}(s_i))$,
    $\Prob^{\calT_1}_{\mu}(\F[\leq k] B)\geq p$, or
  \item for every $\mu\in\Dist(\alpha^{-1}(s_i))$,
    $\Prob^{\calT_1}_{\mu}(\F B)=0$.
  \end{itemize}
\end{enumerate}
  Then $\calT_1$ is decisive w.r.t. every $\alpha$-closed set.
\end{restatable}

While the first condition on transfer of attractors is easily
readable, let us discuss the second one. It intuitively says that,
whenever some ($\alpha$-closed) set $B$ can be reached with positive
probability from some distribution $\mu$ with support
$\alpha^{-1}(s_i)$, where $s_i$ is an element of the finite attractor
of $\calT_2$, then it should be reachable with a probability
lower-bounded by some $p$, $p$ being independent of the choice of
$\mu$; furthermore an upper bound $k$ on the number of steps for
reaching $B$ is technically required in the proof, but we do not know
whether it is needed for the result to hold.

We write $(\dag)$ for the hypotheses over $\calT_1$ in this
proposition.\label{notation:dag}
The idea behind this result is that, with probability $1$, the
attractor of $\calT_1$ will be visited infinitely often, and, if at
each visit of the attractor, there is a positive probability to reach
some ($\alpha$-closed) set $B$, since that probability is by
assumption bounded from below, then $B$ will indeed be visited
infinitely often with probability $1$. This will allow to show the
dichotomy between reachability of $B$ and reachability of
$\widetilde{B}$, which is required for proving the decisiveness
property. The full proof is given in the appendix,
page~\pageref{app-attractorsound}, but we give here a sketch. Note
that this kind of proofs appears quite often in the literature (see
\emph{e.g.}~\cite[Lemma 3.4]{ABM07}, but we have to do it carefully
here, since the framework is rather general).

\begin{proof}[Sketch of proof]
    Fix $B\subseteq S_2$ and $\mu\in\Dist(S_1)$. Towards a
    contradiction, assume that $\calT_1$ is not $\mu$-decisive
    w.r.t. $B$: this means that
    $\Prob_{\mu}^{\calT_1}(\G\alpha^{-1}(B^c)\wedge \G
    \alpha^{-1}((\Btilde)^c))>0$.

    Since $A_1$ is an attractor, we deduce from
    Lemma~\ref{lemma:attractorGF} that
    \[
    \Prob_{\mu}^{\calT_1}(\G\alpha^{-1}(B^c)\wedge \G
    \alpha^{-1}((\Btilde)^c)\wedge\G\F A_1)>0\enspace.
    \]
    We write $A'_2\subseteq A_2$ for the non-empty\footnote{Since $A_1
      = \bigcup_{s \in A_2} \alpha^{-1}(s)$.} set of states $s$ of
    $A_2$ such that
    \[
    \Prob_{\mu}^{\calT_1}(\G\alpha^{-1}(B^c)\wedge \G
    \alpha^{-1}((\Btilde)^c)\wedge\G\F \alpha^{-1}(s))>0
    \]
    Then obviously $A'_2\subseteq B^c\cap(\Btilde)^c$.

    Since $A'_2\subseteq(\Btilde)^c$, from
    Lemma~\ref{lemma:BTildeComplementaire} (third item), the
    hypothesis $(\dag)$
    and the finiteness of $A'_2$, we get that
    there is $p>0$ and $k\in\IN$ such that for every $s\in A'_2$ and
    every $\mu\in\Dist(\alpha^{-1}(s))$, 
    \[
    \Prob_\mu^{\calT_1}(\F[{\leq k}] B)\ge p\enspace.
    \]
    Writing $A'_1$ for $\alpha^{-1}(A'_2)$, we can show that
    \begin{eqnarray*}
      0 & < & \Prob_{\mu}^{\calT_1}(\G\alpha^{-1}(B^c)\wedge \G \alpha^{-1}((\Btilde)^c)\wedge\G\F A'_1)\notag\\
      {} & \leq & \Prob_{\mu}^{\calT_1}(\G\alpha^{-1}(B^c)\wedge\G\F A'_1)\notag\\
      {} & \leq & \lim_{n\to\infty} (1-p)^n =0\notag
    \end{eqnarray*}
    which is the required contradiction.
\end{proof}

Second, we strengthen the hypothesis on the abstraction, assuming it
is finite, but we relax the condition on $\calT_1$, requiring only a
fairness property.

\begin{restatable}{proposition}{propfairness}
  \label{prop:fairness}
  Let $\calT_2$ be a finite Markov chain such that $\calT_2$ is an
  $\alpha$-abstraction of $\calT_1$. 
  Fix $\mu \in \Dist(S_1)$, and assume that $\calT_1$ is $\mu$-fair
  w.r.t. every $\alpha$-closed set.
  Then $\calT_1$ is $\mu$-decisive w.r.t. every $\alpha$-closed set.
\end{restatable}

\begin{proof}[Sketch of proof]
  We give here the main steps of the proof, the details are postponed
  to the appendix (page~\pageref{app-fairness}).

  A key element of the proof relies on the fact that, since $\calT_2$
  is a finite MC, it can be viewed as a graph and we can talk of the
  \emph{bottom strongly connected components (BSCC)} of $\calT_2$. The
  first step of the proof aims at showing that, roughly speaking, the
  union of all BSCCs of $\calT_2$ is a $\mu$-attractor for
  $\calT_1$. More precisely, if $\calC=\lbrace s\in S_2\mid\exists
  C\in\mathsf{BSCC}(\calT_2), \ s\in C\rbrace$, we prove that
  $\Prob_\mu^{\calT_1}(\alpha^{-1}(\calC))=1$. This is shown thanks to
  the following arguments:
  \begin{itemize}
  \item for each $s\in S_2$,
    $\Prob_\mu^{\calT_1}(\G\F\alpha^{-1}(s))>0$ implies that
    $s\in\calC$ -- this uses the $\mu$-fairness assumption of $\calT_1$
    w.r.t. $\alpha$-closed sets, and the core property of BSCCs (we
    cannot escape from them);
  \item using Bayes formula, one can decompose the set of paths
    according to the states which are visited infinitely often (which
    corresponds to a decomposition according to the BSCC the path
    ultimately visit).
  \end{itemize}

  Once we have shown that $\alpha^{-1}(\calC)$ is a $\mu$-attractor
  for $\calT_1$, it suffices to observe that for each
  $B \subseteq S_2$ and each BSCC $C$ of $\calT_2$, either
  $B\cap C\neq\emptyset$, or $C\subseteq\Btilde$. Transferring those
  observations to $\calT_1$ and using Bayes formula to decompose
  $\Prob_\mu^{\calT_1}(\F\alpha^{-1}(B)\vee\F\alpha^{-1}(\Btilde))$
  according to which BSCC is reached, it is easy to check that
  $\Prob_\mu^{\calT_1}(\F\alpha^{-1}(B)\vee\F\alpha^{-1}(\Btilde))=1$.
\end{proof}

\subsection{Conditions for completeness and soundness}

In our applications (Section~\ref{sec:appli}), completeness will be
for free. Indeed, a simple condition (finiteness) implies completeness
as stated in the next lemma.
\begin{lemma}
  \label{lemma:finite-soundness}
  If $\calT_2$ is a finite Markov chain and an $\alpha$-abstraction of
  $\calT_1$, then $\calT_2$ is complete.
\end{lemma}

\begin{proof}
  Pick $s_0 \in S_2$, and $\mu \in \Dist(\alpha^{-1}(\{s_0\}))$ (in
  particular, $\alpha_{\#}(\mu) = \delta_{s_0}$, the Dirac measure
  over $\{s_0\}$).  Assume that $\Prob^{\calT_1}_\mu(\F
  \alpha^{-1}(B)) = 1$ but $\Prob^{\calT_2}_{\alpha_{\#}(\mu)} (\F
  B) < 1$.

  Since $\calT_2$ is a finite Markov chain, there are
  $s_1,\ldots,s_n\in S_2$ such that
  \[
  \Prob_{\delta_{s_0}}^{\calT_2}(\Cyl(s_0,s_1,\ldots,s_n))>0
  \]
  and for each $\rho=(s_i)_{i\ge 0}\in\Cyl(s_0,\ldots,s_n)$ and for
  each $i\ge 0$, $s_i\notin B$.

  For each $0\leq i\leq n$, we write $A_i =
  \alpha^{-1}(\{s_i\})$. Then, following Equation~\eqref{eq:PosReach}
  (page~\pageref{eq:PosReach}), we get that
  $\Prob^{\calT_1}_\mu(\Cyl(A_0,A_1,\dots,A_n))>0$. However,
  $\Cyl(A_0,A_1,\dots,A_n) \cap \ev{\calT}{\F \alpha^{-1}(B)} =
  \emptyset$, yielding a contradiction.
\end{proof}
Note that the above lemma does not hold for denumerable
abstractions. To illustrate this, any two random walks over
$\mathbb{N}$ are abstractions of each other, and it is well-known
that almost-sure reachability depends on the probability values.

In general, completeness can be guaranteed by some decisiveness
condition on the abstract system. Note that, since finite Markov
chains are always decisive, the next lemma actually subsumes the
latter one, that we however found interesting to have as such.

\begin{lemma}
  \label{lemma:completeness}
  Let $\mu \in \Dist(S_1)$. Assume that $\calT_2$ is an
  $\alpha$-abstraction of $\calT_1$ and that $\calT_2$ is
  $\D(\alpha^{\#}(\mu))$. Then, $\calT_2$ is a $\mu$-complete
  $\alpha$-abstraction.
\end{lemma}

\begin{proof}
  Fix $B\in\calB$ and assume that $\Prob^{\calT_1}_\mu(\F
  \alpha^{-1}(B)) = 1$ but $\Prob^{\calT_2}_{\alpha_{\#}(\mu)} (\F B)
  < 1$.

  Since $\calT_2$ is $\D(\alpha^{\#}(\mu))$, we infer from
  Lemma~\ref{lem:Btilde} (fifth item) that
  $\Prob^{\calT_2}_{\alpha_{\#}(\mu)} ({(\neg B) \U \widetilde{B}})
  >0$, and applying Equation~\eqref{eq:PosReach}
  (page~\pageref{eq:PosReach}),
  we get that $\Prob^{\calT_1}_{\mu} ({\alpha^{-1}(\neg B) \U
    \alpha^{-1}(\widetilde{B})}) >0$. This contradicts the hypothesis
  that $\Prob^{\calT_1}_\mu(\F \alpha^{-1}(B)) = 1$.
\end{proof}

\medskip Proving soundness is more delicate. We nevertheless show
  that a decisiveness condition on the concrete system will ensure soundness.

\begin{proposition}\label{coro:DecSound}
  Let $\calT_2$ be an $\alpha$-abstraction of $\calT_1$. Assume
  $\calT_1$ is decisive w.r.t. every $\alpha$-closed set. Then
  $\calT_2$ is a sound $\alpha$-abstraction of $\calT_1$.
\end{proposition}

\begin{proof}
  Towards a contradiction assume that $B \in \Sigma_2$ is such that
  $\Prob_\mu^{\calT_1}(\F \alpha^{-1}(B)) < 1$. Since $\calT_1$ is
  decisive w.r.t. $\alpha^{-1}(B)$ from $\mu$, it holds from
  Lemma~\ref{lem:Btilde} (fifth item) that
  $\Prob_\mu^{\calT_1}(\neg \alpha^{-1}(B) \U
  \widetilde{\alpha^{-1}(B)}) > 0$. Applying
  Equation~\eqref{eq:PosReach} again (page~\pageref{eq:PosReach}), we
  get that
  $\Prob_{\alpha_{\#}(\mu)}^{\calT_2}((\neg B \U \widetilde{B}))>0$,
  which contradicts the assumption that
  $\Prob_{\alpha_{\#}(\mu)}^{\calT_2}( \F B) = 1$.
\end{proof}

For $\calT_2$ an $\alpha$-abstraction of $\calT_1$, notice that
completeness is ensured by a decisiveness assumption on $\calT_2$,
whereas soundness requires $\calT_1$ being decisive w.r.t. every
$\alpha$-closed set. While these conditions look very similar, the
condition for soundness is actually harder to check since the abstract
STS $\calT_2$ is expected to be simpler than the original concrete STS
$\calT_1$.

\section{Using attractors for analyzing STSs}

\label{sec:main}

In this section we emphasize a generic approach to the analysis of
STSs w.r.t. properties given by a DMA, when the STS satisfies some
attractor-based property. This approach is inspired by the works
of~\cite{ABRS05,bertrand06} on lossy channel systems, but is new (as
far as we know) in the general context of DMCs, and \emph{a fortiori} of
STSs.

As we will see in the next sections, this will yield procedures (which
can be turned to effective algorithms for some classes of systems) for
the qualitative as well as for the approximate quantitative analysis
of STSs. Many of the proofs are inlined here, since they convey
interesting ideas.

\subsection{The case of DMCs with a finite attractor}
\label{subsec:dmc-attractor}

We fix a finite set of atomic propositions $\AP$, and we let $\calT =
(S,\Sigma,\kappa,\AP,\calL)$ be a labelled DMC. We also let $\calM =
(Q,q_0,E,\calF)$ be a DMA. The product $\calT \ltimes \calM$ has been
defined in Section~\ref{subsec:LSTS}.
First, attractors transfer from $\calT$ to the product $\calT \ltimes
\calM$, as stated below, and proven in the technical appendix
(page~\pageref{app-attractorproduit}).

\begin{restatable}{lemma}{attractorproduit}
  \label{lemma:soundproduct}\label{attractorproduit}
  Assume that $A$ is an attractor for $\calT$. Then $A \times Q$ is an
  attractor for $\calT \ltimes \calM$. Furthermore, if $A$ is finite,
  then so is $A \times Q$.
\end{restatable}

For the rest of this subsection, we assume that $\calT$ has a finite
attractor.  Applying Lemma~\ref{attractorproduit}, the product $\calT
\ltimes \calM$ admits a finite attractor that we denote $B$. We write
$\mathsf{Graph}_{\calT \ltimes \calM}(B)$ (or simply
$\mathsf{Graph}(B)$ when $\calT$ and $\calM$ are clear from the
context) for the finite graph whose vertices are states of $B$, and in
which there is an edge from $(s_1,q_1)$ to $(s_2,q_2)$ if there is a
path from $(s_1,q_1)$ to $(s_2,q_2)$ in $\calT \ltimes \calM$. The
\emph{bottom strongly connected components} (BSCCs)\footnote{Those are
  strongly connected components which cannot be left.} of the graph
$\mathsf{Graph}_{\calT \ltimes \calM}(B)$ play a central role in the
model checking of $\omega$-regular properties of $\calT$. Let us first
discuss the relationships between the BSCCs and attractors for $\calT
\ltimes \calM$.

\begin{lemma}
  \label{lemma:bscc}
  The following properties are satisfied:
  \begin{itemize}
  \item The set $\{(s,q) \in C \mid C\ \text{BSCC of}\
    \mathsf{Graph}_{\calT \ltimes \calM}(B)\}$ is an attractor of
    $\calT \ltimes \calM$.
  \item If $C$ and $C'$ are two distinct BSCCs of
    $\mathsf{Graph}_{\calT \ltimes \calM}(B)$, for every $\mu \in
    \Dist(S\times Q)$, $\Prob_\mu^{\calT \ltimes \calM}(\F C \wedge \F C') =
    0$.
  \item If $C$ is a BSCC of $\mathsf{Graph}_{\calT \ltimes \calM}(B)$,
    for every $\mu \in \Dist(C)$, $\Prob_\mu^{\calT \ltimes \calM}(\G
    \F C) = 1$.
  \end{itemize}
\end{lemma}

\begin{proof}
  The first property is obvious. The second property is a consequence
  of the fact that there is no path between two states of two
  different BSCCs. This second property implies that for each BSCC
  $C'\neq C$ of $\mathsf{Graph}_{\calT\ltimes\calM}(B)$ and for each
  $\mu\in\Dist(C)$, $\Prob_{\mu}^{\calT\ltimes\calM}(\F C')=0$. From
  the first property and Lemma~\ref{lemma:attractorGF}, we know that
  for each $\mu\in\Dist(S\times Q)$,
  $\Prob_{\mu}^{\calT\ltimes\calM}(\G\F
  \bigvee_{C'\in\mathsf{BSCC}(\mathsf{Graph}_{\calT\ltimes\calM}(B))}
  C')=1$. This holds true in particular for each $\mu\in\Dist(C)$ and
  thus, from the previous observation for such initial distributions,
  we get that $\Prob_{\mu}^{\calT\ltimes\calM}(\G\F C)=1$ for each
  $\mu\in\Dist(C)$.
\end{proof}

From Lemma~\ref{lemma:bscc}, the BSCCs of
$\mathsf{Graph}_{\calT \ltimes \calM}(B)$ form an attractor, and once
the system enters a BSCC $C$, only that BSCC will be visited again,
and this will happen infinitely often with probability $1$.  In
particular, the satisfaction of the Muller condition in
$\calT \ltimes \calM$, inherited from $\calF$, can be characterized by
the BSCCs satisfying the Muller condition $\calF$ (in a sense that we
will make precise).

\begin{definition}[Good BSCC]
  A BSCC $C$ of $\mathsf{Graph}_{\cal T \ltimes \calM}(B)$ is
  \emph{good for $\calF$}, written $C \in \mathsf{Good}^B_{\calT
    \ltimes \calM}(\calF)$, 
  if there exists $F \in \calF$ such that
  \begin{enumerate}
  \item[(a)] for every state $(s,q) \in S \times Q$, if there exists $(r,p) \in C$
    with a path from $(r,p)$ to $(s,q)$ in $\calT \ltimes \calM$, then
    $q \in F$; and
  \item[(b)] for every $q\in F$ there exists $s\in S$, there exists a
    state $(r,p) \in C$ with a path from $(r,p)$ to $(s,q)$ in $\calT
    \ltimes \calM$.
\end{enumerate}
\end{definition}

Let $C$ be an arbitrary BSCC of $\mathsf{Graph}_{\calT \ltimes
  \calM}(B)$. We define the set $F_C = \{q \in Q \mid \exists s \in
S,\ \exists (r,p) \in C\ \text{s.t. there is a path from}\ (r,p)\
\text{to}\ (s,q)\}$ as the set of states of the Muller automaton that
can be reached from $C$.  Within a BSCC, all reachable states will
actually be visited infinitely often almost-surely. More precisely, we
state the following result:

\begin{lemma}
  \label{lemma:BSCCproba}
  For every $(s,q) \in C$, $\Prob_{\delta_{(s,q)}}^{\calT \ltimes
    \calM}(\mathsf{Inf} =F_C) =1$.\footnote{We recall that
      $\mathsf{Inf} =F_C$ characterizes the set of runs $\rho'$ in
      $\calT \ltimes \calM$ such that $\Inf(\calL'(\rho')) = F_C$
      ($\calL'$ is the labelling function of $\calT \ltimes \calM$
      such that $\calL'(s,q)=q$).}
\end{lemma}

\begin{proof}
  Let $(s,q) \in C$, and $\rho = (s,q) (s_1,q_1) (s_2,q_2) \dots$ a
  path in $\calT \ltimes \calM$ starting at $(s,q)$. By definition of
  $F_C$, all $q_i$'s are in $F_C$, hence
  $\Prob_{\delta_{(s,q)}}^{\calT \ltimes \calM}(\mathsf{Inf} \subseteq
  F_C) =1$.  

  \medskip We now argue why all elements of $F_C$ are actually
  almost-surely visited infinitely often.
  Fix $p \in F_C$ and $(r,p)$ that is reachable from $C$. All two
  states of $C$ are reachable one from each other; thus, from every state of $C$, $(r,p)$
  is reachable through a finite path. Hence there is some $\iota>0$
  and $k \in \IN$ such that for every state $(s',q') \in C$, 
  \[
  \Prob_{\delta_{(s',q')}}^{\calT \ltimes \calM}( \F[\le k] (r,p)) \ge
  \iota\enspace.
  \]
  Applying a reasoning similar to the proof of
  Proposition~\ref{prop:attractorSound}, we get that
  $\Prob_{\delta_{(s,q)}}^{\calT \ltimes \calM}(\G \F (r,p) \mid \G \F
  C) = 1$. Indeed,
  $\Prob_{\delta_{(s,q)}}^{\calT \ltimes \calM}(\F \G \neg (r,p)
  \wedge \G \F C) \leq \lim_{n\to\infty} (1-\iota)^n= 0$. Thanks to
  the third item of Lemma~\ref{lemma:bscc}, we obtain that
  \[
  \Prob_{\delta_{(s,q)}}^{\calT \ltimes \calM}(\G \F (r,p)) =
  1\enspace.
  \]
  We conclude that $\Prob_{\delta_{(s,q)}}^{\calT \ltimes
    \calM}(\mathsf{Inf} \supseteq F_C) =1$, which completes the proof.
\end{proof}

As a consequence:

\begin{corollary} \label{coro:BSCCproba} For every initial
  distribution $\mu\in\Dist(S)$ for $\calT$ and for every $q\in Q$,
  $\Prob_{\mu \times \delta_{q_0}}^{\calT \ltimes \calM}(\F C) > 0$
  implies \( \Prob_{\mu \times \delta_{q_0}}^{\calT \ltimes
    \calM}(\mathsf{Inf} =F_C \mid \F C) =1.  \)
\end{corollary}

We can now completely characterize the probability of satisfying an
$\omega$-regular property.

\begin{theorem}
  \label{th:good-bscc-for-quant-analysis}
  Let $\calT$ be a labelled DMC with a finite attractor, and $\calM =
  (Q,q_0,E,\calF)$ be a DMA.  Then, for every initial distribution
  $\mu\in\Dist(S)$ for $\calT$:
  \[
  \Prob_{\mu \times \delta_{q_0}}^{\calT \ltimes \calM}(\mathsf{Inf} \in \calF)
  = \sum_{C \in \mathsf{Good}^B_{\calT
          \ltimes \calM}(\calF)} 
    \Prob_{\mu \times \delta_{q_0}}^{\calT \ltimes \calM}(\F C )
  \]
  where $B$ is an attractor for $\calT \ltimes \calM$.
\end{theorem}

\begin{proof}
  Fix $\mu\in\Dist(S)$ and $q\in Q$. As stated in Lemma~\ref{lemma:bscc}, the BSCCs of
  $\mathsf{Graph}(B)$ form an attractor for $\calT \ltimes \calM$, and
  two BSCCs are probabilistically disjoint.  Using Bayes formula with
  a disjunction over the BSCCs, we can write:
  \[
  \Prob_{\mu \times \delta_{q_0}}^{\calT \ltimes \calM}(\mathsf{Inf} \in \calF)
  = \sum_{\begin{array}{c} {\scriptstyle C\ \text{BSCC of}\
        \mathsf{Graph}(B)} \\[-.2cm]  {\scriptstyle C\
        \mu \times \delta_{q_0}\text{-reachable}} \end{array}}
  \Prob_{\mu \times \delta_{q_0}}^{\calT \ltimes \calM}(\F C) \cdot
  \Prob_{\mu \times \delta_{q_0}}^{\calT \ltimes \calM}(\mathsf{Inf} \in \calF
  \mid \F C)
  \]
  where we say that $C$ is $\mu \times \delta_{q_0}$-reachable
  whenever $\Prob_{\mu \times \delta_{q_0}}^{\calT \ltimes \calM}(\F
  C)>0$. Hence we deduce that:
  \[
  \Prob_{\mu \times \delta_{q_0}}^{\calT \ltimes \calM}(\mathsf{Inf} \in \calF)
  = \sum_{C\ \text{BSCC of}\ \mathsf{Graph}(B)}
  \Prob_{\mu \times \delta_{q_0}}^{\calT \ltimes \calM}(\F C) \cdot
  \mathds{1}_{\calF}(F_C)
  \]
  thanks to Corollary~\ref{coro:BSCCproba}, where $\mathds{1}_\calF$
  is the characteristic function of $\calF$ (that is,
  $\mathds{1}_{\calF}(F)=1$ if $F \in \calF$, and
  $\mathds{1}_{\calF}(F)=0$ otherwise).  This concludes the proof of
  the theorem.
\end{proof}

\subsection{General STSs via an abstraction}
\label{subsec:tototo}

While the previous approach is adapted to DMCs, it does not apply
directly to general STSs: indeed, it is unlikely that general STSs
have \emph{finite} attractors, and finiteness of the attractor is
fundamental for the correctness of the approach.
The idea will then be to rely on an abstraction that admits a
  finite attractor, and to transfer properties through that
abstraction.

Let $\calT_1 = (S_1,\Sigma_1,\kappa_1,\AP,\calL_1)$ and $\calT_2 =
(S_2,\Sigma_2,\kappa_2,\AP,\calL_2)$ be two LSTSs such that $\calT_2$
is a DMC, which is an $\alpha$-abstraction of $\calT_1$.  Under
certain conditions, we show how to express the probability of
satisfying an $\omega$-regular property represented by a DMA $\calM =
(Q,q_0,E,\calF)$ in $\calT_1$ using the abstraction $\calT_2$.
We consider both the product $\calT_1 \ltimes \calM$
and the product $\calT_2 \ltimes \calM$.

First we justify why, within a slight abuse of terminology, $\calT_2
\ltimes \calM$ can be viewed as an $\alpha$-abstraction of $\calT_1
\ltimes \calM$. We also exhibit a sufficient condition under which it
is sound.

\begin{restatable}{lemma}{alphabar}
  \label{lemma:alphabar}
  Let $\alpha_\calM : S_1 \times Q \to S_2 \times Q$ be the lifting of
  $\alpha$ such that $\alpha_\calM(s,q) = (\alpha(s),q)$. If $\calT_2$
  is an $\alpha$-abstraction of $\calT_1$, then $\calT_2 \ltimes
  \calM$ is an $\alpha_\calM$-abstraction of $\calT_1 \ltimes
  \calM$. Furthermore, if $\calT_1\ltimes \calM$ is $\D(\calB)$ where
  $\calB=\lbrace \alpha_\calM^{-1}(B)\mid B\in\Sigma'_2\rbrace$, then
  $\calT_2 \ltimes \calM$ is a sound $\alpha_\calM$-abstraction of
  $\calT_1 \ltimes \calM$.
\end{restatable}

While the proof of the first part of the lemma is technical hence
postponed to the appendix (page~\pageref{app-alphabar}), the second
part of the lemma is a consequence of Proposition~\ref{coro:DecSound}.

\begin{rk}
  \label{rk:produit}
  In the sequel, our applications will be smooth enough to meet the
  hypothesis: $\calT_1\ltimes\calM$ is decisive
  w.r.t. $\alpha_\calM$-closed sets. However we still have several
  open questions. The first one is the following: does soundness
  between $\calT_2$ and $\calT_1$ imply soundness between
  $\calT_2\ltimes\calM$ and $\calT_1\ltimes\calM$?  While this seems
  quite natural, it is surprisingly tricky. Although we did not manage
  to find a counter-example for this general question, we found one
  for a fixed initial distribution. It is described in
  Example~\ref{ex:cex-sound} in the appendix (page~\pageref{app:cex})
  and highlights some difficulties we encounter when aiming at
  transferring analysis from the abstraction to the concrete model.

  This justifies the fact that we assumed decisiveness. As we already
  know, if $\calT_2$ is a sound $\alpha$-abstraction of $\calT_1$ and
  $\calT_2$ is decisive w.r.t. any set of states, then $\calT_1$ is
  decisive w.r.t. any $\alpha$-closed sets. Then the second natural
  question is the following: does decisiveness w.r.t. $\alpha$-closed
  sets for $\calT_1$ imply decisiveness w.r.t. $\alpha_\calM$-closed
  sets for $\calT_1\ltimes\calM$? Again, we do not have a general
  counter-example, but we have one for a fixed initial
  distribution. This is described in Example~\ref{ex:cex-sound} in the
  appendix (page~\pageref{app:cex}).
\end{rk}

From now on, whenever $\calT_1\ltimes\calM$ is decisive
w.r.t. $\alpha_\calM$-closed sets and thus the previous result is
applicable, we will abusively write $\alpha$ for $\alpha_\calM$.

We focus now on the case where $\calT_2$ has a finite
attractor.\footnote{As $\calT_2$ has a finite attractor, it is
  decisive and thus $\calT_2$ is a complete $\alpha$-abstraction of
  $\calT_1$ by Lemma~\ref{lemma:completeness}.} Applying
Lemma~\ref{lemma:soundproduct}, $\calT_2 \ltimes \calM$ has also a
finite attractor, which we denote $B_2$. We reuse notations of the
previous subsection, in particular the graph of the attractor
$\mathsf{Graph}_{\calT_2 \ltimes \calM}(B_2)$, and the set $F_C$ of
recurring states when $C$ is a BSCC of that graph.

The following lemma is a counterpart to Lemma~\ref{lemma:bscc} for
$\calT_1$. Under the hypothesis that $\calT_1\ltimes\calM$ is decisive
w.r.t. $\alpha$-closed sets, even though $\calT_1 \ltimes \calM$ does
not have a finite attractor, it has an attractor with an interesting
structure inherited from $\calT_2 \ltimes \calM$. In the sequel, we
write $\calB=\lbrace \alpha^{-1}(B)\mid B\in\Sigma'_2\rbrace$.
\begin{lemma}
  \label{lemma:bscc2}
  Assume $\calT_2$ has a finite attractor, and assume that $\calT_2
  \ltimes \calM$ is a sound $\alpha$-abstraction of $\calT_1 \ltimes
  \calM$.  Write $B_2$ for an attractor of $\calT_2 \ltimes \calM$.
  The following properties are
  satisfied:
  \begin{itemize}
  \item The set $\alpha^{-1}(\{(s,q) \in C \mid C\ \text{BSCC of}\
    \mathsf{Graph}_{\calT_2 \ltimes \calM}(B_2)\})$ is an attractor of
    $\calT_1 \ltimes \calM$. 
  \item If $C$ and $C'$ are two distinct BSCCs of
    $\mathsf{Graph}_{\calT_2 \ltimes \calM}(B_2)$, for every $\mu \in
    \Dist(S_1\times Q)$, $\Prob_{\mu}^{\calT_1 \ltimes \calM}(\F
    \alpha^{-1}(C) \wedge \F \alpha^{-1}(C')) = 0$.
  \item If $C$ is a BSCC of $\mathsf{Graph}_{\calT_2 \ltimes
      \calM}(B_2)$, for every $\mu \in \Dist(\alpha^{-1}(C))$,
    $\Prob_{\mu}^{\calT_1 \ltimes \calM}(\G \F \alpha^{-1}(C)) = 1$.
  \end{itemize}
\end{lemma}

\begin{proof}
  Since $\calT_2 \ltimes \calM$ is a sound $\alpha$-abstraction of
  $\calT_1 \ltimes \calM$, 
  the first property is derived from Proposition~\ref{lem:attr-via-sound} and
  Lemma~\ref{lemma:bscc}. The second property is a consequence of
  Lemma~\ref{lemma:bscc}, and of the fact that $\calT_2 \ltimes \calM$
  is an $\alpha$-abstraction of $\calT_1 \ltimes \calM$.  Finally, the
  third property is, as in the proof of Lemma~\ref{lemma:bscc}, a
  consequence of the second point and of
  Lemma~\ref{lemma:attractorGF}.
\end{proof}

We then prove a counterpart to Lemma~\ref{lemma:BSCCproba} for
$\calT_1$, which shows that a BSCC is characterized by the set $F_C$
of states that are visited infinitely often from $C$.

\begin{lemma}
  \label{lemma:bsccT1}
  Assume $\calT_2$ has a finite attractor, and assume that
  $\calT_2 \ltimes \calM$ is a sound $\alpha$-abstraction of
  $\calT_1 \ltimes \calM$.  Let $C$ be a BSCC of
  $\mathsf{Graph}_{\calT_2 \ltimes \calM}(B_2)$, and
  $\mu \in \Dist(\alpha^{-1}(C))$. Then:
  \[
  \Prob_{\mu}^{\calT_1 \ltimes \calM} (\mathsf{Inf} = F_C) = 1.
  \]
\end{lemma}

\begin{proof}
  
  As already argued in the proof of Lemma~\ref{lemma:BSCCproba}, for
  every $p \in F_C$, for every state $\mathbf{s}_2 \in C$,
  $\Prob_{\delta_{\mathbf{s}_2}}^{\calT_2 \ltimes \calM}(\F p)=1$ (we
  abusively write $p$ for the measurable set $S_2 \times \{p\}$).
  Since $\calT_2 \ltimes \calM$ is a sound $\alpha$-abstraction of
  $\calT_1 \ltimes \calM$, we derive for every
  $\nu \in \Dist(\alpha^{-1}(C))$ that
  $\Prob_{\nu}^{\calT_1 \ltimes \calM}(\F p)=1$ (as before we
  abusively write $p$ for
  $S_1 \times \{p\} = \alpha^{-1}(S_2 \times \{p\})$). We can then
  show that for each $\nu\in\Dist(\alpha^{-1}(C))$ and for each
  $p\in F_C$,
  \[
  \Prob_\nu^{\calT_1 \ltimes \calM}(\G \F p) = 1.
  \]
  Indeed, towards a contradiction, assume that there is a distribution
  $\nu\in\Dist(\alpha^{-1}(C))$ such that $\Prob_\nu^{\calT_1 \ltimes
    \calM}(\G \F p) < 1$, i.e. $\Prob_\nu^{\calT_1 \ltimes \calM}(\F
  \G \neg p) > 0$. From the third point of Lemma~\ref{lemma:bscc2}, we
  get that $\Prob_\nu^{\calT_1 \ltimes \calM}(\G\F \alpha^{-1}(C)
  \wedge \F \G \neg p) > 0$. Now, observe that
\[\ev{\calT_1\ltimes\calM}{\G\F \alpha^{-1}(C)\wedge\F\G\neg p} \subseteq \ev{\calT_1\ltimes \calM}{\bigcup_{n\in\IN}\big(\F[=n] \alpha^{-1}(C)\wedge\G[\ge n] \neg p\big)}. \]
It follows that there is $n\in\IN$ such that $\Prob_\nu^{\calT_1
  \ltimes \calM}(\F[=n]\alpha^{-1}(C)\wedge\G[\ge n]\neg p)>0$. From
Lemma~\ref{lemma:integration}, we get that there is
$\nu'\in\Dist(S'_1)$ (with $S'_1=S_1\times Q$) such that
\begin{multline*}
  \Prob_\nu^{\calT_1 \ltimes \calM}(\F[=n]\alpha^{-1}(C)\wedge\G[\ge
    n]\neg p) \\
  \begin{array}{cl}
    = & {\displaystyle \lim_{m\to\infty} \Prob_\nu^{\calT_1 \ltimes
      \calM}(\Cyl(\overbrace{S'_1,\ldots,S'_1}^{n\text{ times}},
    \alpha^{-1}(C)\wedge\neg p, \overbrace{\neg p,\ldots,\neg p}^{m
      \text{ times}}))} \\
    \le  & \lim_{m\to\infty} \Prob_{\nu'}^{\calT_1 \ltimes
      \calM}(\Cyl(\alpha^{-1}(C)\wedge\neg p, \overbrace{\neg
      p,\ldots,\neg p}^{m \text{ times}})) \quad \text{from
      Lemma~\ref{lemma:integration}} \\
    = & \lim_{m\to\infty} \nu'(\alpha^{-1}(C))\cdot
    \Prob_{\nu'_{\alpha^{-1}(C)}}^{\calT_1 \ltimes \calM}(\Cyl(\neg p,
    \overbrace{\neg p,\ldots,\neg p}^{m \text{ times}})) \\
    = & \nu'(\alpha^{-1}(C))\cdot \Prob_{\nu'_{\alpha^{-1}(C)}}^{\calT_1 \ltimes \calM}(\G \neg p).
  \end{array}
\end{multline*}
From the assumption, we thus get that $\Prob_{\nu'_{\alpha^{-1}(C)}}^{\calT_1 \ltimes \calM}(\G \neg p)>0$ where $\nu'_{\alpha^{-1}(C)}\in\Dist(\alpha^{-1}(C))$ which is the required contradiction. Hence, for each $\nu\in\Dist(\alpha^{-1}(C))$ and for each $p\in F_C$, $\Prob_{\nu}^{\calT_1\ltimes\calM}(\G\F p)=1$.

It now suffices to show that, from any
$\nu \in \Dist(\alpha^{-1}(C))$, no other state is visited
almost-surely infinitely often.  Fix $p' \notin F_C$. Then, by
definition of $F_C$, we have that
$\Prob_{\alpha_{\#}(\nu)}^{\calT_2 \ltimes \calM}(\F p') = 0$. Since
$\calT_2 \ltimes \calM$ is an $\alpha$-abstraction of
$\calT_1 \ltimes \calM$, we deduce that
$\Prob_{\nu}^{\calT_1 \ltimes \calM}(\F p') =0$.
  
  We conclude that $\Prob_\nu^{\calT_1 \ltimes
    \calM}(\mathsf{Inf}=F_C) = 1$, which is the claim of the lemma.
\end{proof}

We are now in a position to decompose the probability to satisfy the
Muller condition $\calF$ in $\calT_1 \ltimes \calM$ into the
reachability probability of good BSCCs in $\mathsf{Graph}_{\calT_2
  \ltimes \calM}(B_2)$.
\begin{theorem}
  \label{theo:titi}
  Let $\calT_1$ and $\calT_2$ be two LSTSs such that $\calT_2$ is a
  DMC with a finite attractor, and $\calT_2$ is an
  $\alpha$-abstraction of $\calT_1$. Let $\calM = (Q,q_0,E,\calF)$ be
  a DMA. Assume moreover that $\calT_2\ltimes\calM$ is an
  $\alpha$-sound abstraction of $\calT_1\ltimes\calM$, and that $B_2$
  is a finite attractor of $\calT \ltimes \calM$.
  Then, for every initial
  distribution $\mu$ for $\calT_1$:
  \[
  \Prob_{\mu \times \delta_{q_0}}^{\calT_1 \ltimes \calM}(\mathsf{Inf} \in
  \calF) = \sum_{C \in \mathsf{Good}^{B_2}_{\calT_2 \ltimes \calM}(\calF)}
  \Prob_{\mu \times \delta_{q_0}}^{\calT_1 \ltimes \calM}(\F \alpha^{-1}(C)
  )\enspace.
  \]
\end{theorem}

\begin{proof}

  Applying Lemma~\ref{lemma:bsccT1}, for every $\mu \in \Dist(S_1)$,
  assuming $\Prob_{\mu \times \delta_{q_0}}^{\calT_1 \ltimes \calM}(\F
  \alpha^{-1}(C))>0$, then \( \Prob_{\mu \times \delta_{q_0}}^{\calT_1
    \ltimes \calM}(\mathsf{Inf} = F_C \mid \F \alpha^{-1}(C)) = 1.  \)
  By the two first properties of Lemma~\ref{lemma:bscc2}, we can write
  the following Bayes formula, with a disjunction over the BSCCs of
  $\mathsf{Graph}_{\calT_2 \ltimes \calM}(B_2)$:
  \begin{eqnarray*}
    \Prob_{\mu \times \delta_{q_0}}^{\calT_1 \ltimes \calM}(\mathsf{Inf} \in \calF) & = &
    \hspace*{-1.5cm} \sum_{\begin{array}{c} {\scriptstyle C\ \text{BSCC of}\ \mathsf{Graph}_{\calT_2 \ltimes
        \calM}(B_2)} \\[-.2cm] {\scriptstyle C\ \mu \times
      \delta_q\text{-reachable}} \end{array} } \hspace*{-1.5cm}
\Prob_{\mu \times \delta_{q_0}}^{\calT_1 \ltimes \calM}(\F 
    \alpha^{-1}(C)) \cdot \Prob_{\mu \times \delta_{q_0}}^{\calT_1 \ltimes
      \calM}(\mathsf{Inf} \in \calF\mid \F \alpha^{-1}(C)) \\
    & = & \sum_{C\ \text{BSCC of}\ \mathsf{Graph}_{\calT_2 \ltimes
        \calM}(B_2)} \Prob_{\mu \times \delta_{q_0}}^{\calT_1 \ltimes \calM}(\F
    \alpha^{-1}(C)) \cdot \mathds{1}_{\calF}(F_C) \\
    & = & \sum_{C\in \mathsf{Good}^{B_2}_{\calT_2 \ltimes
        \calM}(\calF)} \Prob_{\mu \times \delta_{q_0}}^{\calT_1
      \ltimes \calM}(\F 
    \alpha^{-1}(C))
  \end{eqnarray*}
  This concludes the proof of the theorem.
\end{proof}

\section{Qualitative analysis}
\label{sec:qualitative}

In this section, we rely on the notions previously introduced and
studied to design generic procedures for the qualitative analysis of
properties of STSs, under some assumptions that will be made precise.
We emphasize that these are procedures rather than algorithms, since
algorithms would require some effectiveness conditions on the STSs
(numerical conditions, or decidability of some graph properties in the
underlying non-stochastic model). Next, we will make explicit
necessary conditions to obtain algorithms from the generic procedures.
For most natural STSs (and in particular for our applications~--~see
Section~\ref{sec:appli}), these conditions will be satisfied.

For the next two subsections, we fix an STS $\calT =
(S,\Sigma,\kappa)$.

\subsection{Basic properties under decisiveness hypotheses}
\label{subsec:qual-basic}

Our objective here is to describe generic procedures that capture the
qualitative (almost-sure and positive) satisfaction of reachability
and repeated reachability properties.

Given $B \in \Sigma$ a measurable set, recall that
$\widetilde{B} = \{s \in S \mid \Prob_{\delta_s}^{\calT}(\F B) = 0\}$
denotes its avoid-set. Some properties of this set, while not crucial
for the understanding but required for the proofs, are given in
Appendix~\ref{app:add-qual} (page~\pageref{app:add-qual}).

\medskip Extending the approach of~\cite{ABM07}, we establish
characterizations of the qualitative satisfaction of (repeated)
reachability properties in terms of the positive satisfaction of
reachability-like properties.  We advocate that these are simpler to
check on STSs: positive reachability
amounts to guessing a ``symbolic'' path (or cylinder) leading to the
target, and to showing that this path has a positive measure.  The
next proposition is stated in greater details as
Proposition~\ref{app-qualsimple} in the technical
Appendix~\ref{appendix:qualitative} (page~\pageref{app-qualsimple}).

\begin{proposition}
  \label{prop:qualsimple}
  Let $\mu \in \Dist(S)$. Then we have the following implications,
  yielding various characterizations for the qualitative analysis of
  STSs (under specified assumptions):
  \begin{description}
  \item[Almost-sure reachability] If $\calT$ is $\D(\mu,B)$, then:
    \[
      \Prob^\calT_\mu(\F B) = 1\iff \Prob^\calT_\mu(\neg B \U
      \widetilde{B}) =0 \enspace.
    \]
  \item[Almost-sure repeated reachability] If $\calT$ is $\SD(\mu,B)$,
    then:
    \[
    \Prob^\calT_\mu(\G \F B) = 1 \iff \Prob^\calT_\mu(\F
    \widetilde{B}) =0 \enspace.
    \]
  \item[Positive repeated reachability] If $\calT$ is
    $\D(\mu,\widetilde{B})$ and $\PD(\mu,B)$, then:
    \[
    \Prob^\calT_\mu(\G\F B) >0 \iff \Prob^\calT_\mu(\F
    \widetilde{\widetilde{B}})>0 \enspace.
    \]
  \end{description}
\end{proposition}
While the two first characterizations are quite intuitive under the
corresponding decisiveness assumptions, let us comment on the
characterization of positive repeated reachability: the set
$\widetilde{\widetilde{B}}$ is the set from which one cannot reach
$\widetilde{B}$ (or with probability $0$), hence from which we will be
able to revisit $B$ again and again. With this interpretation in mind,
the characterization is somewhat natural.

\medskip\noindent This reduces all these problems to checking the
(non-)positivity of some reachability, or a slight generalization
thereof, property in the STS.
Those are the simplest properties one can hope to be decidable in a
class of models. Effectiveness hence relies here on the computation of
avoid-sets, avoid-sets of avoid-sets, and on the decidability of the
positive reachability (or Until) problem.

\subsection{Basic properties through abstractions}
\label{subsec:qual-basic-abs}

Via abstractions, one can reduce the qualitative analysis of basic
properties (reachability and repeated reachability) from the concrete
model to the abstract model. Indeed, one can use the previous results
(Propositions~\ref{prop:links} and~\ref{thm:MuDecisiveAbstr} together
with Proposition~\ref{prop:qualsimple}), and show:

\begin{proposition}
  Assume $\calT_2$ is an $\alpha$-abstraction of $\calT_1$, and fix $B
  \in \Sigma_2$. 
  \begin{itemize}
  \item Let $\mu \in \Dist(S_1)$ be an initial distribution for
    $\calT_1$.  Assume that $\calT_2$ is $\mu$-sound and
    $\mu$-complete. Then:
    \[
    \Prob^{\calT_1}_{\mu}(\F \alpha^{-1}(B)) =1\ \text{iff}\
    \Prob^{\calT_2}_{\alpha_{\#}(\mu)}(\F B) =1\enspace.
    \]
  \item Assume that $\calT_2$ is sound and complete, and that
    $\calT_2$ is $\SD(B)$. Then for every $\mu \in \Dist(S_1)$:
    \[
    \Prob^{\calT_1}_{\mu}(\G \F \alpha^{-1}(B)) =1\ \text{iff}\
    \Prob^{\calT_2}_{\alpha_{\#}(\mu)}(\G \F B) =1\enspace.
    \]
  \item Assume that $\calT_2$ is sound and complete, and that
    $\calT_2$ is $\PD(B)$ and $\D(\widetilde{B})$. Then for
    every $\mu \in \Dist(S_1)$:
    \[
    \Prob^{\calT_1}_{\mu}(\G \F \alpha^{-1}(B)) =1\ \text{iff}\
    \Prob^{\calT_2}_{\alpha_{\#}(\mu)}(\G \F B) =1\enspace.
    \]
  \end{itemize}
\end{proposition}
This allows one to perform the qualitative analysis of (repeated)
reachability properties in $\calT_1$ on its abstraction $\calT_2$,
which is quite useful since $\calT_2$ will usually be simpler than
$\calT_1$.

\subsection{$\omega$-regular properties in DMCs with a finite
  attractor}
\label{subsec:qual-dmc}

Following Section~\ref{subsec:dmc-attractor},
under the assumption that the STS has a finite attractor, we have
completely characterized the probability of satisfying the property
defined by a DMA using the probability of reaching BSCCs of a finite
graph (Theorem~\ref{th:good-bscc-for-quant-analysis}). Using that
result, we get the following characterization of the almost-sure
satisfaction relation.
\begin{corollary}[Almost-sure $\omega$-regular property]
  \label{coro:ASomega}
  Let $\calT$ be a labelled DMC with a finite attractor, and $\calM =
  (Q,q_0,E,\calF)$ be a DMA.  Let $B$ be a finite attractor for $\calT
  \ltimes \calM$. For every initial distribution $\mu\in\Dist(S)$ for
  $\calT$:
  \begin{multline*}
    \Prob_{\mu \times \delta_{q_0}}^{\calT \ltimes \calM}(\mathsf{Inf} \in
    \calF) = 1\quad \text{if and only if} \\ \text{every BSCC $C$ of
      $\mathsf{Graph}_{\calT \ltimes \calM}(B)$ such that
      $\Prob_{\mu \times \delta_{q_0}}^{\calT \ltimes \calM}(\F C) >0$ is good
      for $\calF$.}
  \end{multline*}
\end{corollary}

\noindent In order to turn this characterization into a decision
procedure, we need to be able to compute the attractor $B$ for $\calT
\ltimes \calM$, and to build the graph $\mathsf{Graph}_{\calT \ltimes
  \calM}(B)$; also one needs to be able to compute for every BSCC $C$
the set $F_C$.

\subsection{$\omega$-regular properties of general STSs via
  abstraction and finite attractor}
\label{subsec:qual-abs}
Following Section~\ref{subsec:tototo}, under several assumptions over
an abstraction, we have completely characterized the probability for a
concrete system to satisfy a property given as a DMA using a
decomposition of a graph defined for the abstract system
(Theorem~\ref{theo:titi}). From that result, we deduce the following
characterization of the almost-sure satisfaction relation via an
abstraction. It turns out that the value resulting from the
decomposition is equal to $1$ if, and only if, the property is
almost-surely satisfied by the abstract system.

\begin{corollary}\label{coro:theotiti}
  Let $\calT_1$ and $\calT_2$ be two LSTSs such that $\calT_2$ is a
  DMC with a finite attractor, and $\calT_2$ is an
  $\alpha$-abstraction of $\calT_1$. Let $\calM =
  (Q,q_0,E,\calF)$ be a DMA.  Assume moreover that
  $\calT_2\ltimes\calM$ is an $\alpha$-sound abstraction of
  $\calT_1\ltimes\calM$.
  Then, for every initial distribution $\mu$ for $\calT_1$:
  \[
  \Prob_{\mu \times \delta_{q_0}}^{\calT_1 \ltimes \calM}(\mathsf{Inf}
  \in \calF) = 1\quad \text{if and only if}\quad
  \Prob_{\alpha_{\#}(\mu \times \delta_{q_0})}^{\calT_2 \ltimes
    \calM}(\mathsf{Inf} \in \calF) = 1 \enspace.
  \]
\end{corollary}
\noindent Hence, this reduces the almost-sure model-checking of a
property given by $\calM$ in $\calT_1$ to the almost-sure
model-checking of a reachability property (applying
Corollary~\ref{coro:ASomega}). For the approach to be effective, it is
sufficient that the analysis at the level of $\calT_2 \ltimes \calM$
is effective.

As already quickly mentioned, under the hypotheses of
Corollary~\ref{coro:theotiti}, the abstraction $\calT_2 \ltimes \calM$
is complete (since it has a finite attractor). Though it is not
explicitely used, we could not have such an equivalence without some
completeness of the abstraction.

\begin{rk}[Discussion on the approach of~\cite{BBB+14}] While the
  notion of abstraction was not precisely defined in~\cite{BBB+14} for
  stochastic timed automata, it was implicitly already there. Also,
  decidability of the almost-sure satisfaction was ensured thanks to a
  fairness condition. Using the terminology of the current paper, the
  framework was the following: $\calT_1$ and $\calT_2$ are two STSs
  such that $\calT_2$ is a \emph{finite} Markov chain which is an
  $\alpha$-abstraction of $\calT_1$. Then the condition for the
  abstraction to yield interesting results was that $\calT_1$ should
  be fair w.r.t. every $\alpha$-closed sets (the latter condition
  implying the fairness of $\calT_1\ltimes\calM$, for $\calM$ a DMA).
  Thanks to Proposition~\ref{prop:fairness}, this implies that
  $\calT_1 \ltimes \calM$ is actually decisive w.r.t. $\alpha$-closed
  sets.  Applying Proposition~\ref{coro:DecSound}, we get that
  $\calT_2 \ltimes \calM$ is sound abstraction of
  $\calT_1 \ltimes \calM$.  Given that $\calT_2$ is finite, it
  trivially has a finite attractor. Hence, the conditions of
  Theorem~\ref{theo:titi} are satisfied, and the approach
  of~\cite{BBB+14} was then a particular case of that theorem, when
  applied to specific subclasses of stochastic timed automata (further
  details are provided in Subsection~\ref{subsec:sta}).
\end{rk}

\section{Approximate quantitative analysis}
\label{sec:quantitative}

Beyond qualitative analysis, we are interested in quantitative
analysis of stochastic systems, that is, in computing the probability
that an STS satisfies a given property. No generic decidability
results can be stated in the very general context of
STSs,\footnote{Note already that there is no uniform effective way to
  represent STSs, so that we can hardly expect generic (and effective)
  procedures or algorithms.}. We thus focus here on approximate
analysis and develop generic approximation procedures which, under
reasonable assumptions, allow one to compute within arbitrary
precision, the probability of a given property.  As for properties, we
consider first reachability, then repeated reachability, later
$\omega$-regular properties, and finally some timed properties.

\medskip
For the next two subsections, we fix an STS $\calT =
(S,\Sigma,\kappa)$, and a
distribution $\mu \in \Dist(S)$.

\subsection{Quantitative reachability under decisiveness hypotheses}
\label{subsec:approx-reach}
In order to approximate the reachability probability of a set $B \in
\Sigma$ in $\calT$, we define the two following sequences, similar to
the ones given for decisive Markov chains~\cite{ABM07}. For every $n
\in \nats$:
\[
\left\{\begin{array}{lcl}
p_n^{\mathsf{Yes}} & = & \Prob^\calT_\mu(\F[\le n] B);\\
p_n^{\mathsf{No}} & = & \Prob^\calT_\mu(\neg B \U[\le n] \widetilde{B}).
\end{array}\right.
\]

Since the sequences of events $(\F[\le n] B)_{n \in \nats}$ and $(\neg
B \U[\le n] \widetilde{B})_{n \in \nats}$ are non-decreasing and
converge respectively to $\F B$ and $\neg B \U \widetilde{B}$, 
the sequences $(p_n^{\mathsf{Yes}})_n$ and $(p_n^{\mathsf{No}})_n$
  are non-decreasing and converge respectively to $\Prob^\calT_\mu(\F
  B)$ and $\Prob^\calT_\mu(\neg B \U \widetilde{B})$.
%
  Assuming now that $\calT$ is decisive w.r.t. $B$, the two limits are
  related, as stated below. The proof of this proposition can be found
  in Appendix, page~\pageref{app:approxreach}.
\begin{restatable}[Approximation scheme for reachability properties]{proposition}{approxreach}
\label{prop-approx-reach}
If $\calT$ is $\D(\mu,B)$, then the two sequences
$(p_n^{\mathsf{Yes}})_n$ and $(1-p_n^{\mathsf{No}})_n$ are
adjacent\footnote{Recall that two sequences $(a_n)_{n \in \IN}$
    and $(b_n)_{n \in \IN}$ are said \emph{adjacent} if w.l.o.g. $(a_n)$
    is non-decreasing, $(b_n)$ is non-increasing and the sequence
    $(a_n-b_n)_{n \in \IN}$ converges to $0$.} and converge to
$\Prob^\calT_\mu(\F B)$.
\end{restatable}

To obtain an $\varepsilon$-approximation for $\Prob^\calT_\mu(\F B)$,
it suffices to evaluate $p_n^{\mathsf{Yes}} $ and $p_n^{\mathsf{No}} $
for larger and larger values of $n$, until $1 - p_n^{\mathsf{No}} -
p_n^{\mathsf{Yes}} < \varepsilon$, and to return $p_n^{\mathsf{Yes}}
$. This scheme is effective as soon as one can compute
$\widetilde{B}$, and the probability (from $\mu$) of cylinders of the
forms $\Cyl(\underbrace{S,\ldots,S}_{n\text{ times}},B)$ and
$\Cyl(\underbrace{\neg B,\ldots,\neg B}_{n\text{
    times}},\widetilde{B})$.  In case $p_n^{\mathsf{Yes}}$ and
$p_n^{\mathsf{No}}$ cannot be computed exactly, but can only be
approximated up to any desired error bound, this scheme can be refined
to obtain a $2 \varepsilon$-approximation for $\Prob^\calT_\mu(\F B)$.

\begin{remark}
  The above approximation scheme can be adapted to Until properties of
  the form $B' \U B$ (for $B,B' \in \Sigma$) in a straightforward way
  as follows: for every $n \in \bbN$,
  \[
  \left\{\begin{array}{lcl}
      \hat{p}_n^{\mathsf{Yes}} & = & \Prob^\calT_\mu(B' \U[\le n] B);\\
      \hat{p}_n^{\mathsf{No}} & = & \Prob^\calT_\mu(\neg B \U[\le n]
      (\widetilde{B} \vee \neg B')).
    \end{array}\right.
  \]
  Convergence of that scheme here also relies on a decisiveness
  property w.r.t. $B$.
\end{remark}

\subsection{Quantitative repeated reachability under decisiveness hypotheses}
\label{subsec:quant_repeated_reach}
We now define two sequences that will yield an approximation scheme
for a repeated reachability probability, under stronger assumptions
on the model. For every $n \in \IN$:
\[
\left\{\begin{array}{lcl}
q_n^{\mathsf{Yes}} & = & \Prob^\calT_\mu(\F[\le n] \widetilde{\widetilde{B}}); \\
q_n^{\mathsf{No}} & = & \Prob^\calT_\mu(\F[\le n] \widetilde{B}).
\end{array}\right.
\]

Here again, with no assumption on $\calT$, clearly enough, the
sequences $(q_n^{\mathsf{Yes}})_n$ and $(q_n^{\mathsf{No}})_n$ are
non-decreasing and converge respectively to
$\Prob^\calT_\mu(\F \widetilde{\widetilde{B}})$ and
$\Prob^\calT_\mu(\F \widetilde{B})$.  Assuming now that $\calT$ is
persistently decisive w.r.t to $B$ and decisive
w.r.t. $\widetilde{B}$, the two sequences are closely related, as
stated below. The proof of this result can be found
page~\pageref{app:quantrepreach}.
\begin{restatable}[Approximation scheme for repeated reachability]{proposition}{quantrepreach}
\label{prop-approx-buchi}
If $\calT$ is $\PD(\mu,B)$ and $\D(\mu,\widetilde{B})$, then the two sequences $(q_n^{\mathsf{Yes}})_n$ and $(1-q_n^{\mathsf{No}})_n$ are adjacent
 and converge to $\Prob^\calT_\mu(\G\F B)$. 
\end{restatable}

Effectiveness of the scheme relies on the computability of the
avoid sets $\widetilde{B}$ and $\widetilde{\widetilde{B}}$, and on the
effective computation of the probability of cylinders of the forms
$\Cyl(\underbrace{\neg \widetilde{B},\ldots,\neg
  \widetilde{B}}_{n\text{ times}},\widetilde{\widetilde{B}})$ and
$\Cyl(\underbrace{\neg \widetilde{\widetilde{B}},\ldots,\neg
  \widetilde{\widetilde{B}}}_{n\text{ times}},\widetilde{B})$.
Similarly as before, in case $q_n^{\mathsf{Yes}}$ and
$q_n^{\mathsf{No}}$ cannot be computed exactly, but can only be
approximated up to any desired error bound, this scheme can be refined
to obtain a $2 \varepsilon$-approximation for $\Prob^\calT_\mu(\G \F
B)$.

\subsection{$\omega$-regular properties in DMC 
  with a finite attractor}
\label{subsec:quant-attractor}

To go beyond reachability and repeated reachability, we now consider
an $\omega$-regular property given by a DMA
$\calM = (Q,q_0,E,\mathcal{F})$. We assume that $\calT =
(S,\Sigma,\kappa,\AP,\calL)$ is a labelled DMC.


In order to approximate the probability that the model satisfies this
external specification, we assume that $\calT$ has a finite attractor.
Following Section~\ref{subsec:qual-dmc}, we consider the finite
attractor $B$ of $\calT \ltimes \calM$, and we apply
Theorem~\ref{th:good-bscc-for-quant-analysis}: for each
$\mu\in\Dist(S)$,
\[
\Prob_{\mu \times \delta_{q_0}}^{\calT \ltimes \calM}(\mathsf{Inf} \in \calF) =
\sum_{C \in \mathsf{Good}^B_{\calT \ltimes \calM}(\calF)}
\Prob_{\mu \times \delta_{q_0}}^{\calT \ltimes \calM}(\F C ) \enspace.
\]

Thus, the computation of the probability that a given model satisfies
a given external specification is reduced to the computation of a
reachability probability.  Now, given that $\calT$ and hence $\calT
\ltimes \calM$ has a finite attractor, $\calT \ltimes \calM$ is
$\D(\mu \times \delta_{q_0},B)$ for any measurable set $B$, so that we can apply
the approximation scheme from Section~\ref{subsec:approx-reach} to
obtain an approximation of the desired value.


The effectiveness of the approach relies on the effectiveness of the
scheme for reachability, but also on the computability of an attractor
for $\calT$, and of the set of good BSCCs of the graph of the
attractor.



\subsection{$\omega$-regular properties of general STSs via
  abstraction and finite attractor}
  \label{subsec:quantMullerAbstr}

We assume the same framework as in Section~\ref{subsec:qual-abs}, that is
$\calT_1 = (S_1,\Sigma_1,\kappa_1,\AP,\calL_1)$ and $\calT_2 =
(S_2,\Sigma_2,\kappa_2,\AP,\calL_2)$ are two LSTSs such that:
\begin{itemize}
\item $\calT_2$ is a sound $\alpha$-abstraction of $\calT_1$
\item $\calT_2$ is a DMC with a finite attractor $B_2$. 
\end{itemize}
We consider again a DMA $\calM = (Q,q_0,E,\mathcal{F})$, as well as
the products $\calT_1\ltimes\calM$ and $\calT_2\ltimes\calM$. Writing
$\calB=\lbrace \alpha_\calM^{-1}(B)\mid B\in \Sigma'_2\rbrace$, we
assume that $\calT_1\ltimes\calM$ is $\D(\calB)$. Remember that this
implies, from Lemma~\ref{lemma:alphabar}, that $\calT_2\ltimes\calM$
is a sound $\alpha_\calM$-abstraction of $\calT_1\ltimes\calM$.

Fix an initial distribution $\mu$ for $\calT_1$.  Thanks to
Theorem~\ref{theo:titi}: 
\[
\Prob_{\mu \times \delta_{q_0}}^{\calT_1 \ltimes \calM}(\mathsf{Inf} \in
\calF) = \sum_{C \in \mathsf{Good}^{B_2}_{\calT_2 \ltimes \calM}(\calF)}
\Prob_{\mu \times \delta_{q_0}}^{\calT_1 \ltimes \calM}(\F \alpha^{-1}(C)
)\enspace.
\]

Thus, as previously, the computation of the probability that a given
model satisfies a given external specification is reduced to the
computation of a reachability probability. 
Since we assumed $\calT_1\ltimes\calM$ to be $\D(\calB)$, we can use
the approximation scheme from Section~\ref{subsec:approx-reach} to
approximate the searched value.

Effectiveness of the approach requires effective numerical
computations for the distributions, as well as good constructivity
properties for various sets, like the BSCCs of the graph of the
attractor, and avoid-sets of these, etc.

\subsection{Time-bounded verification of stochastic real-time systems}
\label{nonzeno}\label{subsec:bounded}

The initial motivation to consider general STSs stems from real-time
stochastic systems, that is, systems with both timing constraints and
probabilistic choices. While everything which precedes holds for any
kind of STSs, we highlight now some specific features of real-time
stochastic systems. 

\begin{definition} 
  A \emph{real-time stochastic transition system} (RT-STS) is an STS
  $\calT = (\widehat{S},\widehat{\Sigma},\kappa)$ such that (i) there
  is a measurable space $(S,\Sigma)$ with $\widehat{S} = S \times
  \bbR_{\ge 0}$, and $\widehat\Sigma$ is the product $\sigma$-algebra
  of $\Sigma$ and the Borel sets of $\bbR_{\ge 0}$; and (ii) for every
  $(s,t) \in S \times \bbR_{\ge 0}$, $\kappa((s,t),\{(s',t') \in S \times
  \bbR_{\ge 0} \mid t' > t\}) =1$.
\end{definition}
The first condition makes explicit the time component of the system
(given by $\bbR_{\ge 0}$; $S$ then contains the spatial information),
while the second condition imposes the time to increase almost-surely.
By explicitly integrating absolute time into DMCs (where it is
increased by one at each new event) or CTMCs (where it is increased by
the time elapsed in each state -- hence it represents absolute time
since the start of the system), they can be interpreted as
RT-STSs. All other examples that we will consider in
Section~\ref{sec:appli} can also be seen as RT-STSs, after explicit
integration of absolute time in the state-space.

Let $\calT = (\widehat{S},\widehat{\Sigma},\kappa)$ be an RT-STS.
A desirable property of a real-time system is that it should be
(almost-surely) \emph{non-Zeno}: a path $\rho = (s_0,t_0) (s_1,t_1)
\ldots \in \Paths(\calT)$ is non-Zeno whenever $\lim_{n \to +\infty}
t_n = +\infty$. Under such an hypothesis, we first identify natural
attractors of an RT-STS.

\begin{lemma}
  Assume that $\calT$ is almost-surely non-Zeno. Then, for every
  $\Delta \in \bbQ_{\ge 0}$, the set $A_\Delta = \{(s,t) \in S \times
  \bbR_{\ge 0} \mid t > \Delta\}$ is an attractor of $\calT$.
\end{lemma}

As a consequence, as soon as it is almost-surely non-Zeno, an RT-STS
is (strongly) decisive w.r.t. every bounded measurable subset and
every initial distribution. Fix $\mu$ an initial distribution
assigning $0$ to the initial timestamp $t_0$. Assume one wants to
compute the probability of property $B' \U[I] B$ from $\mu$, where $I$
is some bounded interval of $\bbR_{\ge 0}$ with rational bounds, and
$B,B' \in \Sigma$; this is the probability of the following set of
paths:
\[
\{ \rho = (s_0,t_0) (s_1,t_1) \ldots \in \Paths(\calT) \mid \exists n
\in \bbN\ \text{s.t.}\ t_n \in I,\ s_n \in B\ \text{and}\ \forall
j<n,\ s_j \in B'\} \enspace.
\]
Then, for any $\Delta \in \bbQ_{\ge 0}$ with $\Delta>\sup I$,
$A_\Delta$ is included in the avoid-set of $B \times I$, and $\calT$
is therefore decisive w.r.t. $B \times I$. In particular, the
approximation scheme of Subsection~\ref{subsec:approx-reach} applies.

Hence, assuming the RT-STS $\calT$ is almost-surely non-Zeno (which
needs to be proven ``by hand'', or structurally obvious), and under
some effectiveness assumption on $\calT$, the quantitative analysis of
time-bounded until or reachability properties is doable.

\section{Applications}\label{sec:appli}

The general approach to the qualitative and quantitative analysis of
stochastic systems over a possibly continuous state-space can be
instantiated in multiple frameworks. To demonstrate its versatility,
we present three types of models to which it applies: stochastic timed
automata, generalized semi-Markov processes and stochastic time Petri
nets. These models are taken from the literature without further
motivations. This section is technical (since the models themselves
are complex), and can be skipped by the reader not necessarily
interested in these models.  However, it is interesting to observe
that several results from the literature can be recovered (and
extended) via our generic approach.

\subsection{Stochastic timed automata}\label{subsec:sta}
Stochastic timed automata (STA)~\cite{BBB+14} are stochastic real-time
processes derived from timed automata~\cite{AD94} by randomizing both
the delays and the edge choices. The semantics of a STA is naturally
given via a STS as defined in this paper,
although this had not been formulated this way originally.

Several decidability results have been proven for subclasses of STA,
requiring the development of ad-hoc
methods~\cite{BBBBG07,BBBBG08,BBBM08,BBJM12}, and in~\cite{BBB+14}, we
proposed the first unifying method capturing all known decidability
results for the qualitative model-checking problem: the so-called
\emph{thick graph} is a finite graph based on the standard region
automaton construction for timed automata~\cite{AD94}, which allows
one to infer good transfer properties from this finite graph to the
original STA when some \emph{fairness} property is satisfied.
The current work improves our understanding of~\cite{BBBC16} and
allows us both to unify all decidability and approximability results
that were known, and to get new approximability results for the
quantitative model-checking problem (of $\omega$-regular properties).

\subsubsection{Definition}

To define the model properly, we first give some notations. Let $X$ be
a finite set of clocks. We write $G(X)$ for the set of guards defined
as finite conjunctions of constraints of the form $x \bowtie c$, where
$x \in X$, $\mathord{\bowtie} \in \{<,\le,=,\ge,>\}$ and $c \in
\mathbb{N}$. Guards are interpreted over clock valuations $\nu \colon
X \to \IRpos$ in a natural way~--~we then write $\nu \models g$.
Also, for $\nu$ a valuation we define $[Y\leftarrow 0](\nu)$ the
valuation assigning $0$ to every $x\in Y$ and $\nu(x)$ to each other
clock, and if $d \in \IRpos$, we write $\nu +d $ for the valuation
assigning $\nu(x)+d$ to every clock $x \in X$.

\begin{definition}
  A \emph{stochastic timed automaton} (STA) is a tuple 
  \[
  \calA=(L, \ell_0, X, E, (\mu_\gamma)_{\gamma \in L \times \IR_{\ge
      0}^X},(w_e)_{e \in E})
  \] 
  where:
  \begin{itemize}
  \item $L$ is a finite set of states (or locations);
  \item $\ell_0 \in L$ is the initial state;
  \item $X$ is a finite set of clocks;
  \item $E \subseteq L \times G(X) \times 2^X \times L $ is a finite
    set of edges; and
  \item for every configuration $\gamma \in L \times \IRpos^X$,
    $\mu_\gamma$ is a(n a priori) continuous distribution over
    possible delays from $\gamma=(\ell,\nu)$, that is, the support of
    distribution $\mu_\gamma$ is precisely $I(\gamma)
    \stackrel{\mathrm{def}}{=} \{d \in \IRpos \mid \exists e =
    (\ell,g,Y,\ell') \in E\ \text{s.t.}\ \nu +d \models g\}$;
  \item and for every $e \in E$, $w_e \in \INpos$ is a positive
    weight.
  \end{itemize}
\end{definition}

Originally, the semantics of an STA $\calA = (L, \ell_0, X, E,
(\mu_\gamma)_{\gamma \in L \times \IR_{\ge 0}^X},(w_e)_{e \in E})$ was
defined as a probability measure on the set of possible runs of the
underlying timed automaton $(L, \ell_0, X, E)$: a run in such a timed
automaton is an alternating sequence of delay transitions and of
discrete transitions. A delay transition is of the form $(\ell,\nu)
\xrightarrow{d} (\ell,\nu+d)$, where $\gamma
\stackrel{\mathrm{def}}{=} (\ell,\nu) \in L \times \IRpos^X$ is a
configuration and $d \in \IRpos$,\footnote{Later we will also write
  $\gamma+d$ for the configuration $(\ell,\nu+d)$.} and a discrete
transition is of the form $(\ell,\nu) \xrightarrow{e} (\ell',\nu')$
where $e = (\ell,g,Y,\ell') \in E$ is such that $\nu \models g$, and
$[Y \leftarrow 0](\nu) = \nu'$. When $\nu \models g$, we say that $e$
is enabled at $\gamma$.

The probability measure was obtained by sampling delay transitions
from a configuration $\gamma$ following distribution $\mu_\gamma$, and
by sampling discrete transitions using the weights: the probability to
take edge $e$ from configuration $\gamma$ is given by $p_\gamma(e)
\stackrel{\mathrm{def}}{=} \frac{w_e}{\sum \{w_{e'} \mid {e'\
    \text{enabled at}\ \gamma} \}}$ if $e$ is enabled at $\gamma$, and
by $p_\gamma(e) \stackrel{\mathrm{def}}{=} 0$ otherwise.

To have properly-defined measures we need some sanity assumptions
on distributions $(\mu_\gamma)_{\gamma \in L \times \IRpos^X}$: If we
write $\lambda$ for the Lebesgue measure over $\Rpos$, it must be the
case that for each $\gamma\in L\times\Rpos^{X}$, if
$\lambda(I(\gamma))>0$ then $\mu_{\gamma}$ is equivalent to the
restriction of $\lambda$ on $I(\gamma)$; Otherwise, it is the uniform
distribution over the points of $I(\gamma)$.

We now give the semantics of an STA $\calA = (L, \ell_0, X, E,
(\mu_\gamma)_{\gamma \in L \times \IR_{\ge 0}^X},(w_e)_{e \in E})$ as
an STS $\calT_{\A} = (S_{\A},\Sigma_{\A},\kappa_{\A})$ as follows. The
set $S_{\A}$ is the set of configurations $L \times \IRpos^X$,
$\Sigma_{\A}$ is the $\sigma$-algebra product between $2^L$ and the
Borel $\sigma$-algebra on $\IRpos^{|X|}$, and the kernel $\kappa_\A$
is defined by:
\[
\kappa_\A(\gamma,B) = \sum_{e = (\ell,g,Y,\ell')\in E} \ \int_{d \in
  \IRpos} \mathds{1}_B(\ell',[Y \leftarrow 0](\nu+d)) \cdot
p_{\gamma+d}(e) \ \ud \mu_\gamma(d)
\]
where $\mathds{1}_B$ is the characteristic function of $B$.
It gives the probability to hit set $B \subseteq
S_{\A}$ from configuration $\gamma$ in one step (composed of a delay
transition followed by a discrete transition).

The probability measure on paths derived from $\calT_\A$ in
Section~\ref{section:PrelimMeasure} coincides with the original
definition of~\cite{BBB+14}. 

We fix for the rest of this section an STA $\calA = (L, \ell_0, X, E,
(\mu_\gamma)_{\gamma \in L \times \IR_{\ge 0}^X},(w_e)_{e \in E})$,
and $\calT_\A = (S_\calA,\Sigma_\A,\kappa_\A)$ its corresponding
STS.

\begin{example}[A stochastic timed automaton with an ``unfair''
  convergence behaviour]
\label{Example:pacman}
Consider the STA $\A$ of Figure~\ref{fig:pacman}, with: $L=\lbrace
  \ell_0,\ldots,\ell_4\rbrace$, $X=\lbrace x,y\rbrace$ and the set of
  edges $E$ as described on the figure.
  We assume that each edge has a weight of~$1$ and that each location
  is either equipped with a uniform distribution over possible delays
  (in $\ell_0$, $\ell_2$ and $\ell_4$) or a Dirac distribution over
  the unique possible delay (in $\ell_1$ and $\ell_3$).\footnote{When
    we reach $\ell_1$, the value $v_y$ of clock $y$ is smaller than
    $1$; since the constraint on the edge between $\ell_1$ and
    $\ell_2$ is constrained by $y=1$, there is a single possible delay
    for taking this edge: wait $d$ such that $v_y+d=1$.}  As said
  previously, it can be considered as an STS $\calT_\A$ where the set
  of states is given by $L\times\Rpos^2$ and the Markov kernel is
  computed according to the distributions over the edges and the
  delays.

\begin{figure}[h]
  \begin{center}
    \begin{tikzpicture}[yscale=.85]
      \path[use as bounding box] (-6,-1.4) -- (6,.7);
      \path (0,0) node[draw,circle,inner sep=2pt] (q0) {$\ell_0$};
      \path (0,-1.1) node[] (q0b) {$\scriptstyle x=0$};
      \path (0,-1.5) node[] (q0c) {$\scriptstyle 0<y<1$};

      \path (3,0) node[draw,circle,inner sep=2pt] (q1) {$\ell_1$};
      
      \path (6,0) node[draw,circle,inner sep=2pt] (q3) {$\ell_2$};

      \path (-3,0) node[draw,circle,inner sep=2pt] (q4) {$\ell_3$};
          
      \path (-6,0) node[draw,circle,inner sep=2pt] (q6) {$\ell_4$};

      \draw[arrows=latex'-] (q0) -- (0,-1);
      
      \draw[arrows=-latex'] (q0) -- (q1) node[pos=.5, above,sloped]
      {$y<1$};
      
      \draw[arrows=-latex'] (q1) -- (q3) node[pos=.5, above,sloped]
      {$y=1$} node[pos=.5, below, sloped] {$y:=0$};
      
      \draw[arrows=-latex',rounded corners] (q3) -- ++(0,-1.5) -- node
      [midway, above] {$x>1\wedge y<1$} node [midway,below] {$x:=0$} ++(-5,0) -- (q0);
      
      \draw[arrows=-latex'] (q0) -- (q4) node[pos=.5, above,sloped]
      {$1<y<2$};
      
      \draw[arrows=-latex'] (q4) -- (q6) node[pos=.5, above,sloped]
      {$y=2$} node[pos=.5, below, sloped] {$y:=0$};
      
      \draw[arrows=-latex',rounded corners] (q6) -- ++(0,-1.5) -- node
      [midway, above] {$x>2\wedge y<1$} node [midway,below] {$x:=0$}  ++(5,0) -- (q0);
    \end{tikzpicture}
  \end{center}
  \caption{A two-clock STA $\A$ with an unfair convergence
    behaviour} \label{fig:pacman}
\end{figure}
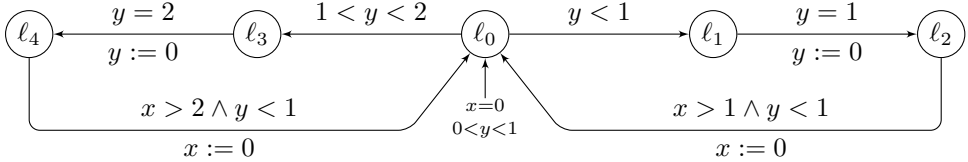
We would like to stress that $\A$ suffers from a time-convergence
phenomenon. This convergence phenomenon (that we make precise in the
following lines), is due to the timing constraints, and is in fact
inherent to the underlying timed automaton (without the stochastic
aspects).  We will see later (Example~\ref{Example:pacmanbis}) how it
impacts the stochastic behaviour of the STA $\A$. Let us now discuss
the (non stochastic) time-convergence phenomenon. In order to do so,
for the rest of the paragraph, we see $\A$ as a timed automaton and
forget about the stochastic aspects. Let us imagine that we enter
location $\ell_0$ with the value of clock $x$ (resp. $y$) being $0$
(resp. $0<\nu<1$), we are thus in configuration
$(\ell_0,(0,\nu))$. Let us consider the case where we take the right
loop. We thus first enter location $\ell_1$ and then $\ell_2$, that we
reached with configuration $(\ell_2,1-\nu,0)$, after a total delay of
$1-\nu$ time units spent since the last arrival in $\ell_0$. In order
to return to $\ell_0$, we have to wait a delay $\nu'$ such that
$\nu'<1$ (because of the guard $y<1$) and $\nu < \nu'$ (because of the
guard $x>1$). We thus return to $\ell_0$ with configuration
$(\ell_0,(0,\nu'))$, where $\nu < \nu'$. One can check that a similar
situation occurs when taking the left loop. Thus when considering an
infinite path of $\A$, if we denote by $(\ell_0,(0,\nu_n))$ its
configuration at the $n$-th passage in $\ell_0$, we can infer that the
sequence $(\nu_n)_{n \in \IN}$ is increasing (and bounded by $1$), and
thus converging.
\end{example}

\subsubsection{The thick graph abstraction}
\label{subsubsec:thickgraph}

The thick graph of~\cite{BBB+14} is an abstraction in our context.  To
see this, we recall the concept of regions, that have been designed
for standard timed automata~\cite{AD94}. We write $M_\A$ for the maximal
integer appearing in a guard of $\A$. Let $\nu,\nu' \in \IRpos^X$ be
two valuations over $X$. We say that $\nu$ and $\nu'$ are
\emph{region-equivalent} for $\A$ whenever the following conditions
hold:
\begin{enumerate}
\item for every $x \in X$, either both $\nu(x)$ and $\nu'(x)$ are
  stricly larger than $M_\A$, or  the integral parts of $\nu(x)$ and
  $\nu'(x)$ coincide;
\item for every $x,y \in X$ such that $\nu(x),\nu(y) \le M_\A$,
  writing $\{ \cdot \}$ for the fractional part, $\{\nu(x)\} \le
  \{\nu(y)\}$ if and only if $\{\nu'(x)\} \le \{\nu'(y)\}$.
\end{enumerate}
This region-equivalence has finite-index, and partitions the set of
valuations $\IRpos^X$ into classes which are called \emph{regions},
and we write $R_\A$ for the set of regions. If $\nu \in \IRpos^X$, we
write $[\nu]_\A$ for the region to which $\nu$ belongs.

We define the abstraction $\alpha : L \times \IRpos^X \to L \times
R_\A$ as the projection which associates $(\ell,\nu)$ onto
$(\ell,[\nu]_\A)$. We then define the finite Markov chain
$\calT_\A^{\mathsf{tg}}$ as follows:
\begin{itemize}
\item its set of states is $L \times R_\A$;
\item there is an edge from $(\ell,r)$ to $(\ell',r')$ whenever there
  exists some $\nu \in r$ such that
  $\kappa_{\A}((\ell,\nu),\{\ell'\} \times r'))>0$;\footnote{Note that
    it is a local condition which is easy to check.}
\item from each state $(\ell,r) \in L\times R_\A$, we associate the
  uniform distribution over $\{(\ell',r') \in L \times R_\A \mid
  \text{there is an edge from}\ (\ell,r)\ \text{to}\ (\ell',r')\}$.
\end{itemize}

By construction, we get:

\begin{lemma}
  $\calT_\A^{\mathsf{tg}}$ is a finite $\alpha$-abstraction of
  $\calT_\A$.
\end{lemma}
\noindent Let us notice that finiteness of the abstraction implies
completeness (Lemma~\ref{lemma:finite-soundness}).

As witnessed in \cite[Appendix~D.2]{BBB+14}, this abstraction may not
give much information in general about the probability of linear-time
properties in the original STA (see Example~\ref{Example:pacmanbis}).
However we will see that, in several cases, it helps to obtain
decidability and approximability results (among which some are new).

\begin{example}[A stochastic timed automaton with an ``unfair''
  convergence behaviour (continued)]
  \label{Example:pacmanbis}
  We know that the thick graph viewed as a finite Markov chain, is an
  $\alpha$-abstraction of the original STA, but it can be shown that
  in general it is not sound. Let us denote $\calT_\A$ the STS
  naturally associated with the STA of Example~\ref{Example:pacman}
  (see Fig.~\ref{fig:pacman}).  One can show that
  $\calT^{\mathsf{tg}}_\A$, the thick graph associated with
  $\calT_\A$, is the one provided on Fig.~\ref{fig:pacmanthickgraph},
  starting from a Dirac distribution $\delta_{(\ell_0, (0,\nu))}$ with
  $0<\nu<1$. The regions are the following ones: $r_0=\{(x,y) \mid x=0
  \wedge 0<y<1\}$, $r_1=\{(x,y) \mid 0<x<y<1\}$, $r_2=\{(x,y) \mid y=0
  \wedge 0<x<1\}$, $r_3=\{(x,y) \mid 1<x<y<2\}$, $r_4=\{(x,y) \mid y=0
  \wedge 1<x<2\}$. We clearly have that $\calT^{\mathsf{tg}}_\A$ is an
  $\alpha$-abstraction of $\calT_\A$.
   
  However, it can be shown that $\calT^{\mathsf{tg}}_\A$ is not a
  sound abstraction of $\calT_\A$.  The time-convergence phenomenon of
  $\A$ (described in Example~\ref{Example:pacman}) implies that each
  time we return to location $\ell_0$, the probability to take the
  right loop decreases while the probability to take the left loop
  increases. More precisely, it has been shown that in the original
  STA $\A$, the probability to reach $\ell_2$ from $\ell_0$ is stricly
  lower than $1$.
  This has been done formally via a tedious and technical calculation
  of Taylor series in~\cite[Section~6.2.2]{BBB+14}.  This implies that
  $\calT^{\mathsf{tg}}_\A$ is not a sound abstraction of $\calT_\A$
  (since the probability to reach $(\ell_2,r_2)$ from $(\ell_0,r_0)$
  is $1$ in $\calT^{\mathsf{tg}}_\A$).  In fact, a sound abstraction
  of $\calT_\A$ would rather behave as the
  \textbf{non-homogeneous}\footnote{Although we only consider
    homogeneous systems, the non-homogeneous ones can fit our general
    model of STS by unfolding it. For instance
    Example~\ref{counterexample:fairness} can be seen as the unfolding
    of a finite non-homogeneous Markov chain with two states.} finite
  Markov chain of Fig.~\ref{fig:pacmannonhom}, where $n$ represents
  the $n$-th passage in $(\ell_0,r_0)$. This shows in particular that
  general STA are not fair (and thus not decisive). This is why we
  focus on two subclasses of STA in the rest of this section.
    
  \begin{figure}[h]
  \begin{center}
    \begin{tikzpicture}[yscale=.85]
       \everymath{\scriptstyle}
      \path[use as bounding box] (-6,-1.4) -- (6,.7);
      \path (0,0) node[draw,rectangle,inner sep=3pt] (q0) {$(\ell_0,r_0)$};
      
      \path (2,0) node[draw,rectangle,inner sep=3pt] (q1) {$(\ell_1,r_1)$};
      
      \path (4,0) node[draw,rectangle,inner sep=3pt] (q3) {$(\ell_2,r_2)$};

      \path (-2,0) node[draw,rectangle,inner sep=3pt] (q4) {$(\ell_3,r_3)$};
          
      \path (-4,0) node[draw,rectangle,inner sep=3pt] (q6) {$(\ell_4,r_4)$};
   
      \draw[arrows=latex'-] (q0) -- (0,-.75);
      
      \draw[arrows=-latex'] (q0) -- (q1) node[pos=.5, above,sloped]
      {$\frac{1}{2}$};
      
      \draw[arrows=-latex'] (q1) -- (q3) ;
      
      \draw[arrows=-latex',rounded corners] (q3) -- ++(0,-1) --
      ++(-3,0) -- (q0);
      
      \draw[arrows=-latex'] (q0) -- (q4) node[pos=.5, above,sloped]
      {$\frac{1}{2}$};
      
      \draw[arrows=-latex'] (q4) -- (q6) ;
      
      \draw[arrows=-latex',rounded corners] (q6) -- ++(0,-1) --   ++(3,0) -- (q0);
    \end{tikzpicture}
  \end{center}
  \caption{$\calT^{\mathsf{tg}}_\A$, the thick graph (viewed as a
    finite Markov chain) associated with the two-clock STA with an
    ``unfair'' convergence behaviour
    (Fig.~\ref{fig:pacman}).} \label{fig:pacmanthickgraph}
\end{figure}
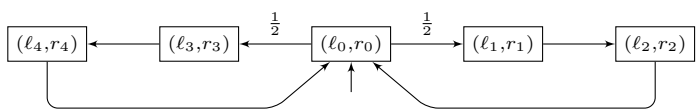

  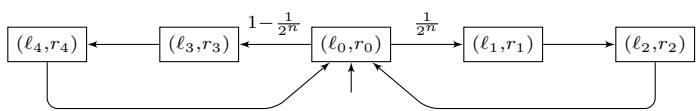
\begin{figure}[h]
  \begin{center}
    \begin{tikzpicture}[yscale=.85]
       \everymath{\scriptstyle}
      \path[use as bounding box] (-6,-1.4) -- (6,.7);
      \path (0,0) node[draw,rectangle,inner sep=3pt] (q0) {$(\ell_0,r_0)$};
      
      \path (2,0) node[draw,rectangle,inner sep=3pt] (q1) {$(\ell_1,r_1)$};
      
      \path (4,0) node[draw,rectangle,inner sep=3pt] (q3) {$(\ell_2,r_2)$};

      \path (-2,0) node[draw,rectangle,inner sep=3pt] (q4) {$(\ell_3,r_3)$};
          
      \path (-4,0) node[draw,rectangle,inner sep=3pt] (q6) {$(\ell_4,r_4)$};
   
      \draw[arrows=latex'-] (q0) -- (0,-.75);
      
      \draw[arrows=-latex'] (q0) -- (q1) node[pos=.5, above,sloped]
      {$\frac{1}{2^n}$};
      
      \draw[arrows=-latex'] (q1) -- (q3) ;
      
      \draw[arrows=-latex',rounded corners] (q3) -- ++(0,-1) --
      ++(-3,0) -- (q0);
      
      \draw[arrows=-latex'] (q0) -- (q4) node[pos=.5, above,sloped]
      {$1-\frac{1}{2^n}$};
      
      \draw[arrows=-latex'] (q4) -- (q6) ;
      
      \draw[arrows=-latex',rounded corners] (q6) -- ++(0,-1) --   ++(3,0) -- (q0);
    \end{tikzpicture}
  \end{center}
  \caption{A non-homogeneous finite Markov chain which is in some
    sense equivalent to the STA with an ``unfair'' convergence
    behaviour} \label{fig:pacmannonhom}
\end{figure}

\end{example}

\subsubsection{Reactive STA}

Following~\cite{BBJM12}, the STA $\calA$ is \emph{reactive} whenever
for every configuration $\gamma = (\ell,\nu) \in S_\A$, $I(\gamma) =
\IRpos$, and for every $\ell$, there exists a distribution $\mu_\ell$
with support $\IRpos$ such that for every $\nu \in \IRpos^X$,
$\mu_{(\ell,\nu)} = \mu_\ell$. Note that we do not make any Markovian
hypothesis on time elapsing, and $\mu_\ell$ does not need to be
exponential.

We take the notations used in the previous subsection for defining the
thick-graph abstraction. A region $r$ is \emph{memoryless}
whenever for every clock $x \in X$, either $\nu(x)=0$ for every $\nu
\in r$, or $\nu(x) > M_\A$ for every $\nu \in r$. We write
$R_\A^{\mathsf{mem}}$ for the set of memoryless regions.

From \cite[Lemma~13]{BBB+14}, which states that the set of memoryless
regions is visited infinitely often almost-surely from every
configuration $\gamma \in S_\A$,\footnote{To give all arguments, it is
  easy to see that, in one step, one can ensure reaching a memoryless
  region by delaying at least $M_\A+1$ time units; since there is one
  single distribution which is applied at every configuration of a
  given location, the probability to do so is uniformly bounded from
  below from every configuration.} we get:

\begin{proposition}
  \label{prop:reactive-attractor}
  The set $\alpha^{-1}(L \times R_\A^{\mathsf{mem}})$ is an
  attractor for $\calT_\A$.
\end{proposition}

Using Propositions~\ref{prop:attractorSound} and~\ref{coro:DecSound}, we also get that:

\begin{proposition}\label{prop:staReactiveSound}
  $\calT_\A^{\mathsf{tg}}$ is a sound $\alpha$-abstraction of
  $\calT_\A$.
\end{proposition}

\begin{proof}
  It can easily be shown that $L \times R_\A^{\mathsf{mem}}$ is
  a finite attractor of $\calT_\A^{\mathsf{tg}}$.  Thanks to
  Proposition~\ref{prop:reactive-attractor}, $\alpha^{-1}(L \times
  R_\A^{\mathsf{mem}})$ is an attractor for
  $\calT_\A^{\mathsf{tg}}$. It remains to show the last condition of
  the hypotheses of Proposition~\ref{prop:attractorSound}. We
  therefore need to prove that for each
  $(\ell_{\mathsf{m}},r_{\mathsf{m}}) \in L \times
  R_\A^{\mathsf{mem}}$, there are $p>0$ and $k\in\IN$ such that
  for each region $(\ell,r) \in L \times R_\A$:
  \begin{itemize}
  \item for each $\mu\in\Dist(\alpha^{-1}(\ell_{\mathsf{m}},r_{\mathsf{m}}))$,
    $\Prob^{\calT_\A}_{\mu}(\F[\leq k] \alpha^{-1}(\ell,r))\geq p$, or
  \item for each
    $\mu\in\Dist(\alpha^{-1}(\ell_{\mathsf{m}},r_{\mathsf{m}}))$,
    $\Prob^{\calT_\A}_{\mu}(\F \alpha^{-1}(\ell,r))=0$.
  \end{itemize}

  This is a consequence of \cite[Lemma~F.4]{BBB+14} which says that
  from a memoryless region, the future (and its probability) is
  independent of the precise current configuration. This in particular
  implies that for two configurations $\gamma,\gamma' \in
  \alpha^{-1}(\ell_{\mathsf{m}} ,r_{\mathsf{m}})$, for every
  $\alpha$-closed set $B$, for every integer $k$,
  $\Prob_{\delta_\gamma}^{\calT_\A}(\F[=k] B) =
  \Prob_{\delta_{\gamma'}}^{\calT_\A}(\F[=k] B)$.  By extension, for
  every $\mu\in\Dist(\alpha^{-1}(\ell_{\mathsf{m}},r_{\mathsf{m}}))$,
  $\Prob_{\mu}^{\calT_\A}(\F[=k] B) =
  \Prob_{\delta_{\gamma}}^{\calT_\A}(\F[=k] B)$. This implies the
  expected bounds, by taking $B = \alpha^{-1}(\ell,r)$.
\end{proof}

Similarly to labelled STS, we consider labelled STA, where each
location is labelled by atomic propositions.  As consequences of
Sections~\ref{sec:qualitative} and~\ref{sec:quantitative}, we get the following
decidability and approximability results for reactive STA:

\begin{corollary}\label{coro:staResults}
  Let $\A$ be a reactive labelled STA, and $\calM$ a DMA. Then:
  \begin{enumerate}
  \item we can decide whether $\A$ satisfies almost-surely $\calM$;

  \item for every initial distribution $\mu$ which is numerically
    amenable w.r.t. $\A$\footnote{We say that a distribution $\mu$ is
      numerically amenable w.r.t. $\calA$ if, given $k \in \IN$, given
      $\varepsilon>0$ and given a sequence of locations and regions
      $(\ell_0,r_0), (\ell_1,r_1), \ldots, (\ell_k,r_k)$, one can
      compute a numerical  approximation
      $\Prob^{\A}_\mu(\Cyl((\ell_0,r_0),(\ell_1,r_1),\dots,(\ell_k,r_k)))$
      up to $\varepsilon$.},
    we can compute arbitrary
    approximations of $\Prob_\mu^{\calT_\A}(\calM)$.
  \end{enumerate}
\end{corollary}

\begin{proof}
  This is an application of Theorem~\ref{theo:titi},
    Corollary~\ref{coro:theotiti} and of
    Sections~\ref{subsec:approx-reach}
    and~\ref{subsec:quantMullerAbstr}. It should be noted that all the
    hypotheses are met:
  \begin{itemize}
  \item $\calT_\A^{\mathsf{tg}}\ltimes\calM$ has a finite attractor:
    since $\calT_\A^{\mathsf{tg}}$ is a finite MC then so is
    $\calT_\A^{\mathsf{tg}}\ltimes\calM$ and we get a trivial finite
    attractor;
  \item $\calT_\A\ltimes\calM$ is decisive w.r.t. any
    $\alpha_\calM$-closed sets.
  \end{itemize}
  This second point is a little more tricky. First one should realise
  that since $\calT_\A$ is reactive, then $\calT_\A\ltimes\calM$ is
  also reactive, since the condition to be reactive concerns only the
  distributions over the delays on each location of the STA and those
  distributions are not modified from the product with $\calM$. It
  should be noted that $\calT_\A^{\mathsf{tg}}\ltimes\calM$
  corresponds to the thick region graph abstraction of
  $\calT_\A\ltimes\calM$ since $\calM$ does not influence the
  behaviour of $\calT_\A$. Then from
  Proposition~\ref{prop:staReactiveSound}, we know that
  $\calT_\A^{\mathsf{tg}}\ltimes\calM$ is a sound
  $\alpha_\calM$-abstraction of $\calT_\A\ltimes\calM$. Since
  $\calT_\A^{\mathsf{tg}}\ltimes\calM$ is a finite MC, we get that it
  is decisive w.r.t. any set of states. We can thus conclude from
  Proposition~\ref{thm:MuDecisiveAbstr}.
\end{proof}

\begin{rk}
  We believe that the proposed approach through abstractions and
  finite attractors simplifies drastically the proof of decidability
  of almost-sure model-checking, and in particular avoids the
ad-hoc but long and technical
  proof of~\cite[Lemma~7.14]{BBB+14}. Furthermore, we obtain
  interesting approximability results, some of them being consequences
  of~\cite{BBBC16}, but the general case of $\omega$-regular
  properties (in particular \LTL properties) being new to this paper.
\end{rk}


\begin{rk} Corollary~\ref{coro:staResults} can be extended to
  properties expressed as deterministic and complete Muller
  \emph{timed} automata (DCMTA), which are standard deterministic and
  complete\footnote{In this context, complete means that from every
    configuration, for every subset of $\AP$, and every
    $t \in \IRpos$, there is an edge labelled by that subset which is
    enabled after $t$ time units. So this is complete w.r.t time and
    actions.}  timed automata~\cite{AD94} with a Muller accepting
  condition. Indeed, the product of a reactive STA with such a DCMTA
  is reactive. Hence, the whole theory that we have developed applies:
  the STS of the product admits a sound finite abstraction. DCMTA
  allows one to express rich properties with timing constraints and
  one can evaluate their likelihood in the STA. It should be noticed
  that the convergence proof of the approximation scheme
  of~\cite{CDKM11} can be recovered as a byproduct, since CTMCs are
  particular cases of reactive STA. Another way to obtain this result
  would have been to apply the approach of Subsection~\ref{nonzeno} on
  time-bounded verification.
\end{rk}

\subsubsection{Single-clock STA}

We will apply a similar reasoning to single-clock STA. We therefore
assume that $\A$ is now a single-clock STA. As in~\cite[Section
7.1]{BBB+14}, we assume the following conditions: 
\begin{enumerate}[(i)]
\item for all $\ell\in L$, for all $\intervalcc{a,b}\subseteq\Rpos$,
  the function $\nu \mapsto \mu_{(\ell,\nu)}(\intervalcc{a,b})$ is
  continuous;
\item if $\gamma'=\gamma+t$ for some $t\ge 0$, and if $0\notin
  I(\gamma+t',e)$ for each $0\leq t'\leq t$, then
  $\mu_{\gamma}(I(\gamma,e))\leq\mu_{\gamma'}(I(\gamma',e))$;
\item there is $0<\lambda_0<1$ such that for every state $\gamma$ with
  $I(\gamma)$ unbounded, $\mu_\gamma(\intervalcc{0,\frac{1}{2}})\leq
  \lambda_0$,
\end{enumerate}
where for each $\gamma=(\ell,\nu)\in L\times\Rpos^X$ and for each
$e=(\ell,g,Y,\ell')\in E$, $I(\gamma, e)=\lbrace d\in \Rpos\mid
\nu+d\models g\rbrace$. These requirements are technical, but they are
rather natural and easily satisfiable. For instance, a timed automaton
equipped with uniform (resp. exponential) distributions on bounded
(resp. unbounded) intervals satisfy these conditions. If we assume
exponential distributions on unbounded intervals, the very last
requirement corresponds to the bounded transition rate condition
in~\cite{DP03}, required to have reasonable and realistic behaviours.

In~\cite[Section 7.1]{BBB+14}, there is no clear attractor
property. From the details of the proofs we can nevertheless define
$A_\A^{\max} = \{(\ell,r_0) \mid \ell \in L\} \cup \{(\ell,r) \in L
\times R_\A \mid \forall (\ell',r')\in L\times R_\A, \ (\ell,r) \to^*
(\ell',r')\ \text{in}\ \calT_\A^{\mathsf{tg}}\ \text{implies}\ r'
=r\}$ where $r_0$ is the region composed of the single null valuation.

\begin{proposition}
  \label{prop:oneclock-attractor}
  The set $\alpha^{-1}(A_\A^{\max})$ is an attractor for $\calT_\A$.
\end{proposition}

\begin{proof}
  Let $C = \{0\} \cup \{c \mid c\ \text{constant appearing in a guard
    of}\ \A\} \stackrel{\mathrm{def}}{=} \{c_0<c_1< \dots< c_h\}$.
  The set of regions for $\A$ can be chosen as $\{\{c_i\} \mid 0 \le i
  \le h\} \cup \{]c_{i-1};c_i[ \mid 1 \le i \le h\}$
  (see~\cite{LMS04}).

  Following the proof of~\cite[Theorem 7.2]{BBB+14}, the set of
  infinite paths in $\A$ can be divided into (a) the set of paths that
  take resetting edges infinitely often, and (b) the set of paths that
  take resetting edges only finitely often.

  We assume that the probability that (a) happens is positive, and we
  reason now in the $\sigma$-algebra which is conditioned by (a). Then
  under condition (a), $\alpha^{-1}(\{(\ell,r_0) \mid \ell \in L\})$
  is reached almost-surely.

  We assume that the probability that (b) happens is positive, and we
  reason now in the $\sigma$-algebra which is conditioned by
  (b). Under condition (b), almost-surely the value of the clock is
  non-decreasing along the path, and almost-surely a final region $r$
  is reached (that is, ultimately the value of the clock along the
  path belongs to $r$ forever). We fix such a region $r$, and we
  condition again with regard to that ``final region'' $r$. We write
  $E_r$ for the event (b) intersected with ``the path ends up in
  $r$''. Let $r'$ be a strict successor region of $r$, with dimension
  at least as big as that of $r$ (if $r$ is an open interval, then
  $r'$ has to be an open interval). There exists $\alpha>0$ such that
  for every $\nu \in r$, for every $\ell \in L$, for every $e =
  (\ell,g,Y,\ell')$ with $r' \subseteq g$,
  $\Prob_{\delta_{(\ell,\nu)}}^{\calT_\A}((\ell,\nu) \xrightarrow{e} )
  \ge \alpha$. Hence, using standard technics, we show that with
  probability $1$, if infinitely often such edges are enabled,
  infinitely often they will be taken; this contradicts hypothesis
  $E_r$. Hence, under condition $E_r$, with probability $1$, one
  cannot visit infinitely often configurations enabling edges guarded
  by some strict time-successor $r'$ of $r$. Once this is assumed, we
  can then show that almost-surely, only finitely many resetting edges
  can be enabled. This means that, under condition $E_r$,
  almost-surely, ultimately only states of $\alpha^{-1}(\{(\ell,r) \in
  L \times R_\A \mid \forall (\ell',r')\in L\times R_\A, \ (\ell,r) \to^* (\ell',r')\ \text{in}\
  \calT_\A^{\mathsf{tg}}\ \text{implies}\ r' =r\})$ are
  visited. Hence, that set is an attractor, under condition (b).

  Using some Bayes formula w.r.t. conditions (a) and (b), we conclude
  that $\alpha^{-1}(A_\A^{\max})$ is an attractor; this ends the
  proof.
%
\end{proof}

As before, we get:

\begin{proposition}
  $\calT_\A^{\mathsf{tg}}$ is a sound $\alpha$-abstraction of
  $\calT_\A$.
\end{proposition}

\begin{proof}
  We easily get that $A_\A^{\max}$ is a finite attractor for
  $\calT_\A^{\mathsf{tg}}$, whereas $\alpha^{-1}(A_\A^{\max})$ is an
  attractor for $\calT_\A$
  (Proposition~\ref{prop:oneclock-attractor}).

  As for reactive STA, it remains to show the last property appearing
  in the hypotheses of Proposition~\ref{prop:attractorSound}. The
  required bounds obviously exist for the region $r_0$ (since only a
  single valuation belongs to $r_0$). Furthermore, as argued in the
  proof of Proposition~\ref{prop:oneclock-attractor}, when condition
  (b) is assumed, ultimately, the paths almost surely end up in
  $\alpha^{-1}(\{(\ell,r) \in L \times R_\A \mid \forall (\ell,r)
  \to^* (\ell',r')\ \text{in}\ \calT_\A^{\mathsf{tg}}\ \text{implies}\
  r' =r\})$, hence, ultimately, the STA behaves like a \textbf{finite}
  Markov chain. The required bounds can be inferred.

  This allows to conclude that $\calT_\A^{\mathsf{tg}}$ is a sound
  $\alpha$-abstraction of $\calT_\A$ (using
  Propositions~\ref{prop:attractorSound} and~\ref{coro:DecSound}).
\end{proof}

As a consequence, we get the following decidability and approximability
results for one-clock STA:
\begin{corollary}\label{coro:staOneClockResults}
  Let $\A$ be a one-clock labelled STA, and $\calM$ a DMA. Then:
  \begin{enumerate}
  \item we can decide whether $\A$ satisfies almost-surely $\calM$;
  \item for every initial distribution $\mu$ which is numerically
    amenable w.r.t. $\A$,
      we can compute arbitrary
        approximations of $\Prob_\mu^{\calT_\A}(\calM)$. 
  \end{enumerate}
\end{corollary}

\begin{proof}
  {Similarly to the proof of Corollary~\ref{coro:staResults}, this is an
    application of Theorem~\ref{theo:titi},
    Corollary~\ref{coro:theotiti} and of
    Sections~\ref{subsec:approx-reach}
    and~\ref{subsec:quantMullerAbstr}. The facts that:
  \begin{itemize}
  \item $\calT_\A^{\mathsf{tg}}\ltimes\calM$ has a finite attractor, and
  \item $\calT_\A\ltimes\calM$ is decisive w.r.t. any $\alpha$-closed sets.  
  \end{itemize}
  can be deduced by similar arguments. We only observe that if $\calT_\A$ is a single-clock STA, then so is $\calT_\A\ltimes\calM$ and that hypotheses (i), (ii) and (iii) are preserved through the product with $\calM$ as those only concern distributions over the STA which are not altered from the product with $\calM$.}
\end{proof}

\begin{rk}
  The proof of the existence of an attractor is very similar to the
  one we used for proving the fairness property in~\cite[Section
  7.1]{BBB+14}. However, for free, we get all the approximation
  results (as previously only few results could be inferred
  from~\cite{BBBC16})! It is worth noting that these results encompass
  the results of~\cite{BBBM08}, where a strong assumption on cycles of
  the STA was made (but a closed-form for the probability could be
  computed). We remark here that the graph used in~\cite{BBBM08} is
  actually the graph of the attractor, as done in
  Section~\ref{subsec:qual-dmc}.
\end{rk}





\begin{rk}[Time-bounded analysis] Let us finish by discussing how and
  when the approach of Subsection~\ref{nonzeno} for timed properties
  can be applied to STA. Similarly to CTMCs~\cite{BHHK03} that they
  extend, reactive STA are almost-surely non-Zeno. Hence one can apply
  the approximation scheme of Subsection~\ref{subsec:approx-reach} to
  (time-)bounded until formulas or time-bounded reachability
  properties.

  One can decide whether a single-clock STA is almost-surely
  non-Zeno~\cite{BBBBG08}. In the positive case, the approximation
  scheme of Subsection~\ref{subsec:approx-reach} can therefore be
  applied as well.
\end{rk}

\subsection{Generalized semi-Markov processes}

A generalized semi-Markov process~\cite{BKK+11b, Glynn89} is a
stochastic process with a finite control, built on a set of events.
Each event is equipped with a random variable representing its
duration: an event can either be a variable-delay event, whose
duration is given by a probability distribution defined by a density
function, or be a fixed-delay event, modelled by a Dirac
distribution. A transition is characterized by a set of events which
expire, and schedules a set of new events. This model is known to
generalize CTMCs. In this section, we show how to exploit our
techniques to recover and generalize results from the literature on
quantitative verification of generalized semi-Markov processes.

\subsubsection{Definition}

\begin{definition}
  A \emph{generalized semi-Markov process} (GSMP) is a tuple
  $\calG=(Q, \calE, \ell, u, f, \bfE, \Succ)$ where
  \begin{itemize}
  \item $Q$ is a finite set of states; 
  \item $\calE = \{e_1,\dots,e_p\}$ is a finite set of events;
  \item $\ell : \calE \to \IN_{\ge 0}$ and $u : \calE \to
    \IN_{>0}\cup\lbrace\infty\rbrace$ are bounds such that for every
    $e \in \calE$, $\ell(e) \le u(e)$; 
  \item $f : \calE \to \Dist([\ell(e);u(e)])$ assigns distributions
    to every event $e \in \calE$;
  \item $\bfE: Q\rightarrow 2^{\calE}$ assigns to each state $q$ a set
    of events enabled (or active) in $q$;
  \item $\Succ:Q\times 2^{\calE} \rightarrow \Dist(Q)$ is the
    successor function defined for $(q, E)$ whenever $E\subseteq
    \bfE(q)$;
  \end{itemize}
\end{definition} 
Each event $e\in\calE$ has an upper (resp. lower) bound $u_e
\stackrel{\textrm{def}}{=} u(e)$ (resp. $\ell_e
\stackrel{\textrm{def}}{=} \ell(e)$) on its delay: the duration of
event $e$ is randomly chosen in the interval $[\ell_e,u_e]$ according
to density $f_e \stackrel{\textrm{def}}{=} f(e)$. In contrast to
\emph{fixed-delay} events, $e$ is called a \emph{variable-delay}
event, if $\ell_e< u_e$. Events can alternatively be seen as random
variables: with a variable-delay event is associated a density
function and with a fixed-delay event is associated the corresponding
Dirac distribution.

The semantics of a GSMP $\calG$ is given as an STS
$\calT_\calG=(S_\calG,\Sigma_\calG,\kappa_\calG)$.  There are two
points-of-view to define the semantics of $\calG$, one is through a
residual-time semantics using races between events~\cite{BKK+11b}
(clocks behave like in timed automata), and the other is to sample the
delay of an event once, when it is scheduled~\cite{AK05} (clocks are
``countdown''). Though the results of~\cite{BKK+11b} are stated using
the first convention, we prefer the second option, since it is easier
to understand the semantics. Note that the duality between the two
allows obviously to interpret the results of~\cite{BKK+11b} in our
setting.

Let $q \in Q$ be a state; a valuation $\nu \in
  (\IRpos^\perp)^{\calE}$, where
  $\IRpos^\perp=\IRpos\cup\lbrace\perp\rbrace$, is compatible with $q$
  whenever $\nu(e) = \perp$ if $e \notin \bfE(q)$, and $\nu(e) \in
  \IRpos$ otherwise; in the latter case, $\nu(e)$ is the remaining time
  for $e$ before expiring. Configurations of $\calG$ are then given
  by:
\[
S_\calG = \{(q,\nu) \in Q \times (\IRpos^\perp)^{\calE} \mid \nu\
\text{is compatible with}\ q\}.
\]

Let $\gamma = (q,\nu) \in S_\calG$ be a configuration, and define
$E_0(\gamma) = \{e \in \bfE(q) \mid \forall e' \in \bfE(q),\ \nu(e)
\le \nu(e')\}$ and $d(\gamma) = \nu(e)$ for $e \in E_0(\gamma)$.  From
configuration $\gamma$, there is a transition to configuration
$\gamma' = (q', \nu')$ on occurrence of the set of events
$E_0(\gamma)$ after delay $d(\gamma)$ whenever:
\[
\nu'(e) = \begin{cases}
  \perp & \textrm{ if } e\notin \bfE(q')\\
  \nu(e)-d(\gamma) & \textrm{ if } e\in (\bfE(q)\cap\bfE(q'))\setminus E_0(\gamma) \\
  t & \textrm{ otherwise, with $\ell_e \le t \le u_e$}.
\end{cases}
\]

The $\sigma$-algebra $\Sigma_{\calG}$ is obtained as the product
  between $2^{\calE}$ and the Borel $\sigma$-algebra on
  $(\IRpos^\perp)^{\calE}$.  Let $\gamma = (q,\nu)$, and $B = \{q'\}
\times B'$. Then, assuming $\Succ(q,E_0(\gamma))(q')>0$, we define the
Markov kernel $\kappa_\calG$ by:
\[
\kappa_\calG(\gamma,B) = \Succ(q,E_0(\gamma))(q') \cdot
\int_{(t_1,\dots,t_p) \in B'} \Big(\prod_{e \in \bfE(q')} g_e(t_e)
\Big)\ud t_{e_1} \dots \ud t_{e_p}
\]
where $g_e(t) = f_e(t)$ if $e \in \bfE(q') \setminus (\bfE(q)
\setminus E_0(\gamma))$; $g_e(t) = \delta_{\nu(e)-d(\gamma)}$ if $e
\in (\bfE(q) \cap \bfE(q')) \setminus E_0(\gamma)$; $g_e(t) =
\delta_\perp$ if $e \notin \bfE(q')$. In other words, for a newly
activated event $e$, its timestamp $t_e$ is sampled (independently
from the other events) according to density $f_e$; for events
inherited from the previous state, the delay which has elapsed is
applied (hence the Dirac distributions in the definition).

\begin{example} 
    Consider a two-machine network (call $M_1$ and $M_2$ the two
    machines), in which crash times (event denoted $\textsf{crash}_i$
    for machine $M_i$) follow an exponential distribution with
    parameter $\lambda_i$ ($\lambda_i \in \mathbb{R}_{>0}$) and reboot
    times (event denoted $\textsf{reboot}_i$ for machine $M_i$) follow
    a uniform distirbution over interval $[0,U_i]$ for some positive
    integer $U_i$. A GSMP model for the network is given on
    Figure~\ref{fig:ex:gsmp}.
    \begin{figure}
      \centering
      \begin{tikzpicture}
        \draw (0,0) node [draw,rounded corners=2pt] (UpUp)
        {\begin{tabular}{c} $M_1$ up \\ $M_2$ up \end{tabular}} node
        [left=.9cm] {$\{\textsf{crash}_1,\textsf{crash}_2\}$};
        \draw (8,0) node [draw,rounded corners=2pt]  (DownDown)
        {\begin{tabular}{c} $M_1$ down \\ $M_2$ down \end{tabular}} node
        [right=1cm] {$\{\textsf{reboot}_1,\textsf{reboot}_2\}$};;

        \draw (4,2) node [draw,rounded corners=2pt]  (UpDown)
        {\begin{tabular}{c} $M_1$ up \\ $M_2$ down \end{tabular}} node
        [above=.7cm] {$\{\textsf{reboot}_1,\textsf{crash}_2\}$};;
        \draw (4,-2) node [draw,rounded corners=2pt]  (DownUp)
        {\begin{tabular}{c} $M_1$ down \\ $M_2$ up \end{tabular}} node
        [below=.7cm] {$\{\textsf{crash}_1,\textsf{reboot}_2\}$};;

        \draw [-latex',rounded corners=3pt] (UpUp) |- (UpDown) node
        [pos=.8,above] {$\textsf{crash}_2$}; 
        \draw [-latex',rounded corners=3pt] (UpUp) |- (DownUp) node
        [pos=.8,below] {$\textsf{crash}_1$}; 
        \draw [-latex',rounded corners=3pt] (DownUp) -| (DownDown) node
        [pos=.2,below] {$\textsf{crash}_2$}; 
        \draw [-latex',rounded corners=3pt] (UpDown) -| (DownDown) node
        [pos=.2,above] {$\textsf{crash}_1$}; 
        \draw [-latex'] (UpDown) -- (UpUp) node [midway,above,sloped] {$\textsf{reboot}_2$};
        \draw [-latex'] (DownUp) -- (UpUp) node [midway,above,sloped] {$\textsf{reboot}_1$};
        \draw [-latex'] (DownDown) -- (DownUp)  node [midway,above,sloped] {$\textsf{reboot}_2$};
        \draw [-latex'] (DownDown) -- (UpDown) node
        [midway,above,sloped] {$\textsf{reboot}_1$};
        
        \draw [latex'-latex',dotted] (UpUp) -- (DownDown);
        \draw [latex'-latex',dotted] (DownUp) -- (UpDown);

        \path (5.5,-4) node {\begin{tabular}{r} 
          $f(\textsf{crash}_i)$ is an exponential distribution with
          parameter $\lambda_i$ \\ 
          $f(\textsf{reboot}_i)$ is a uniform distribution with
          support $[0,U_i]$ \end{tabular}};
      \end{tikzpicture}
      \caption{An example GSMP: a network with two machines \label{fig:ex:gsmp}}
    \end{figure}
    
    When machine $M_i$ is up (resp. down), a delay before event
    $\textsf{crash}_i$ (resp. $\textsf{reboot}_i$) occurs and is sampled
    according to distribution $f(\textsf{crash}_i)$
    (resp. $f(\textsf{reboot}_i)$). After that delay the event is
    triggered, unless it is preempted by some concurrent event. For
    instance, in state ``$M_1$ and $M_2$ up'' (leftmost state), delays
    for the two events $\textsf{crash}_1$ and $\textsf{crash}_2$ are
    sampled, and the shortest delay decides the next transition to be
    taken and the next state which is reached. Note that dotted arrows
    represent events that happen with probability $0$ (for instance,
    the very same delays are sampled for the two actions
    $\textsf{crash}_1$ and $\textsf{crash}_2$); we have omitted the
    corresponding labels (for readability).
\end{example}

We fix for the rest of this section a GSMP
$\calG = (Q,\calE,\ell,u,f,\bfE,\Succ)$, and
$\calT_\calG = (S_\calG,\Sigma_\calG,\kappa_\calG)$ its corresponding
STS.  To avoid too much technicalities, we assume that $\calG$ has no
fixed-delay events, that is, for every event $e$, $\ell_e<u_e$. What
we will present here would nevertheless extend to so-called
\emph{single-ticking} GSMPs~\cite{BKK+11b}.

\subsubsection{The refined region graph abstraction}

Due to the choice of the countdown-clock semantics (``clock values''
decrease down to $0$), the thick graph defined in
subsection~\ref{subsubsec:thickgraph} has to be twisted a
bit. Furthermore standard regions will not be fine enough to yield an
interesting abstraction. We will therefore refine regions using sets
of separated configurations, that we define now.

Let $\varepsilon>0$. We say that configuration $\gamma = (q,\nu)$ is
\emph{$\varepsilon$-separated} if for every $a,b \in \{0\} \cup
\{\{\nu(e)\} \mid e \in \bfE(q)\}$, either $a=b$ or
$|a-b|>\varepsilon$. We write $C_\calG^\varepsilon$ for the set of
$\varepsilon$-separated configurations.

The following lemma, stated and proven in~\cite{BKK+11b}, 
allows us to  find an adequate granularity for a refined region
abstraction.

\begin{lemma}[Lemma~1 of~\cite{BKK+11b}]
  \label{lemma1}
  There exists $\varepsilon>0$, $m \in \IN$ and $p>0$ such that for
  every $\gamma \in S_\calG$, $\Prob_{\delta_\gamma}(\F[\le m]
  C_\calG^\varepsilon) \ge p$.
\end{lemma}

We select $\varepsilon>0$ following Lemma~\ref{lemma1}, and
w.l.o.g. we assume $\varepsilon$ is of the form $\frac{1}{d}$ with $d
\in \IN_{>0}$.  We let $M_\calG$ be the maximal constant appearing in
constants $\{\ell_e \mid e \in \calE\}$ and $\{u_e \mid e \in \calE\
\text{and}\ u_e<\infty\}$. Each event $e \in \calE$ is virtually
assigned a clock variable $x_e$, and we consider a refinement of the
region equivalence for clocks $\{x_e \mid e \in \calE\}$
w.r.t. maximal constant $M_\calG$ and granularity $\frac{1}{d}$ as
follows. Two valuations $\nu,\nu' \in (\IRpos^\perp)^\calE$ are
equivalent whenever the following conditions hold:
\begin{enumerate}
\item for every $e \in \calE$, either both $\nu(e)$ and $\nu'(e)$ are
  stricly larger than $M_\calG$, or the integral parts of $d \cdot
  \nu(e)$ and $d \cdot \nu'(e)$ coincide;
\item for every $e_1,e_2 \in \calE$, for every $c \in \frac{1}{d}
  \cdot \IN \cap [-M_\calG;M_\calG]$, for every $\mathord{\bowtie} \in
  \{<,\le,=,\ge,>\}$, $\nu(e_1) - \nu(e_2) \bowtie c$ if and only if
  $\nu'(e_1) - \nu'(e_2) \bowtie c$.
\end{enumerate}
Note that the above conditions refine the ones given in
subsection~\ref{subsubsec:thickgraph} using diagonal constraints
(\cite{AD94}), and w.r.t. the granularity $\frac{1}{d}$ as well. We
write $R^\varepsilon_\calG$ for the set of equivalence classes, also
called regions. We realize that any region $r \in R^\varepsilon_\calG$
has either only $\varepsilon$-separated configurations, or only
non-$\varepsilon$-separated configurations. In particular,
$C_\calG^\varepsilon$ is a finite union of such regions.

We then define the abstraction $\alpha : Q \times \IRpos^\calE \to Q
\times R^\varepsilon_\calG$ by projection, and the finite Markov chain
$\calT_\calG^{\mathsf{rg},\varepsilon}$ as follows:
\begin{itemize}
\item its set of states is $Q \times R^\varepsilon_\calG$;
\item there is an edge from $(q,r)$ to $(q',r')$ whenever there exists
  $\nu \in r$ such that $\kappa_\calG((q,\nu),\{q'\} \times r')>0$;
\item from each state $(q,r) \in Q \times R^\varepsilon_\calG$, we
  associate the uniform distribution over $\{(q',r') \in Q \times
  R^\varepsilon_\calG \mid \text{there is an edge from}\ (q,r)\
  \text{to}\ (q',r')\}$.
\end{itemize}

Since $\calT_\calG^{\mathsf{rg},\varepsilon}$  is just a rescaling of a
standard region automaton, we immediately get:

\begin{lemma}
  $\calT_\calG^{\mathsf{rg},\varepsilon}$ is a finite
  $\alpha$-abstraction of $\calT_\calG$.
\end{lemma}
\noindent As previously, we notice that the above abstraction is
obviously complete (since it is finite).

\subsubsection{Analyzing GSMPs}
Let $A_\calG^\varepsilon = \{(q,r) \in Q \times R_\calG^\varepsilon \mid
\alpha^{-1}(q,r) \subseteq C_\calG^\varepsilon\}$. As a 
consequence of Lemma~\ref{lemma1} we get:

\begin{proposition}
  The set $\alpha^{-1}(A_\calG^\varepsilon)$ is a finite attractor for
  $\calT_\calG$.
\end{proposition}

Finally, as for STA and using~\cite[Lemma~2]{BKK+11b} (which allows to
prove that the conditions of Proposition~\ref{prop:attractorSound} are
actually satisfied), we also get:

\begin{proposition}
  $\calT_\calG^{\varepsilon,\mathsf{rg}}$ is a sound $\alpha$-abstraction
  of $\calT_\calG$.
\end{proposition}

As consequences, we get the following decidability and approximability
results for GSMPs:

\begin{corollary}
  Let $\calG$ be a labelled GSMP (with no fixed-delay event), and
  $\calM$ be a DMA. Then:
  \begin{enumerate}
  \item we can decide whether $\calG$ satisfies almost-surely $\calM$;
  \item for every initial distribution $\mu$ which is numerically
    amenable w.r.t. $\calG$,\footnote{We say that a distribution $\mu$
      is numerically amenable w.r.t. $\calG$ if, given $k \in \IN$,
      given $\varepsilon>0$ and given a sequence of states and refined
      regions $(q_0,r_0), (q_1,r_1), \ldots, (q_k,r_k)$, one can
      approximate
      $\Prob^{\calT_\calG}_\mu(\Cyl((q_0,r_0),(q_1,r_1),\dots,(q_k,r_k)))$
      up to $\varepsilon$.} we can compute arbitrary approximations of
    $\Prob_\mu^{\calT_\calG}(\calM)$.
  \end{enumerate}
\end{corollary}

\begin{proof}
  Again, the proof is similar to the ones of
  Corollaries~\ref{coro:staResults}
  and~\ref{coro:staOneClockResults}. We just notice that it is obvious
  that if $\calT_\calG$ is a GSMP with no fixed-delay events, then so
  is $\calT_\calG\ltimes\calM$.
\end{proof}

\begin{rk}
    We believe our approach gives new hints into the approximate
    quantitative model-checking of GSMPs, for which, up to our
    knowledge, only few results are known. For instance
    in~\cite{AB06,BA07}, the authors approximate the probability of
    until formulas of the form ``the system reaches a target before
    time $T$ within $k$ discrete events, while staying within a set of
    safe states'' (resp. ``the system reaches a target while staying
    within a set of safe states'') for GSMPs (resp. a restricted class
    of GSMPs), and study numerical aspects. However one can notice
    that the contributions of~\cite{AB06,BA07} are more the
    mathematical analysis of integral equations that need to be solved
    than convergence of approximation schemes.

    Our approach permits to have approximation algorithms (with an
    arbitrary precision) for reachability or until properties for
    single-ticking GSMPs, or time-bounded reachability or until
    properties for the class of GSMPs with no cycle of immediate
    events (immediate events are fixed-delay events with delay
    $0$).\footnote{We recall the discussion on non-Zeno real-time
      systems of Section~\ref{nonzeno} (page~\pageref{nonzeno}), and
      realize that such GSMPs are almost-surely non-Zeno (see
      Appendix~\ref{sec:appli}).}  The numerical aspects in our
    computations can be dealt with as in~\cite{AB06,BA07}.
\end{rk}

\subsection{Stochastic time Petri nets} \label{subsec:stpn} As a last
instantiation of our general framework, we briefly discuss a
particular class of general semi-Markov processes, namely GSMPs
induced by \emph{stochastic time Petri nets}
(STPNs)~\cite{MCB84}. Explaining how STPNs fit in our framework is
interesting on its own since, on the one hand they admit a simpler
abstraction than GSMPs, and on the other hand it allows us to compare
to existing work on the transient analysis of
STPNs~\cite{HPRV-peva12,PHV16}.  STPNs form a probabilistic and timed
extension of the well-known Petri nets. These models are natural to
represent concurrent systems in which events have random durations. An
STPN is defined by a finite set of places, and a finite set of
transitions, each equipped with a probability distribution over
$\IR_{\geq 0}$. The transitions in STPNs correspond to events in
GSMPs. Now, the set of enabled events (here enabled transitions), is
determined by the current marking. As for standard Petri nets, a
\emph{marking} maps places to natural numbers, representing the number
of tokens in each place of the net. A transition is then
\emph{enabled} when the marking contains at least one token per input
place. The difference with standard Petri nets lies in the choice of
which transition will fire and when. When a transition has just fired,
one token is consumed in each input place, one is added to each output
place, and for every newly enabled transition, a delay is fired
according to its associated probability distribution. The transition
with minimum delay fires after this delay elapses, possibly disabling
transitions or enabling new ones. Given an initial marking, the
semantics of an STPN is thus an uncountable stochastic transition
system where the states are tuples consisting of a marking, and a time
to fire (the remaining delay) for each enabled transition.

  In order to fit in our framework, we consider STPNs with the
  following two restrictions, as in~\cite{HPRV-peva12,PHV16}: on the
  one hand, the underlying Petri net is bounded, \emph{i.e.}  the
  number of reachable markings is finite; and on the other hand, the
  stochastic process defined by the STPN is Markov regenerative. Let
  us explain in details this last assumption. When using arbitrary
  distributions to model the random duration of the transitions in the
  STPN, the remaining delays of transitions that stay enabled (when a
  transition just fired) have to be stored. Indeed, they will be
  compared with newly sampled durations to determine the next minimum
  delay, and thus impact the time to elapse as well as the next
  transition to fire. A \emph{regeneration point} in the stochastic
  process defined by an STPN is a state for which the remaining
  duration of each enabled transition is non-zero only for transitions
  with exponentially distributed durations~\cite{HPRV-peva12}
(in~\cite{PHV16}, the condition is slightly more general, but
    more technical to explain, hence we focus on~\cite{HPRV-peva12}).
  For such states, the non-zero remaining durations can indeed be
  forgotten and re-sampled using the same exponential distribution
  without altering the STPN semantics, thanks to the memoryless
  property of exponential distributions. An STPN is said to be
  \emph{Markov regenerative} when almost surely regeneration points
  are encountered infinitely often with probability $1$. Under our
  assumptions of boundedness of the Petri net, and Markov
  regeneration, clearly enough the STPN admits a finite attractor,
  namely the set of its regeneration points.

  To apply our results to stochastic processes generated by STPNs, we
  first identify a finite state abstraction which will happen to be a
  sound abstraction. We let $\calT_1$ be the stochastic process
  defined by an STPN, and define a finite-state Markov chain
  $\calT_2$, that corresponds to the state-class graph (see
  \emph{e.g.}~\cite{HPRV-peva12}) equipped with uniform discrete
  distributions. More precisely, the states of the abstraction
  $\calT_2$ consist of state classes, that are the equivalent of
  regions for time Petri nets (we consider state classes just before
  sampling newly enabled transitions): state classes gather
  configurations with same marking and same ordering on remaining
  delays for transitions. Given that Petri nets are bounded, and their
  set of transitions is finite, the state space of $\calT_2$ is
  finite. Some states correspond to regeneration points: those in
  which the only transitions with non-zero remaining delay are
  exponentially distributed; in that case we assume w.l.o.g. that they
  will be resampled, which allows one to describe such a state in
  $\calT_2$ using only its marking. Further, writing $\alpha$ for the
  abstraction, to build $\calT_2$, there is a transition from state
  $s$ to state $t$ as soon as there exists a transition in $\calT_1$
  between some state of $\alpha^{-1}(s)$ and $\alpha^{-1}(t)$. If
  $A = \{s_1, \cdots, s_n\}$ denotes the set of regeneration points,
  it is important to realize that under this abstraction $\alpha$, for
  every $s_i \in A$, $\alpha^{-1}(s_i) = s_i$. Due to the hypothesis
  that regeneration points are encountered almost-surely infinitely
  often, the set $A$ is an attractor both in the concrete stochastic
  process, and in the finite-state abstraction. Moreover, for every
  $s_i \in A$ and every $\alpha$-closed set $B$ of the concrete
  stochastic process $\calT_1$ such that there is a path from $s_i$ to
  $B$, we let $k_{i,B}$ be its length and $p_{i,B}$ be its
  probability, so that
  $\Prob^{\calT_1}_{\delta_{s_i}}(\F[\leq k_{i,B}] B)\geq
  p_{i,B}$. Letting $k = \max_{i,B} k_{i,B}$ and
  $p = \min_{i,B} p_{i,B}$, we obtain that for every $s_i \in A$ and
  every $\alpha$-closed set $B$, either
  $\Prob^{\calT_1}_{s_i}(\F[\leq k] B)\geq p$, or
  $\Prob^{\calT_1}_{s_i}(\F[\leq k] B) =0$.  Therefore, we can apply
  Proposition~\ref{prop:attractorSound} (see
  page~\pageref{prop:attractorSound}), to derive that $\calT_1$ is
  decisive w.r.t. every $\alpha$-closed set. Further, by
  Proposition~\ref{coro:DecSound}, we obtain that $\calT_2$ is a sound
  $\alpha$-abstraction of $\calT_1$.

  Similarly to the general case of GSMPs, we can thus conclude that
  the approximate quantitative model checking of STPNs can be done,
  provided numerical computations are amenable.

  In~\cite{PHV16}, the authors provide a technique to perform the
  quantitative verification of time-bounded until properties in
  bounded STPNs that are Markov regenerative, with the assumption of a
  bound on the number of transitions between regeneration points (we
  call it strong Markov regenerative). As discussed above, we can
  relax the hypothesis on the bounded number of transitions between
  regeneration points, since the mere Markov regeneration property
  suffices to guarantee the existence of an attractor. A simple
  criterion to ensure Markov regeneration is to assume that
  every cycle in the state class graph contains a regeneration point
  (it is in fact equivalent to the strong Markov regeneration
  property). Note also that already the (weak) Markov regeneration
  implies that the stochastic process is almost-surely non-Zeno (since
  almost-surely states with exponentially distributed events are
  visited infinitely often), and therefore the approximation schemes
  for time-bounded reachability or until properties apply. For a fair
  comparison with the work of~\cite{HPRV-peva12,PHV16}, note however
  that their contribution focuses on \emph{efficient} numerical
  computations of probabilities for time-bounded reachability or until
  properties. In contrast, we totally ignore the efficiency of
  probability computations and focus on sufficient conditions that
  enable our general framework to apply.

\section{A guided tour of STSs}
\label{sec:guide}

We now give an overview of the results presented in this paper. For
improving readability, not all precise statements are listed. For
instance, we omit the results which assume a fixed initial
distribution.  Also, few notations are borrowed from the paper, yet
the global picture is almost self-contained.

The guide should be read as follows. Given an STS $\calT$ and a
property $\varphi$, Figures~\ref{fig:qualitative}
and~\ref{fig:quantitative} provide the assumptions on $\calT$ to be
able to perform the qualitative or quantitative analysis of $\varphi$
on $\calT$. Note that when we consider abstractions $\calT_1
\xrightarrow{\alpha} \calT_2$, then we assume $\calT_1 = \calT$.
Then, Figures~\ref{fig:properties},~\ref{fig:transfer}
and~\ref{fig:abstraction} summarize the relationships between the various
notions. They should be used to know how to prove the properties that
are expected of the model, either directly or via an abstraction
(which needs to be designed).

\begin{figure}[h]
  \pgfdeclarelayer{background}
    \pgfsetlayers{background,main}
      \begin{tikzpicture}
        \node (hyp) {$\left.\begin{tabular}{r} $\calT$ satisfying
            decisiveness properties \\ $\varphi$ (repeated)
            reach. property \end{tabular}\right\}$};

        \node (conc) [right=of hyp,xshift=.5cm] {\begin{tabular}{l}
            simple $0$-reachability property on $\calT$ \end{tabular}}; 

        \draw [decorate, decoration={snake,amplitude=.05cm},->] (hyp) -- (conc) node
        [midway,below] {{\footnotesize Sec.~\ref{subsec:qual-basic}}};

      \begin{pgfonlayer}{background}
        \node[fill=LemonChiffon,inner sep=0pt,rounded
        corners,fit=(hyp) (conc)] {};
      \end{pgfonlayer}
      \end{tikzpicture}

      \medskip
      \begin{tikzpicture}
        \node (hyp) {$\left.\begin{tabular}{r} $\calT_1
            \xrightarrow{\alpha} \calT_2$ sound and complete abst. \\
            $\calT_1$ satisfying decisiveness properties \\
            w.r.t. $\alpha$-closed sets \\
            $\varphi$ (repeated)
            reach. propert
          \end{tabular}\right\}$};
        \node (conc) [right=of hyp,xshift=.5cm]{\begin{tabular}{l}
            simple $0$-reachability property on $\calT_2$ 
          \end{tabular}};
        \draw [decorate, decoration={snake,amplitude=.05cm},->] (hyp) -- (conc) node
        [midway,below]  {{\footnotesize Sec.~\ref{subsec:qual-basic-abs}}};
      \begin{pgfonlayer}{background}
        \node[fill=LemonChiffon,inner sep=0pt,rounded
        corners,fit=(hyp) (conc)] {};
      \end{pgfonlayer}
      \end{tikzpicture}

      \medskip
      \begin{tikzpicture}
        \node (hyp) {$\left.\begin{tabular}{r} $\calT$ DMC with
        finite attractor \\ $\varphi$ given by  (det.) automaton $\calM$ \end{tabular}\right\}$};

  \node (conc) [right=of hyp,xshift=.5cm] {\begin{tabular}{l}
      almost-sure reachability property in $\calT \ltimes \calM$ \\ of
      states given by the
      abstract graph \\ of the attractor of $\calT \ltimes
      \calM$ \end{tabular}};

        \draw [decorate, decoration={snake,amplitude=.05cm},->] (hyp) -- (conc) node
        [midway,below] {{\footnotesize Sec.~\ref{subsec:qual-dmc}}};
      \begin{pgfonlayer}{background}
        \node[fill=LemonChiffon,inner sep=0pt,rounded
        corners,fit=(hyp) (conc)] {};
      \end{pgfonlayer}
      \end{tikzpicture}

      \medskip
      \begin{tikzpicture}
        \node (hyp) {$\left.\begin{tabular}{r}
        $\calT_1$ STS and  $\calT_1 \xrightarrow{\alpha} \calT_2$ abst. \\
        $\calT_2$ DMC with finite attractor \\
        $\varphi$ given by  (det.) automaton $\calM$ \\
        $\calT_1 \ltimes \calM \xrightarrow{\alpha_\calM} \calT_2
        \ltimes \calM$ sound abst.
      \end{tabular}\right\}$};

  \node (conc) [right=of hyp,xshift=.5cm] {\begin{tabular}{l}
      almost-sure reachability property in $\calT_1 \ltimes \calM$ \\ of
      some states given by the abstract graph \\ of the attractor of
      $\calT_2 \ltimes \calM$ \end{tabular}};

        \draw [decorate, decoration={snake,amplitude=.05cm},->] (hyp) -- (conc) node
        [midway,below] (A) {{\footnotesize Sec.~\ref{subsec:qual-abs}}};
      \begin{pgfonlayer}{background}
        \node[fill=LemonChiffon,inner sep=0pt,rounded
        corners,fit=(hyp) (conc)] {};
      \end{pgfonlayer}
      \end{tikzpicture}
      \caption{Qualitative analysis
        : given $\calT$ an STS and $\varphi$ a property, decide
        whether $\Prob^\calT(\varphi) =0$ or $=1$. $\calT$ is replaced
        with $\calT_1$ in case of an abstraction.  The edge
        $\tikz{\draw [decorate, decoration={snake,amplitude=.05cm},->]
          (0,0) -- (1,0);}$ reads ``amounts
        to''. \label{fig:qualitative}}
\end{figure}

\begin{figure}[h]

  \pgfdeclarelayer{background}
    \pgfsetlayers{background,main}
      \begin{tikzpicture}
        \node (hyp) {$\left.\begin{tabular}{r} $\calT$ satisfying
            decisiveness properties \\ $\varphi$ (repeated)
            reach. property \end{tabular}\right\}$};

        \node (conc) [right=of hyp,xshift=.5cm] {\begin{tabular}{c}
            approx. scheme \end{tabular}}; 

        \draw [decorate, decoration={snake,amplitude=.05cm},->] (hyp) -- (conc) node
        [midway,below] (A) {\begin{tabular}[t]{c }{\footnotesize
              Sec.~\ref{subsec:approx-reach}} \\[-.1cm] {\footnotesize
              and~\ref{subsec:quant_repeated_reach}} \end{tabular}};
      \begin{pgfonlayer}{background}
        \node[fill=LemonChiffon,inner sep=0pt,rounded
        corners,fit=(hyp) (conc) (A)] {};
      \end{pgfonlayer}
      \end{tikzpicture}

      \medskip
      \begin{tikzpicture}
        \node (hyp) {$\left.\begin{tabular}{c} $\calT$ DMC with
        finite attractor \\ $\varphi$ given by (det.) automaton $\calM$ \end{tabular}\right\}$};

        \node (conc) [right=of hyp,xshift=.5cm] {\begin{tabular}{l}
            approx. scheme on $\calT \ltimes \calM$ applied to a \\
            reach. property given by the
          abstract graph \\ of the attractor of $\calT \ltimes \calM$ \end{tabular}}; 

        \draw [decorate, decoration={snake,amplitude=.05cm},->] (hyp) -- (conc) node
        [midway,below] {{\footnotesize Sec.~\ref{subsec:quant-attractor}}};
      \begin{pgfonlayer}{background}
        \node[fill=LemonChiffon,inner sep=0pt,rounded
        corners,fit=(hyp) (conc)] {};
      \end{pgfonlayer}
      \end{tikzpicture}

\medskip
      \begin{tikzpicture}
        \node (hyp) {$\left.\begin{tabular}{r}
        $\calT_1$ STS and  $\calT_1 \xrightarrow{\alpha} \calT_2$ \\
        $\calT_2$ DMC with finite attractor \\
        $\varphi$ given by (det.) automaton $\calM$ \\ 
        $\calT_1 \ltimes \calM \xrightarrow{\alpha_\calM} \calT_2
        \ltimes \calM$ sound abst.        
      \end{tabular}\right\}$};

    \node (conc) [right=of hyp,xshift=.5cm] {\begin{tabular}{l}
        approx. scheme on $\calT_1 \ltimes \calM$
        applied to  a \\ reach. 
        property given by  the abstract graph  \\ of the attractor of $\calT_2
        \ltimes \calM$ \end{tabular}};

        \draw [decorate, decoration={snake,amplitude=.05cm},->] (hyp) -- (conc) node
        [midway,below] {{\footnotesize Sec.~\ref{subsec:quantMullerAbstr}}};
      \begin{pgfonlayer}{background}
        \node[fill=LemonChiffon,inner sep=0pt,rounded
        corners,fit=(hyp) (conc)] {};
      \end{pgfonlayer}
      \end{tikzpicture}

\medskip
      \begin{tikzpicture}
        \node (hyp) {$\left.\begin{tabular}{r}
              $\calT$ RT-STS a.s. non-zeno \\
              time-bounded until or reach. property      
            \end{tabular}\right\}$};

        \node (conc) [right=of hyp,xshift=.5cm] {\begin{tabular}{l}
            natural attractors $A_\Delta$, \\
            hence approx. scheme \end{tabular}};

        \draw [decorate, decoration={snake,amplitude=.05cm},->] (hyp)
        -- (conc) node 
        [midway,below] {{\footnotesize Sec.~\ref{subsec:bounded}}};
      \begin{pgfonlayer}{background}
        \node[fill=LemonChiffon,inner sep=0pt,rounded
        corners,fit=(hyp) (conc)] {};
      \end{pgfonlayer}
      \end{tikzpicture}
      \caption{Quantitative analysis
        (Section~\ref{sec:quantitative}): given $\calT$ an STS and $\varphi$ a property, 
        compute (or approximate) $\Prob^\calT(\varphi)$. In case of the
        abstraction, $\calT_1 = \calT$. The edge $\tikz{\draw [decorate, decoration={snake,amplitude=.05cm},->] (0,0) --
        (1,0);}$ reads ``amounts to''. \label{fig:quantitative}}

\end{figure}

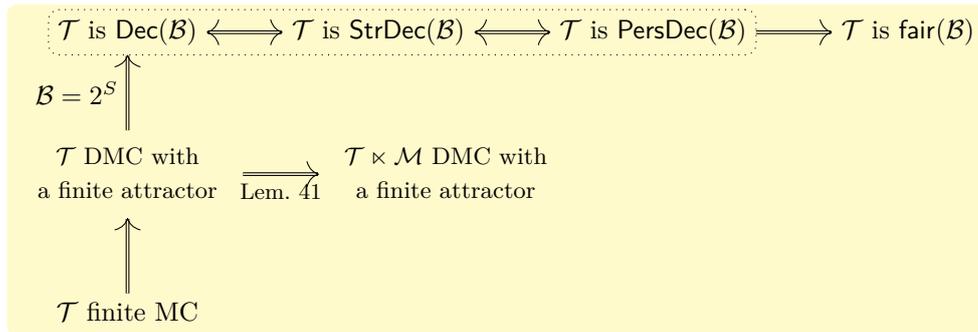
\begin{figure}[h]
    \pgfdeclarelayer{background}
    \pgfsetlayers{background,main}
  \begin{tikzpicture}
    \begin{scope}
      \node (decisive) {$\calT$ is $\D(\calB)$};
      \node (strongly_decisive) [right=of decisive] {$\calT$ is
        $\SD(\calB)$};
      \node (persistently_decisive) [right=of strongly_decisive]
      {$\calT$ is $\PD(\calB)$};
      \node (fair) [right=of persistently_decisive] {$\calT$ is
        $\fair(\calB)$};

      \draw [double=LemonChiffon,<->] (decisive) -- (strongly_decisive);
      \draw [double=LemonChiffon,<->] (strongly_decisive) --
      (persistently_decisive);
      \draw [double=LemonChiffon,->] (persistently_decisive) -- (fair);

      \node (attractor) [below=of decisive] {\begin{tabular}{c}
          $\calT$ DMC with \\ a finite  attractor \end{tabular}};
      \node (finite) [below=of attractor] {$\calT$ finite MC};

      \node (produit) [right=of attractor] {\begin{tabular}{c} $\calT
          \ltimes \calM$ DMC with \\ a finite
          attractor \end{tabular}};

      \draw [double=LemonChiffon,->] (attractor) -- (decisive) node [midway,left] (B)
      {$\calB = 2^S$};
      \draw [double=LemonChiffon,->] (finite) -- (attractor);
      \draw [double=LemonChiffon,->] (attractor) -- (produit) node[midway, below] {{\footnotesize Lem.~\ref{lemma:soundproduct}}};

      \begin{pgfonlayer}{background}
        \node[fill=LemonChiffon,inner sep=2pt,rounded
        corners,fit=(decisive)
        (strongly_decisive) (persistently_decisive) (produit) (finite)
        (fair) (attractor) (B)] {};
        \node[draw,dotted,inner sep=0pt,rounded corners,fit=(decisive)
        (strongly_decisive) (persistently_decisive)] {};
      \end{pgfonlayer}
    \end{scope}
  \end{tikzpicture}
  \caption{Properties of STS (Section~\ref{subsec:link} and
    Lemma~\ref{lemma:soundproduct})\label{fig:properties}}
\end{figure}

\begin{figure}[h]

    \pgfdeclarelayer{background}
    \pgfsetlayers{background,main}
    \begin{tikzpicture}
      \begin{scope}
        \node (hyp) {$\left.\begin{tabular}{r}
              $\calT_2$ DMC with finite attractor \\
       $\calT_1$ satisfying $(\dag)$ (cf
       page~\pageref{prop:attractorSound}) \\
       $\calT_1 \xrightarrow{\alpha} \calT_2$ abstraction
     \end{tabular}\right\}$};
   
   \node [right=of hyp,xshift=.5cm]  (conc) {\begin{tabular}{l} $\calT_1$ decisive
       w.r.t. $\alpha$-closed sets \end{tabular}};
   \draw [double=LemonChiffon,->] (hyp) -- (conc) node [midway,below] {{\footnotesize Prop.~\ref{prop:attractorSound}}};
   
   \node [below=of hyp,yshift=1cm] (hyp2) {$\left.\begin{tabular}{r}
         $\calT_2$ finite MC \\
         $\calT_1$ fair w.r.t. $\alpha$-closed sets \\
       $\calT_1 \xrightarrow{\alpha} \calT_2$ abstraction
       \end{tabular}\right\}$};
   
   \node [right=of hyp2,xshift=.5cm]  (conc2) {\begin{tabular}{l} $\calT_1$ decisive
       w.r.t. $\alpha$-closed sets \end{tabular}};
   \draw [double=LemonChiffon,->] (hyp2) -- (conc2) node [midway,below]
   {{\footnotesize Prop.~\ref{prop:fairness}}};
   
   \begin{pgfonlayer}{background}
     \node[fill=LemonChiffon,inner sep=0pt,rounded corners,fit=(hyp) (conc)
     (hyp2) (conc2)] {};
   \end{pgfonlayer}
 \end{scope}
\end{tikzpicture}

\bigskip
  \pgfdeclarelayer{background}
    \pgfsetlayers{background,main}
    \begin{tikzpicture}
      \begin{scope}
        \node (T2decisive) 
        {$\left.\begin{tabular}{r}$\calT_2$ decisive \\ $\calT_1 \xrightarrow{\alpha}
          \calT_2$ \textbf{sound} abstraction \end{tabular}\right\}$};
        \node (T1decisive) [right=of T2decisive,xshift=.5cm]
        {\begin{tabular}{r}$\calT_1$ decisive w.r.t. $\alpha$-closed
            sets \end{tabular}};
        
        \draw [double=LemonChiffon,->] (T2decisive) -- (T1decisive) node [midway,below]
        {{\footnotesize Prop.~\ref{coro:SoundDecisive}}};
        
        \node (attractor2) [below=of T2decisive,yshift=1cm] {$\left.\begin{tabular}{r}$A_2$ attractor of
            $\calT_2$ \\ $\calT_1 \xrightarrow{\alpha}
          \calT_2$ \textbf{sound} abstraction
          \end{tabular}\right\}$};
        \node (attractor1) [right=of attractor2,xshift=.5cm]
        {\begin{tabular}{r}$\alpha^{-1}(A_2)$ attractor of $\calT_1$ \end{tabular}};
        
        \draw [double=LemonChiffon,->] (attractor2) -- (attractor1) node
        [midway,below] (C)
        {{\footnotesize Prop.~\ref{lem:attr-via-sound}}};
      \end{scope}
      
      \begin{pgfonlayer}{background}
        \node[fill=LemonChiffon,inner sep=0pt,rounded
        corners,fit= (T2decisive) (T1decisive)
        (attractor2) (attractor1) (C)] {};
      \end{pgfonlayer}
    \end{tikzpicture}
    
\caption{Transfer of properties via abstractions \label{fig:transfer}}
\end{figure}
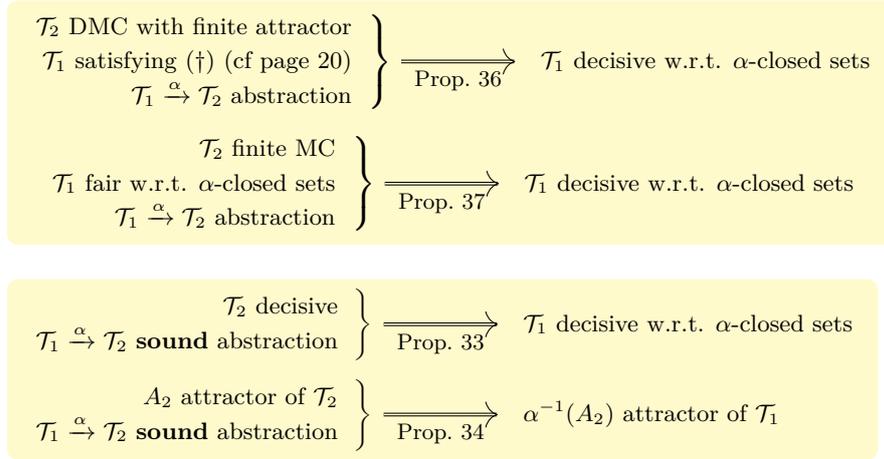

\begin{figure}[h]
{
  \pgfdeclarelayer{background}
    \pgfsetlayers{background,main}
  \begin{tikzpicture}
    \node (alpha) [right] {\begin{tabular}{l} $\calT_1
        \xrightarrow{\alpha} \calT_2$
      abstraction \\ $\calT_1 \ltimes \calM
      \xrightarrow{\alpha_\calM} \calT_2 \ltimes \calM$ abstraction
    \end{tabular}};
    
    \path (6,0) node [right] (danger) {\danger  \begin{tabular}[t]{c}
        Soundness/completeness of $\alpha$ does not imply \\
        soundness/completeness of $\alpha_\calM$! \end{tabular}};

      \begin{pgfonlayer}{background}
        \node[fill=magenta!10,inner sep=0pt,rounded
        corners,fit=(danger) (alpha)] {};
      \end{pgfonlayer}
  \end{tikzpicture}

\smallskip
  \textit{Condition for completeness:} \\
  \pgfdeclarelayer{background}
    \pgfsetlayers{background,main}
  \begin{tikzpicture}
    \node (decisive) [right=of finite,xshift=.5cm] {\begin{tabular}{c}
        $\calT_2$ decisive DMC \\ $\calT_1
        \xrightarrow{\alpha} \calT_2$ abstraction\end{tabular}};
    \node (complete) [right=of decisive,xshift=.5cm]
    {\begin{tabular}{c}$\calT_1 \xrightarrow{\alpha} \calT_2$
        complete abstraction
      \end{tabular}};

    \draw [double=LemonChiffon,->] (decisive) -- (complete) node [midway,below] (D)
    {{\footnotesize Lem.~\ref{lemma:completeness}}};

      \begin{pgfonlayer}{background}
        \node[fill=LemonChiffon,inner sep=0pt,rounded
        corners,fit= (decisive) (complete) (D)] {};
      \end{pgfonlayer}
\end{tikzpicture}

\medskip
\textit{Condition for soundness:} \\
  \pgfdeclarelayer{background2}
    \pgfsetlayers{background2,main}
\begin{tikzpicture}
  \node (decisive) {$\left.\begin{tabular}{r} $\calT_2$ DMC
        and  $\calT_1$ decisive \\
      w.r.t. $\alpha$-closed sets \\ $\calT_1 \xrightarrow{\alpha}
      \calT_2$ abstraction \end{tabular}\right\}$};
    \node (sound) [right=of decisive,xshift=.5cm]
    {\begin{tabular}{c}$\calT_1 \xrightarrow{\alpha} \calT_2$ sound abstraction
      \end{tabular}};

    \draw [double=LemonChiffon,->] (decisive) -- (sound) node [midway,below]
    {{\footnotesize Prop.~\ref{coro:DecSound}}};

      \begin{pgfonlayer}{background2}
        \node[fill=LemonChiffon,inner sep=0pt,rounded
        corners,fit=(decisive) (sound)] {};
      \end{pgfonlayer}
  \end{tikzpicture}
\caption{Completeness and soundness of abstractions} \label{fig:abstraction}}

\end{figure}

\section{Conclusion}
This paper deals with general stochastic transition systems (hence
possibly continuous state-space Markov chains). We defined abstract
properties of such stochastic processes, which allow one to design
general procedures for their qualitative or quantitative analysis.
The effectiveness of the approach requires some effectiveness assumption
on specific high-level formalisms that are used to describe the
stochastic process. We have demonstrated the effectiveness of the
approach on three classes of systems: stochastic timed automata,
generalized semi-Markov processes and stochastic time Petri nets. In
both cases, we recover known results; but our approach yields further
approximability results, which, up to our knowledge, are new.

We believe that, more importantly, we provide in this paper a
methodology to understand stochastic models from a verification and
algorithmics point-of-view.  Section~\ref{sec:guide} gives a
high-level description of our results, and of properties that should
be satisfied by the stochastic model in order to apply our
algorithms. In many cases, we showed that the hypotheses were really
necessary to get the expected results, by providing counter-examples
when the hypotheses are relaxed.

As future work, we plan to investigate new applications, such as
infinite-state systems that occur in parameterized
verification. Applying our results to stochastic hybrid systems is
very tempting. However, our approach heavily relies on the existence
of a sound finite or countable abstraction of the system. In case of
stochastic timed automata, we noticed that although a finite
time-abstract bisimulation always exists, decisiveness is not
guaranteed (see Example~\ref{Example:pacman}). It thus seems hard to
identifiy decisive subclasses of stochastic hybrid systems: most
often, the underlying hybrid system does not admit a finite or
countable time-abstract bisimulation~\cite{Hen96}. Nevertheless, as
mentionned at the end of Section~\ref{sec:quantitative}, one can
consider revisiting results on the time-bounded analysis of stochastic
hybrid systems (see \emph{e.g.}~\cite{CDKM11,SMA16}).
 
Also, we would like to adopt a similar generic approach for processes
with non-determinism like Markov decision processes, or even
stochastic two-player games.

Finally, let us mention that our approach could be interpreted in the
context of stochastic relations~\cite{doberkat07}, and that, for
instance, the pushforward $\alpha_{\#}$ corresponds to the Giry monad
applied to the abstraction $\alpha$. We thank an anonymous reviewer to
point us~\cite{doberkat07}, and believe this may inspire further work
to better understand abstractions of STSs.


\medskip\noindent \textbf{Acknowledgement.} We would like to warmly
thank the anonymous referees, who provided exceptionally detailed
reviews which greatly helped improving the paper.



\begin{thebibliography}{10}

\bibitem{ABM07}
Parosh~A. Abdulla, Noomene {Ben Henda}, and Richard Mayr.
\newblock Decisive {M}arkov chains.
\newblock {\em Logical Methods in Computer Science}, 3(4), 2007.

\bibitem{ABRS05}
Parosh~A. Abdulla, Nathalie Bertrand, Alexander Rabinovich, and {\relax
  Ph}ilippe Schnoebelen.
\newblock Verification of probabilistic systems with faulty communication.
\newblock {\em Information and Computation}, 202(2):141--165, 2005.

\bibitem{AAGT-lics12}
Manindra Agrawal, S.~Akshay, Blaise Genest, and P.~S. Thiagarajan.
\newblock Approximate verification of the symbolic dynamics of {M}arkov chains.
\newblock In {\em Proc. 27th Annual Symposium on Logic in Computer Science
  (LICS'12)}. {IEEE} Computer Society, 2012.

\bibitem{alur91}
Rajeev Alur.
\newblock {\em Techniques for Automatic Verification of Real-Time Systems}.
\newblock PhD thesis, Stanford University, Stanford, CA, USA, 1991.

\bibitem{AB06}
Rajeev Alur and Mikhail Bernadsky.
\newblock Bounded model checking for {GSMP} models of stochastic real-time
  systems.
\newblock In {\em Proc. 9th International Workshop on Hybrid Systems:
  Computation and Control (HSCC'06)}, volume 3927 of {\em Lecture Notes in
  Computer Science}, pages 19--33. Springer, 2006.

\bibitem{AD94}
Rajeev Alur and David~L. Dill.
\newblock A theory of timed automata.
\newblock {\em Theoretical Computer Science}, 126(2):183--235, 1994.

\bibitem{BBBBG07}
Christel Baier, Nathalie Bertrand, Patricia Bouyer, {\relax Th}omas Brihaye,
  and Marcus Gr{\"o}{\ss}er.
\newblock Probabilistic and topological semantics for timed automata.
\newblock In {\em Proc. 27th Conference on Foundations of Software Technology
  and Theoretical Computer Science (FSTTCS'07)}, volume 4855 of {\em Lecture
  Notes in Computer Science}, pages 179--191. Springer, 2007.

\bibitem{BBBBG08}
Christel Baier, Nathalie Bertrand, Patricia Bouyer, {\relax Th}omas Brihaye,
  and Marcus Gr{\"o}{\ss}er.
\newblock Almost-sure model checking of infinite paths in one-clock timed
  automata.
\newblock In {\em Proc. 23rd Annual Symposium on Logic in Computer Science
  (LICS'08)}, pages 217--226. IEEE Computer Society Press, 2008.

\bibitem{BHHK03}
Christel Baier, Boudewijn Haverkort, Holger Hermanns, and Joost-Pieter Katoen.
\newblock Model-checking algorithms for continuous-time {M}arkov chains.
\newblock {\em IEEE Transactions on Software Engineering}, 29(7):524--541,
  2003.

\bibitem{BK08}
Christel Baier and Joost-Pieter Katoen.
\newblock {\em Principles of model checking}.
\newblock MIT Press, 2008.

\bibitem{BK98}
Christel Baier and Marta~Z. Kwiatkowska.
\newblock On the verification of qualitative properties of probabilistic
  processes under fairness constraints.
\newblock {\em Information Processing Letters}, 66(2):71--79, 1998.

\bibitem{BRS-csl02}
Dani{\`{e}}le Beauquier, Alexander~Moshe Rabinovich, and Anatol Slissenko.
\newblock A logic of probability with decidable model-checking.
\newblock In {\em Proc. 16th International Workshop on Computer Science Logic
  (CSL'02)}, volume 2471 of {\em Lecture Notes in Computer Science}, pages
  306--321. Springer, 2002.

\bibitem{BA07}
Mikhail Bernadsky and Rajeev Alur.
\newblock Symbolic analysis for {GSMP} models with one stateful clock.
\newblock In {\em Proc. 10th International Workshop on Hybrid Systems:
  Computation and Control (HSCC'07)}, volume 4416 of {\em Lecture Notes in
  Computer Science}, pages 90--103. Springer, 2007.

\bibitem{bertrand06}
Nathalie Bertrand.
\newblock {\em Mod{\`e}les stochastiques pour les pertes de messages dans les
  protocoles asynchrones et techniques de v{\'e}rification automatique.}
\newblock PhD thesis, {\'E}cole Normale Sup{\'e}rieure de Cachan, Cachan,
  France, 2006.

\bibitem{BBBC16}
Nathalie Bertrand, Patricia Bouyer, {\relax Th}omas Brihaye, and Pierre
  Carlier.
\newblock Analysing decisive stochastic processes.
\newblock In {\em Proc. 43rd International Colloquium on Automata, Languages
  and Programming (ICALP'16)}, volume~55 of {\em LIPIcs}, pages 101:1--101:14.
  Leibniz-Zentrum f{\"u}r Informatik, 2016.

\bibitem{BBBM08}
Nathalie Bertrand, Patricia Bouyer, {\relax Th}omas Brihaye, and Nicolas
  Markey.
\newblock Quantitative model-checking of one-clock timed automata under
  probabilistic semantics.
\newblock In {\em Proc. 5th International Conference on Quantitative Evaluation
  of Systems (QEST'08)}, pages 55--64. IEEE Computer Society Press, 2008.

\bibitem{BBB+14}
Nathalie Bertrand, Patricia Bouyer, Thomas Brihaye, Quentin Menet, Christel
  Baier, Marcus Gr{\"o}{\ss}er, and Marcin Jurdzi{\'n}ski.
\newblock Stochastic timed automata.
\newblock {\em Logical Methods in Computer Science}, 10(4):1--73, 2014.

\bibitem{BBJM12}
Patricia Bouyer, Thomas Brihaye, Marcin Jurdzi{\'n}ski, and Quentin Menet.
\newblock Almost-sure model-checking of reactive timed automata.
\newblock In {\em Proc. 9th International Conference on Quantitative Evaluation
  of Systems (QEST'12)}, pages 138--147. IEEE Computer Society Press, 2012.

\bibitem{BKK+11b}
Tom{\'a}\v{s} Br{\'a}zdil, Jan Kr\v{c}{\'a}l, Jan K\v{r}et{\'i}nsk{\'y}, and
  Vojt\v{e}ch {\v{R}}eh{\'a}k.
\newblock Fixed-delay events in generalized semi-{M}arkov processes revisited.
\newblock In {\em Proc. 22nd International Conference on Concurrency Theory
  (CONCUR'11)}, volume 6901 of {\em Lecture Notes in Computer Science}, pages
  140--155. Springer, 2011.

\bibitem{Bremaud}
Pierre Brémaud.
\newblock {\em {Markov Chains}}.
\newblock Springer, 1999.

\bibitem{CDKM11}
Taolue Chen, Marco Diciolla, Marta~Z. Kwiatkowska, and Alexandru Mereacre.
\newblock Time-bounded verification of {CTMCs} against real-time
  specifications.
\newblock In {\em Proc. 9th International Conference on Formal Modeling and
  Analysis of Timed Systems (FORMATS'11)}, volume 6919 of {\em Lecture Notes in
  Computer Science}, pages 26--42. Springer, 2011.

\bibitem{DDP03}
Vincent Danos, Jos{\'e}e Desharnais, and Prakash Panangaden.
\newblock Conditional expectation and the approximation of labelled {M}arkov
  processes.
\newblock In {\em Proc. 14th International Conference on Concurrency Theory
  (CONCUR'03)}, volume 2761 of {\em Lecture Notes in Computer Science}, pages
  468--482. Springer, 2003.

\bibitem{AK05}
Pedro~R. D{'A}rgenio and Joost-Pieter Katoen.
\newblock A theory of stochastic systems {P}art {I}: {S}tochastic automata.
\newblock {\em Information and Computation}, 203(1):1--38, 2005.

\bibitem{DP03}
Jos{\'e}e Desharnais and Prakash Panangaden.
\newblock Continuous stochastic logic characterizes bisimulation of
  continuous-time {M}arkov processes.
\newblock {\em Journal of Logic and Algebraic Programming}, 56:99--115, 2003.

\bibitem{doberkat07}
Ernst-Erich Doberkat.
\newblock {\em Stochastic Relations~--~Foundations for {M}arkov Transition
  Systems}.
\newblock Chapman \& Hall/CRC, 2007.

\bibitem{FHH+11}
Martin Fr{\"a}nzle, Ernst~Moritz Hahn, Holger Hermanns, Nicol{\'a}s Wolovick,
  and Lijun Zhang.
\newblock Measurability and safety verification for stochastic hybrid systems.
\newblock In {\em Proc. 14th International Conference on Hybrid Systems:
  Computation and Control (HSCC'11)}, pages 43--52, 2011.

\bibitem{FHT08}
Martin Fr{\"a}nzle, Holger Hermanns, and Tino Teige.
\newblock Stochastic satisfiability modulo theory: {A} novel technique for the
  analysis of probabilistic hybrid systems.
\newblock In {\em Proc. 11th International Workshop on Hybrid Systems:
  Computation and Control (HSCC'08)}, pages 172--186, 2008.

\bibitem{GBK16}
Daniel Gburek, Christel Baier, and Sascha Kl{\"u}ppelholz.
\newblock Composition of stochastic transition systems.
\newblock In {\em Proc. 43rd International Colloquium on Automata, Languages
  and Programming (ICALP'16)}, volume~55 of {\em LIPIcs}, pages 102:1--102:15.
  Leibniz-Zentrum f{\"u}r Informatik, 2016.

\bibitem{Glynn89}
Peter~W. Glynn.
\newblock A {GSMP} formalism for discrete event systems.
\newblock {\em Proc. of the IEEE}, 77(1):14--23, 1989.

\bibitem{lncs2500}
Erich Gr\"adel, Wolfgang Thomas, and {\relax Th}omas Wilke, editors.
\newblock {\em Automata, Logics, and Infinite Games: A Guide to Current
  Research}, volume 2500 of {\em Lecture Notes in Computer Science}. Springer,
  2002.

\bibitem{Hen96}
Thomas~A. Henzinger.
\newblock The theory of hybrid automata.
\newblock In {\em Proceedings, 11th Annual {IEEE} Symposium on Logic in
  Computer Science, New Brunswick, New Jersey, USA, July 27-30, 1996}, pages
  278--292. {IEEE} Computer Society, 1996.

\bibitem{HPRV-peva12}
Andr{\'{a}}s Horv{\'{a}}th, Marco Paolieri, Lorenzo Ridi, and Enrico Vicario.
\newblock Transient analysis of non-{M}arkovian models using stochastic state
  classes.
\newblock {\em Performance Evaluation}, 69(7-8):315--335, 2012.

\bibitem{IN97}
S.~Purushothaman Iyer and Murali Narasimha.
\newblock Probabilistic lossy channel systems.
\newblock In {\em Proc. 7th International Joint Conference on Theory and
  Practice of Software Development (TAPSOFT'97)}, volume 1214 of {\em Lecture
  Notes in Computer Science}, pages 667--681. Springer, 1997.

\bibitem{KVAK10}
Vijay~Anand Korthikanti, Mahesh Viswanathan, Gul Agha, and YoungMin Kwon.
\newblock Reasoning about {MDP}s as transformers of probability distributions.
\newblock In {\em Proc. 7th International Conference on Quantitative Evaluation
  of Systems (QEST'10)}, pages 199--208. IEEE Computer Society Press, 2010.

\bibitem{KA-icfem04}
YoungMin Kwon and Gul Agha.
\newblock Linear inequality {LTL} {(iLTL)}: {A} model checker for discrete time
  {M}arkov chains.
\newblock In {\em Proc. 6th International Conference on Formal Engineering
  Methods (ICFEM'04)}, volume 3308 of {\em Lecture Notes in Computer Science},
  pages 194--208. Springer, 2004.

\bibitem{LMS04}
Fran{\c c}ois Laroussinie, Nicolas Markey, and {\relax Ph}ilippe Schnoebelen.
\newblock Model checking timed automata with one or two clocks.
\newblock In {\em Proc. 15th International Conference on Concurrency Theory
  (CONCUR'04)}, volume 3170 of {\em Lecture Notes in Computer Science}, pages
  387--401. Springer, 2004.

\bibitem{MCB84}
Marco~Ajmone Marsan, Gianni Conte, and Gianfranco Balbo.
\newblock A class of generalized stochastic {P}etri nets for the performance
  evaluation of multiprocessor systems.
\newblock {\em ACM Transactions on Computer Systems}, 2(2):93--122, 1984.

\bibitem{Pan01}
Prakash Panangaden.
\newblock Measure and probability for concurrency theorists.
\newblock {\em Theor. Comput. Sci.}, 253(2):287--309, 2001.

\bibitem{Pan09}
Prakash Panangaden.
\newblock {\em Labelled {M}arkov Processes}.
\newblock Imperial College Press, 2009.

\bibitem{PHV16}
Marco Paolieri, Andr{\'a}s Horv{\'a}th, and Enrico Vicario.
\newblock Probabilistic model checking of regenerative concurrent systems.
\newblock {\em IEEE Transactions on Software Engineering}, 42(2):153--169,
  2016.

\bibitem{Pnu83}
Amir Pnueli.
\newblock On the extremely fair treatment of probabilistic algorithms.
\newblock In {\em Proc. 15th Ann.\ Symp.\ Theory of Computing (STOC'83)}, pages
  278--290. ACM Press, 1983.

\bibitem{PZ93}
Amir Pnueli and Lenore~D. Zuck.
\newblock Probabilistic verification.
\newblock {\em Information and Computation}, 103(1):1--29, 1993.

\bibitem{rabinovitch06}
Alexander Rabinovich.
\newblock Quantitative analysis of probablistic lossy channel systems.
\newblock {\em Information and Computation}, 204(5):713--740, 2006.

\bibitem{SMA16}
Sadegh Soudjani, Rupak Majumdar, and Alessandro Abate.
\newblock Safety verification of continuous-space pure jump {M}arkov processes.
\newblock In {\em Proc. 22nd International Conference on Tools and Algorithms
  for the Construction and Analysis of Systems (TACAS'16)}, volume 9636 of {\em
  Lecture Notes in Computer Science}, pages 147--163. Springer, 2016.

\bibitem{vardi85}
Moshe~Y. Vardi.
\newblock Automatic verification of probabilistic concurrent finite-state
  programs.
\newblock In {\em Proc. 26th Annual Symposium on Foundations of Computer
  Science (FOCS'85)}, pages 327--338. IEEE Computer Society Press, 1985.

\bibitem{VW94}
Moshe~Y. Vardi and Pierre Wolper.
\newblock Reasoning about infinite computations.
\newblock {\em Information and Computation}, 115(1):1--37, 1994.

\end{thebibliography}

\newpage
\appendix

\begin{center}
  {\huge\bfseries\sf --~Appendix~--}
\end{center}

Former results already stated in the core of the paper are put in a
box. New results are normally stated without box.

\section{Technical results of Section~\ref{sec:prelim}}

\subsection{Additional technical results for Subsection~\ref{section:PrelimMeasure}}

We discuss some properties on the probability measures that we defined
on paths of an STS.
While not essential for the global understanding of the paper, they
are useful in some of the coming proofs.

\medskip Recall that, if $s \in S$, the \emph{Dirac distribution over
  $s$}, denoted $\delta_s$, is defined for every measurable set $A$,
by $\delta_s(A) = 1$ if $s\in A$, and $\delta_s(A)=0$ otherwise.

Given any initial distribution $\mu$, we can decompose the probability
measure $\Prob^\calT_\mu$ into the various probability measures
$\Prob^\calT_{\delta_s}$ for $s \in S$.

\begin{lemma}
  For every $\varpi\in\calF_{\calT}$,
  \[
  \Prob_{\mu}^{\calT}(\varpi)=\int_{s_0\in S}
  \Prob_{\delta_{s_0}}^{\calT}(\varpi)\mu(\ud s_0)
  \]
\end{lemma}

\begin{proof}
  Observe that if the initial distribution is the Dirac distribution
  $\delta_s$ over state $s \in S$, then we have that
  \[\Prob_{\delta_s}^{\calT}(\Cyl(A_0,\ldots,A_n))=
  \begin{cases}
    0 & \text{if } s\notin A_0,\\
    \Prob_{\kappa(s,\cdot)}^{\calT}(\Cyl(A_1,\ldots,A_n)) &
    \text{otherwise.}
  \end{cases} 
  \]
  It follows that for every $\mu\in\Dist(S)$, we can write
  \[
  \Prob_{\mu}^{\calT}(\Cyl(A_0,\ldots,A_n))=\int_{s_0\in A_0}
  \Prob_{\delta_{s_0}}^{\calT}(\Cyl(A_0,\ldots,A_n)) \mu(\ud s_0)
  \]
  and thus, by uniqueness of the measure extension, for every
  $\varpi\in\calF_{\calT}$,
  \[
  \Prob_{\mu}^{\calT}(\varpi)=\int_{s_0\in S}
  \Prob_{\delta_{s_0}}^{\calT}(\varpi)\mu(\ud s_0).
  \]
  This concludes the proof of the lemma.
\end{proof}

\medskip Recall that given two probability distributions $\mu$ and
$\nu$ over some probability space $(S,\Sigma)$, $\mu$ and $\nu$ are
\emph{qualitatively equivalent} if for each $A\in\Sigma$,
$\mu(A)=0\Leftrightarrow\nu(A)=0$.
The next lemma establishes that two qualitatively equivalent initial
distributions yield two qualitatively equivalent distributions over
paths.

\begin{lemma}
  \label{lemma:equiv}
  Let $\mu$ and $\nu$ be two probability measures over
  $(S,\Sigma)$. If $\mu$ and $\nu$ are qualitatively equivalent, then
  $\Prob_\mu^{\calT}$ and $\Prob_\nu^\calT$ are also qualitatively
  equivalent.
\end{lemma}




\begin{proof}
  We have to show that for each $\pi\in\calF_T$,
  $\Prob_{\mu}^{\calT}(\pi)=0\Leftrightarrow\Prob_{\nu}^{\calT}(\pi)=0$. Since
  the complement of each cylinder is a finite union of cylinders
  and since each denumerable unions of cylinders can be written as a
  denumerable disjoint union of cylinders, it suffices to show this
  for each cylinder $\Cyl(A_0,\ldots,A_n)$ with
  $A_0,\ldots,A_n\in\Sigma$. We have to show that for each
  $A_0,\ldots,A_n\in\Sigma$,
\[
\Prob_{\mu}^{\calT}(\Cyl(A_0,\ldots,A_n))=0\Leftrightarrow\Prob_{\nu}^{\calT}(\Cyl(A_0,\ldots,A_n))=0.
\]
It should be observed that, by symmetry, it suffices to show one of the implications. First, assume $n=0$ and fix $A_0\in\Sigma$. Then from the definition of $\Prob_{\mu}^{\calT}$ and $\Prob_{\nu}^{\calT}$ and from the hypothesis, we get that:
\[\Prob_{\mu}^{\calT}(\Cyl(A_0))=0 \Leftrightarrow \mu(A_0)=0 \Leftrightarrow\nu(A_0)=0 \Leftrightarrow\Prob_{\nu}^{\calT}(\Cyl(A_0))=0. \]
Now consider $n=1$ and fix $A_0,A_1\in\Sigma$. Suppose that $\Prob_{\mu}^{\calT}(\Cyl(A_0,A_1))=0$, i.e. from the definition:
\begin{equation}\label{eq:EquivProbaFirstCase}
\int_{s_0\in A_0} \kappa(s_0, A_1)\mu(\ud s_0)=0.
\end{equation}
Write $B=\lbrace s_0\in A_0\mid \kappa(s_0,A_1)>0\rbrace$. We can write $B=\kappa(\cdot,A_1)^{-1}(\intervaloc{0,1})\cap A_0$ which is in $\Sigma$ from the hypotheses over $\kappa$. From~\eqref{eq:EquivProbaFirstCase}, we can easily check that $\mu(B)=0$, which implies that $\nu(B)=0$ and thus
\[\int_{s_0\in A_0} \kappa(s_0, A_1)\nu(\ud s_0)=0.\]
Using again the definition, it follows that  $\Prob_{\nu}^{\calT}(\Cyl(A_0,A_1))=0$. Now, assume that $n\ge 2$, fix $A_0,\ldots, A_n\in\Sigma$ and assume that $\Prob_{\mu}^{\calT}(\Cyl(A_0,\ldots, A_n))=0$. Remember that
\begin{alignat}{7}
{} & \Prob_{\mu}^{\calT}(\Cyl(A_0,\ldots, A_n))=\notag\\
{} & \quad \int_{s_0\in A_0}\Big(\int_{s_1\in A_1}\dots\Big(\int_{s_{n-1}\in A_{n-1}} \kappa(s_{n-1}, A_n)\kappa(s_{n-2},\ud s_{n-1})\Big)\dots \kappa(s_0, \ud s_1)\Big) \mu(\ud s_0).\notag
\end{alignat}
We inductively define:
\[\begin{cases}
B_{n-1} = \kappa(\cdot, A_n)^{-1}(\intervaloc{0,1})\cap A_{n-1} & \\
B_i = \kappa(\cdot, B_{i+1})^{-1}(\intervaloc{0,1})\cap A_i & \forall 0\leq i\leq n-2.
\end{cases} \]
From the hypotheses over $\kappa$, it is easily seen that for each $0\leq i\leq n-1$, $B_i\in \Sigma$. Let us consider the value $\int_{s_{n-1}\in A_{n-1}} \kappa(s_{n-1}, A_n)\kappa(s_{n-2},\ud s_{n-1})$. From the definition of $B_{n-1}$, it holds that
\begin{alignat}{7}
\int_{s_{n-1}\in A_{n-1}} \kappa(s_{n-1}, A_n)\kappa(s_{n-2},\ud s_{n-1}) & = \int_{s_{n-1}\in B_{n-1}} \kappa(s_{n-1}, A_n)\kappa(s_{n-2},\ud s_{n-1})\notag\\
{} & = \Prob_{\kappa(s_{n-2},\cdot)}^{\calT}(\Cyl(B_{n-1}, A_n)).\notag
\end{alignat}
We thus get that
\begin{alignat}{7}
{} & \Prob_{\mu}^{\calT}(\Cyl(A_0,\ldots, A_n))=\notag\\
{} & \quad \int_{s_0\in A_0}\dots\Big(\int_{s_{n-2}\in A_{n-2}}  \Prob_{\kappa(s_{n-2},\cdot)}^{\calT}(\Cyl(B_{n-1}, A_n))\kappa(s_{n-3},\ud s_{n-2})\Big)\dots \mu(\ud s_0).\notag
\end{alignat}
We prove the two following statements: for each $0\leq i\leq n-2$,
\begin{enumerate}[(a)]
\item $\lbrace s_i\in S\mid \Prob_{\kappa(s_i,\cdot)}^{\calT}(\Cyl(B_{i+1},\ldots,B_{n-1},A_n))>0\rbrace\cap A_i= B_i$ and
\item \[\int_{s_i\in A_i} \Prob_{\kappa(s_i,\cdot)}^{\calT}(\Cyl(B_{i+1},\ldots,B_{n-1},A_n)) \kappa(s_{i-1},\ud s_i)=\Prob_{\kappa(s_{i-1},\cdot)}^{\calT}(\Cyl(B_{i},\ldots,B_{n-1},A_n)),\]
\end{enumerate}
where if $i=0$, $\kappa(s_{i-1},\cdot)$ will stand for the initial distribution $\mu$. Point (a) is here in order to establish that the sets $\lbrace s_i\in S\mid \Prob_{\kappa(s_i,\cdot)}^{\calT}(\Cyl(B_{i+1},\ldots,B_{n-1},A_n))>0\rbrace\cap A_i$ are measurable, and point (b) aims at reducing our integrals to sets whose images have positive values. It should be observed that the second point is an immediate consequence of the first point. We thus only need to prove point (a). We do this by induction over $i$. First, if $i=n-2$, we show that
\[\lbrace s_{n-2}\in S\mid \Prob_{\kappa(s_{n-2},\cdot)}^{\calT}(\Cyl(B_{n-1},A_n))>0\rbrace=\lbrace s_{n-2}\in S\mid \kappa(s_{n-2}, B_{n-1})>0\rbrace \]
which will ensure that (a) is satisfied. First assume that $s_{n-2}\in
S$ is such that 
\[
\Prob_{\kappa(s_{n-2},\cdot)}^{\calT}(\Cyl(B_{n-1},A_n))>0.
\]
 Towards a contradiction, assume that $\kappa(s_{n-2}, B_{n-1})=0$. Then it holds that
\[0=\kappa(s_{n-2}, B_{n-1}) =\Prob_{\kappa(s_{n-2},\cdot)}^{\calT}(\Cyl(B_{n-1})) \geq \Prob_{\kappa(s_{n-2},\cdot)}^{\calT}(\Cyl(B_{n-1},A_n))>0 \]
which is the needed contradiction. Now assume that $\kappa(s_{n-2},B_{n-1})>0$. Then from the definitions of $B_{n-1}$ and of $\Prob_{\kappa(s_{n-2},\cdot)}^{\calT}$, and from classical properties on integrals, it is straightforward to check that the second inclusion holds. Now suppose that point (a) holds for each $i+1\leq j\leq n-2$ for some $i\geq 0$, and let us show that it is still true for $i$. As before, it suffices to establish that
\[\lbrace s_i\in S\mid \Prob_{\kappa(s_i,\cdot)}^{\calT}(\Cyl(B_{i+1},\ldots,B_{n-1},A_n))>0\rbrace=\lbrace s_i\in S\mid \kappa(s_i, B_i)>0\rbrace. \]
The first inclusion can be verified just like in the first case. Now assume that $\kappa(s_i,B_{i+1})>0$. We know that
\begin{alignat}{7}
{} & \Prob_{\kappa(s_i,\cdot)}^{\calT}(\Cyl(B_{i+1},\ldots,B_{n-1},A_n)) =\notag\\ {} & \quad\int_{s_{i+1}\in B_{i+1}} \Prob_{\kappa(s_{i+1},\cdot)}^{\calT}(\Cyl(B_{i+2},\ldots,B_{n-1},A_n))\kappa(s_i,\ud s_{i+1}).\notag
\end{alignat}
Using the induction hypothesis over $i+1$, we get that for each
$s_{i+1}\in B_{i+1}$, 
\[
\Prob_{\kappa(s_{i+1},\cdot)}^{\calT}(\Cyl(B_{i+2},\ldots,B_{n-1},A_n))>0.
\]
And since $\kappa(s_i,B_{i+1})>0$, this induces that
\[\Prob_{\kappa(s_i,\cdot)}^{\calT}(\Cyl(B_{i+1},\ldots,B_{n-1},A_n))>0
\] which concludes that point (a) is satisfied. Hence from points (a)
and (b), we get that \begin{alignat}{7}
  \Prob_{\mu}^{\calT}(\Cyl(A_0,\ldots, A_n)) & = \Prob_{\mu}^{\calT}(\Cyl(B_0,\ldots, B_{n-1}, A_n))\notag\\
  & = \int_{s_0\in B_0}
  \Prob_{\kappa(s_0,\cdot)}^{\calT}(\Cyl(B_1,\ldots, B_{n-1},
  A_n))\mu(\ud s_0).\notag \end{alignat} Since $B_0=\lbrace s_0\in
A_0\mid \Prob_{\kappa(s_0,\cdot)}^{\calT}(\Cyl(B_1,\ldots, B_{n-1},
A_n))>0\rbrace$ and since $\Prob_{\mu}^{\calT}(\Cyl(A_0,\ldots,
A_n))=0$, it follows that $\mu(B_0)=0$. From the hypothesis, we thus
get that $\nu(B_0)$. Now observing that we can prove similarly that
$\Prob_{\nu}^{\calT}(\Cyl(A_0,\ldots, A_n)) =
\Prob_{\nu}^{\calT}(\Cyl(B_0,\ldots, B_{n-1}, A_n))$, we can establish
that $\Prob_{\nu}^{\calT}(\Cyl(A_0,\ldots, A_n)) =0$ which concludes
the proof. \end{proof}

\subsection{Missing proofs of Subsections 2.3 and 2.5}

\bigskip
\noindent\fbox{\begin{minipage}{\linewidth}
\lemmaintegration* \end{minipage}}
\label{app:lemma_integration}

\begin{proof}
  The proof is by induction on $j$. 
  Assume that $j=0$, we have to show:
  \[
  \Prob^{\calT}_\mu(\Cyl(A_0,A_1,\dots,A_n)) =
  \mu(A_0)\cdot\Prob^{\calT}_{\Omega_\calT(\nu_0)}(\Cyl(A_1,\dots,A_n)).
  \] 
  First:
\begin{alignat}{7}
\Prob_{\mu}^{\calT}(\Cyl(A_0,\ldots,A_n)) & = \Prob_{\mu}^{\calT}(\Cyl(A_0)\cap\Cyl(S, A_1,\ldots, A_n))\notag\\
{} & = \Prob_{\mu}^{\calT}(\Cyl(A_0))\cdot \Prob_{\mu}^{\calT}(\Cyl(S, A_1,\ldots,A_n)\mid\Cyl(A_0))\notag\\
{} & = \mu(A_0)\cdot \Prob_{\mu_{A_0}}^{\calT}(\Cyl(A_0,\ldots,A_n)).\notag
\end{alignat}
Now let us unfold
$\Prob^{\calT}_{\Omega_\calT(\nu_0)}(\Cyl(A_1,\dots,A_n))$:
\begin{alignat}{7}
\Prob^{\calT}_{\Omega_\calT(\nu_0)}(\Cyl(A_1,\dots,A_n)) & = \int_{s_1\in A_1} \Prob_{\kappa(s_1,\cdot)}^{\calT}(\Cyl(A_2,\ldots,A_n))(\Omega_{\calT}(\nu_0))(\ud s_1)\notag\\
{} & = \int_{s_1\in A_1} \Prob_{\kappa(s_1,\cdot)}^{\calT}(\Cyl(A_2,\ldots,A_n))\int_{s_0\in S} \kappa(s_0, \ud s_1) \nu_0(\ud s_0)\notag\\
{} & = \int_{s_0\in A_0} \Big(\int_{s_1\in A_1} \Prob_{\kappa(s_1,\cdot)}^{\calT}(\Cyl(A_2,\ldots,A_n)) \kappa(s_0, \ud s_1)\Big) \mu_{A_0}(\ud s_0)\notag\\
{} & = \int_{s_0\in A_0} \Prob_{\kappa(s_0,\cdot)}^{\calT}(\Cyl(A_1,\ldots,A_n)) \mu_{A_0}(\ud s_0)\notag\\
{} & = \Prob_{\mu_{A_0}}^{\calT}(\Cyl(A_0,\ldots,A_n)).\notag
\end{alignat}
This concludes the proof for $j=0$.

Fix $0<j\leq n$ and assume that for each  $0\leq i < j$ the
equality above holds. We will prove that it is still the case for
$j$. First, observe that if $j=n$ then the induction hypothesis states
that
\begin{alignat}{7}
\Prob^{\calT}_\mu(\Cyl(A_0,A_1,\dots,A_n)) & = \mu(A_0)\cdot\prod_{i=1}^{n-1}(\Omega_{\calT}(\nu_{i-1}))(A_i)\cdot
\Prob^\calT_{\Omega_\calT(\nu_{n-1})}(\Cyl(A_n)) \notag\\
{} & = \mu(A_0)\cdot\prod_{i=1}^{n-1}(\Omega_{\calT}(\nu_{i-1}))(A_i)\cdot \Omega_\calT(\nu_{n-1})(A_n) \notag \\
{} & = \mu(A_0)\cdot\prod_{i=1}^{n}(\Omega_{\calT}(\nu_{i-1}))(A_i)\notag
\end{alignat}
which is what we wanted. Otherwise, if $j<n$, then the hypothesis induction states that
\begin{multline*}
\Prob^{\calT}_\mu(\Cyl(A_0,A_1,\dots,A_n)) = \\
\mu(A_0)\cdot\prod_{i=1}^{j-1}(\Omega_{\calT}(\nu_{i-1}))(A_i)\cdot
\Prob^\calT_{\Omega_\calT(\nu_{j-1})}(\Cyl(A_{j},\dots,A_n)).
\end{multline*}
Then using a similar argument as in the first case, we get that
\begin{alignat}{7}
\Prob^\calT_{\Omega_\calT(\nu_{j-1})}(\Cyl(A_{j},\dots,A_n)) &=
\Omega_\calT(\nu_{j-1})(A_j) \cdot
\Prob_{\Omega_\calT(\Omega_\calT(\nu_{j-1})_{A_j})}(\Cyl(A_{j+1},\dots,A_n)) \notag
\\
&= \Omega_\calT(\nu_{j-1})(A_j) \cdot
\Prob^\calT_{\Omega_\calT(\nu_{j})}(\Cyl(A_{j+1},\dots,A_n)) \notag
\end{alignat}
since $\nu_{j} = (\Omega_\calT(\nu_{j-1}))_{A_j}$. 
This concludes the proof.
\end{proof}

\bigskip
\noindent\fbox{\begin{minipage}{\linewidth} Proof
    of the fact that the $\sigma$-algebra $\Sigma_p$ (the $\sigma$-algebra product over $S\times Q$)
    coincides with $\Sigma'$, the set of all subsets of $S \times Q$
    of the form $\bigcup_{q\in Q} C_q \times \{q\}$, where $C_q \in
    \Sigma$ for every $q \in Q$
    (stated page~\pageref{product-sigma-algebra}).
  \end{minipage}}\label{app:product-sigma-algebra}

\begin{proof}
  It suffices to show that
\begin{enumerate}[(i)]
\item $\Sigma'$ contains all rectangles;
\item $\Sigma'\subseteq \Sigma_p$; and
\item $\Sigma'$ is a $\sigma$-algebra.
\end{enumerate}
Property (i) follows from the decomposition of any rectangle $X\times
Q'$ into elements of $\Sigma'$:
\[X\times Q' = \bigcup_{q\in Q'} X\times \lbrace q\rbrace \cup
\bigcup_{q\in (Q')^c} \emptyset\times\lbrace q\rbrace.\]

\noindent Property (ii) is straightforward since for every 
$q\in Q$, $C_q\times\lbrace q\rbrace$ is a rectangle and therefore,
the union $\bigcup_{q\in Q} C_q\times\lbrace q\rbrace$ belongs to
the $\sigma$-algebra $\Sigma_p$ generated by the rectangles.

\noindent We finally establish property (iii). First $\Sigma'$ is
non-empty as $\emptyset\in \Sigma'$. Then, for $A=\bigcup_{q\in Q}
C_q\times\lbrace q\rbrace\in\Sigma'$, the complement $A^c =
\bigcup_{q\in Q} C_q^c\times\lbrace q\rbrace$ still belongs to
$\Sigma'$ since $\Sigma$ is a $\sigma$-algebra and hence for each $q$,
$C_q^c\in\Sigma$.  Similarly, we get that $\Sigma'$ is closed under
denumerable unions.
\end{proof}

\bigskip \noindent\fbox{\begin{minipage}{\linewidth} 
    \produit*
  \end{minipage}}\label{app:produit-technique}

%
\begin{proof}
We will establish a link between distributions over $\Paths(\calT)$ and
distributions over $\Paths(\calT\ltimes\calM)$. In order to do so, we
introduce some notations. Given $A_0, A_1,\ldots, A_n\in\Sigma'$ we
write for each $i$, $A_i=\bigcup_{q\in Q} A_{i,q}\times\lbrace
q\rbrace$. Also given $u_1,\ldots, u_n\in 2^\AP$ and $q\in Q$ we
inductively define
\[
\begin{cases}
q_{u_1} = q'\in Q & \text{such that } (q,u_1,q')\in E\\
q_{u_1\ldots u_i} = q'\in Q & \text{such that } (q_{u_1\ldots u_{i-1}}, u_i, q')\in E, \ \forall 2\leq i\leq n.
\end{cases}
\]
Observe that since $\calM$ is deterministic, those states are uniquely
defined. We then have the following result.


\begin{lemma}\label{lemma:ProbProduct} 
  For each initial distribution
  $\mu\in\Dist(S)$ for $\calT$, for each state $q\in Q$ of $\calM$,
  for each $n\in\IN$ and for each $A_0,\ldots, A_n\in\Sigma'$, it
  holds that \begin{multline*}
    \Prob_{\mu \times
      \delta_q}^{\calT\ltimes\calM}(\Cyl(A_0,A_1,\ldots, A_n)) =\notag \\
    \sum_{u_1,\ldots, u_n\in 2^\AP}
    \Prob_{\mu}^{\calT}(\Cyl(A_{0,q}\cap\calL^{-1}(u_1), A_{1,
      q_{u_1}}\cap\calL^{-1}(u_2),\ldots,  \\
    A_{n-1, q_{u_1\ldots
        u_{n-1}}}\cap\calL^{-1}(u_n), A_{n, q_{u_1\ldots u_n}})).
  \end{multline*} 
\end{lemma} 

\begin{proof}[Proof of the lemma]
  We prove it by induction
  over $n$. First if $n=0$, we have to show that for every
  $\mu\in\Dist(S)$, every $q\in Q$ and every $A_0\in\Sigma'$,
  \[\Prob_{\mu \times
    \delta_q}^{\calT\ltimes\calM}(\Cyl(A_0))=\Prob_{\mu}^{\calT}(A_{0,q})
  \] which is trivial from the definition of $\mu \times \delta_q$.
  Now fix $n\ge 0$. Assume that for each $0\leq i\leq n$, the above
  property holds true and show that it is still the case for $i=n+1$.
  Let $\mu\in\Dist(S)$, $q\in Q$ and $A_0,\ldots, A_{n+1}\in\Sigma'$.
  We have that \begin{alignat}{7}\label{eq:ProbProd}
    {} & \Prob_{\mu \times \delta_q}^{\calT\ltimes\calM}(\Cyl(A_0,\ldots, A_{n+1}))\notag\\
    {} & = \int_{(s_0, q')\in A_0}
    \Prob_{\kappa'((s_0,q'),\cdot)}^{\calT\ltimes\calM}(\Cyl(A_1,\ldots,
    A_{n+1}))\ud(\mu \times \delta_q)((s_0,q')) \notag\\
    {} & = \int_{s_0\in A_{0,q}} \Prob_{\kappa'((s_0,q),\cdot)}^{\calT\ltimes\calM}(\Cyl(A_1,\ldots, A_{n+1}))\ud\mu(s_0) \notag\\
    {} & = \sum_{u_1\in 2^\AP} \int_{s_0\in A_{0,q}\cap\calL^{-1}(u_1)} \Prob_{\kappa'((s_0,q),\cdot)}^{\calT\ltimes\calM}(\Cyl(A_1,\ldots, A_{n+1}))\ud\mu(s_0) \notag\\
    {} & = \sum_{u_1\in 2^\AP} \int_{s_0\in
      A_{0,q}\cap\calL^{-1}(u_1)}
    \Prob_{\kappa(s_0,\cdot) \times \delta_{q_{u_1}}}^{\calT\ltimes\calM}(\Cyl(A_1,\ldots,
    A_{n+1}))\ud\mu(s_0) \\
    & \hspace*{3cm} \text{from unicity of } q_{u_1}. \notag
    \end{alignat}
Using the induction hypothesis, we get that
\begin{multline*}
\Prob_{\kappa(s_0,\cdot) \times \delta_{q_{u_1}}}^{\calT\ltimes\calM}(\Cyl(A_1,\ldots, A_{n+1}))=\notag\\
\sum_{u_2,\ldots, u_{n+1}\in 2^\AP} \Prob_{\kappa(s_0,\cdot)}^{\calT}(\Cyl(A_{1,q_{u_1}}\cap\calL^{-1}(u_2), \ldots, A_{n, q_{u_1\ldots u_{n}}}\cap\calL^{-1}(u_{n+1}), A_{n+1, q_{u_1\ldots u_{n+1}}})).
\end{multline*}
Combining with~\eqref{eq:ProbProd}, we thus obtain that
\begin{multline*}
\Prob_{\mu \times \delta_q}^{\calT\ltimes\calM}(\Cyl(A_0,\ldots, A_{n+1}))=\\
\sum_{u_1,\ldots, u_{n+1}\in 2^\AP} \Prob_{\kappa(s_0,\cdot)}^{\calT}(\Cyl(A_{0,q}\cap\calL^{-1}(u_1), \ldots, A_{n, q_{u_1\ldots u_{n}}}\cap\calL^{-1}(u_{n+1}), A_{n+1, q_{u_1\ldots u_{n+1}}}))
\end{multline*}
     which concludes the proof. 
\end{proof}
The proposition is a direct consequence of the previous lemma.
\end{proof}

\section{Technical results of Section~\ref{sec:prop}}

We give here the missing proofs of Section~\ref{sec:prop}.

\subsection{Proof of Lemma~\ref{lem:Btilde}}

\noindent\fbox{\begin{minipage}{\linewidth}
\BtildeMes* \end{minipage}}
\label{app-BtildeMes}

\begin{proof}
We first prove the first point.
Remember that given $B\in\Sigma$, 
\[
\Btilde=\lbrace s\in S\mid \Prob_{\delta_s}^{\calT}(\F B)=0\rbrace.
\]
Observe that we can write:
\[\Btilde=\bigcap_{n \ge 0} \lbrace s\in S\mid \Prob_{\delta_s}^{\calT}(\Cyl(\overbrace{S,\ldots,S}^{n\text{ times}}, B))=0\rbrace. \]
It thus suffices to show that for each $n\ge 0$,
\[\lbrace s\in S\mid \Prob_{\delta_s}^{\calT}(\Cyl(\overbrace{S,\ldots,S}^{n\text{ times}}, B))=0\rbrace\in\Sigma. \]
We will use similar arguments as in the proof of Lemma~\ref{lemma:equiv}. Remember that if $n\ge 1$, it holds that $\Prob_{\delta_s}^{\calT}(\Cyl(\overbrace{S,\ldots,S}^{n\text{ times}}, B))=\Prob_{\kappa(s,\cdot)}^{\calT}(\Cyl(\overbrace{S,\ldots,S}^{n-1\text{ times}}, B))$. First, if $n=0$ then this set corresponds to the set $\lbrace s\in S\mid \delta_s(B)=0\rbrace = B^c$ which is in $\Sigma$. Now if $n=1$ then
\[\lbrace s\in S\mid \Prob_{\kappa(s,\cdot)}(\Cyl(B))=0\rbrace= (\kappa(\cdot,B))^{-1}(\lbrace 0\rbrace) \]
which is in $\Sigma$ from the hypotheses over $\kappa$. Now assume
that $n\ge 2$, it holds that
\[\Prob_{\kappa(s,\cdot)}^{\calT}(\Cyl(\overbrace{S,\ldots,S}^{n-1\text{ times}}, B))= \int_{s_1\in S}\cdots\int_{s_{n-1}\in S} \kappa(s_{n-1}, B)\kappa(s_{n-2},\ud s_{n-1})\cdots\kappa(s_1,\ud s_2)\kappa(s, \ud s_1). \]
We inductively define:
\[\begin{cases}
B_{n-1} = \kappa(\cdot, B)^{-1}(\intervaloc{0,1}) & \\
B_i = \kappa(\cdot, B_{i+1})^{-1}(\intervaloc{0,1}) & \forall 0\leq i\leq n-2.
\end{cases} \]
From the hypotheses over $\kappa$, it holds that $B_i\in\Sigma$ for each $0\leq i < n$. In the sequel, $s_0$ denotes $s$. As in the proof of Lemma~\ref{lemma:equiv}, we can show that firstly, $\int_{s_{n-1}\in S} \kappa(s_{n-1}, B)\kappa(s_{n-2},\ud s_{n-1})=\Prob_{\kappa(s_{n-2},\cdot)}^{\calT}(\Cyl(B_{n-1},B))$ and that for each $1\leq i\leq n-2$,
\begin{enumerate}[(a)]
\item $\lbrace s_i\in S\mid \Prob_{\kappa(s_i,\cdot)}^{\calT}(\Cyl(B_{i+1},\ldots,B_{n-1},B))>0\rbrace = B_i$ and
\item \[\int_{s_i\in S} \Prob_{\kappa(s_i,\cdot)}^{\calT}(\Cyl(B_{i+1},\ldots,B_{n-1},A_n)) \kappa(s_{i-1},\ud s_i)=\Prob_{\kappa(s_{i-1},\cdot)}^{\calT}(\Cyl(B_{i},\ldots,B_{n-1},A_n)).\]
\end{enumerate}
It follows that
\begin{alignat}{7}
\Prob_{\kappa(s,\cdot)}^{\calT}(\Cyl(\overbrace{S,\ldots,S}^{n-1\text{ times}},B)) & = \Prob_{\kappa(s,\cdot)}^{\calT}(\Cyl(B_1,\ldots,B_{n-1},B))\notag\\
{} & = \int_{s_1\in B_1} \Prob_{\kappa(s_1,\cdot)}^{\calT}(\Cyl(B_{2},\ldots,B_{n-1},B)) \kappa(s,\ud s_1)\notag
\end{alignat}
Now since for each $s_1\in B_1$,
$\Prob_{\kappa(s_1,\cdot)}^{\calT}(\Cyl(B_{2},\ldots,B_{n-1},B))>0$,
it holds that
\[\Prob_{\kappa(s,\cdot)}^{\calT}(\Cyl(\overbrace{S,\ldots,S}^{n-1\text{
    times}},B))=0\] if and only if $\kappa(s,B_1)=0$, i.e. if and only
if $s\notin B_0$. And since $B_0\in\Sigma$, it follows that
$B_0^c\in\Sigma$ and thus
\[B_0^c=\lbrace s\in S\mid
\Prob_{\delta_s}^{\calT}(\Cyl(\overbrace{S,\ldots,S}^{n\text{ times}},
B))=0\rbrace\in\Sigma. \]

\medskip The second property is a direct consequence of the definition
of $\widetilde{B}$.  

\medskip We now focus on the third property. Towards a contradiction,
assume that there is $\mu\in\Dist(S)$ such that
$\mu((\widetilde{B})^c)>0$ but $\Prob_{\mu}^{\calT}(\F B)=0$. It
follows that there is $s\in(\widetilde{B})^c$ such that
$\Prob_{\delta_s}^{\calT}(\F B)=0$ and thus $s\in \widetilde{B}$ which
is the wanted contradiction.


\medskip Let us show the fourth item. It should be observed that given $\mu\in\Dist(S)$, $\Prob_{\mu}^{\calT}(\F\G\Btilde)\leq\Prob_{\mu}^{\calT}(\G\F\Btilde)\leq\Prob_{\mu}^{\calT}(\F\Btilde)$. It thus suffices to show that $\Prob_{\mu}^{\calT}(\F\G\Btilde)=\Prob_{\mu}^{\calT}(\F\Btilde)$. Since $\ev{\calT}{\F \G
  \Btilde} \subseteq \ev{\calT}{\F \Btilde}$, towards a contradiction,
we assume that $\Prob_{\mu}^{\calT}(\F \Btilde \wedge \G\F
(\Btilde)^c)>0$. Since
\begin{alignat}{7}
\ev{\calT}{\F \Btilde \wedge \G\F (\Btilde)^c} & \subseteq \ev{\calT}{\bigvee_{n\ge 0}(\F[=n] \Btilde \wedge\F[>n](\Btilde)^c)} \notag\\
{} & = \bigcup_{n\ge
  0}\bigcup_{m> 0} \Cyl(\overbrace{S,\ldots,S}^{n-1\text{ times}},
\Btilde,\overbrace{S,\ldots,S}^{m\text{ times}},
(\Btilde)^c)\notag
\end{alignat}
it follows that there is $n\in\IN$ and $m>0$ such that
\[\Prob_{\mu}^{\calT}(\Cyl(\overbrace{S,\ldots,S}^{n-1\text{ times}}, \Btilde,\overbrace{S,\ldots,S}^{m\text{ times}},(\Btilde)^c))>0. \]

From Lemma~\ref{lemma:integration}, writing $\nu=\Omega_{\calT}^{(n)}(\mu)$, we get that
\[\Prob_{\nu}^{\calT}(\Cyl(\Btilde,\overbrace{S,\ldots,S}^{m\text{ times}},(\Btilde)^c))>0. \]
And from the third property proven previously, we
deduce that
\[\Prob_{\nu_{\Btilde}}^{\calT}(\F B)>0 \]
with $\nu_{\Btilde}\in\Dist(\Btilde)$ which contradicts the second
property of this lemma.

\medskip Finally, we prove the last property. It is straightforward by observing that the two events measured in this equality are exactly the same:
\[
\ev{\calT}{\F B \vee \F \widetilde{B}} =\ev{\calT}{\F B \vee (\neg B \U \widetilde{B})}.
\]
\end{proof}

\subsection{Proof of Proposition~\ref{prop:links}}

\noindent\fbox{\begin{minipage}{\linewidth}
\proplink* \end{minipage}}

\begin{proof}


  From the definitions, the following implications obviously hold
  true. For each $\calB\subseteq \Sigma$ and for each
  $\mu\in\Dist(S)$:
\begin{itemize}
\item $\calT\ \text{is}\ \SD(\mu,\calB) \implies \calT\ \text{is}\
  \D(\mu,\calB)$, and
\item $\calT\ \text{is}\ \PD(\mu,\calB) \implies \calT\ \text{is}\
  \D(\mu,\calB)$.
\end{itemize}


It then turns out that strong decisiveness and persistent decisiveness
are two equivalent notions.

\begin{lemma}\label{lemma:EquivStrPers}
  For each $\calB\subseteq \Sigma$ and for each $\mu\in\Dist(S)$, it
  holds that $\SD(\mu,\calB)$ is equivalent to $\PD(\mu,\calB)$.
\end{lemma}

\begin{proof}
  Fix $\calB\subseteq\Sigma$ and $\mu\in\Dist(S)$. Fix $B\in\calB$ and
  assume that $\calT$ is $\PD(\mu, B)$, i.e. for each $p\ge 0$,
  $\Prob_{\mu}^{\calT}(\F[\ge p] B \vee \F[\ge p] \Btilde)=1$. We
  want to show that $\calT$ is $\SD(\mu,B)$, i.e. that
  $\Prob_\mu^\calT(\G \F B \vee \F \widetilde{B}) =1$, or equivalently
  that $\Prob_{\mu}^{\calT}(\F\G B^c \wedge \G (\Btilde)^c) =0$.
  We have that:
  \begin{eqnarray*}
\Prob_{\mu}^{\calT}(\F\G B^c \wedge \G (\Btilde)^c) & \leq & \sum_{p\ge 0} \Prob_{\mu}^{\calT}(\G_{\ge p} (B^c\cap (\Btilde)^c)\\
{} & = & \sum_{p\ge 0} (1-\Prob_{\mu}^{\calT}(\F[\ge p] B \vee \F[\ge
p] \Btilde))\notag\\
{} & = & 0\ \text{from the hypothesis.}
\end{eqnarray*}
Hence we get that $\Prob_{\mu}^{\calT}(\G\F B \vee \F \Btilde)=1$ and
thus $\calT$ is $\SD(\mu, B)$ and $\SD(\mu,\calB)$ as it holds true
for each $B\in\calB$.

Now fix again $B\in\calB$ and assume that $\calT$ is $\SD(\mu,B)$,
i.e. $\Prob_{\mu}^{\calT}(\G\F B\vee \F \Btilde)=1$. From
Lemma~\ref{lemma:BTildeEquivFGF} (fourth item), we get that
$\Prob_{\mu}^{\calT}(\G\F B\vee \G\F \Btilde)=1$ and it is then
straightforward to establish that for each $p\ge 0$,
$\Prob_{\mu}^{\calT}(\F[\ge p] B\vee \F[\ge p] \Btilde)=1$. We hence
deduce that $\calT$ is $\PD(\mu, B)$ and thus $\PD(\mu,\calB)$ as it
holds true for each $B\in\calB$. This concludes the proof.
\end{proof}

Now, we have the following equivalences between the decisiveness notions.

\begin{lemma}
  For each $\calB\subseteq \Sigma$, it holds that all three notions 
  $\PD(\calB)$, $\SD(\calB)$ and $\D(\calB)$ are equivalent.
\end{lemma}

\begin{proof}
  Fix $\calB\subseteq\Sigma$. From the above results, it only remains
  to prove that $\D(\calB)\Rightarrow\SD(\calB)$ or
  $\D(\calB)\Rightarrow\PD(\calB)$.
We prove the last one. We pick $B\in\calB$ and assume
that $\calT$ is $\D(B)$, i.e. for each $\mu\in\Dist(S)$,
$\Prob_{\mu}^{\calT}(\F B\vee \F\Btilde)=1$. Pick $\mu \in \Dist(S)$
and $i \ge 0$. We get that
\begin{eqnarray*}
\Prob_\mu^{\calT}(\G[\ge i] B^c \wedge \G[\ge i]  (\Btilde)^c) &
\le & \Prob_{\mu_i}^{\calT} (\G(B^c \cap (\Btilde)^c)) \\
 & & \text{where $\mu_i = \Omega_{\calT}^{(i)}(\mu)$, from Lemma~\ref{lemma:integration}} \\
 & & \text{and from a similar argument as in the proof of Lemma~\ref{lemma:attractorGF}} \\
 & \le & 0\ \text{since $\calT$ is $\D(B)$.}
\end{eqnarray*}
Hence for each $i\ge 0$, $\Prob_{\mu}^{\calT}(\F[\ge i] B \vee \F[\ge
i] \Btilde)=1$ and since it holds true for each $\mu\in\Dist(S)$ and each $B\in\calB$, we get that $\calT$ is $\PD(\calB)$.
\end{proof}

Finally, we show the following links between fairness and decisiveness.

\begin{lemma}
For each $\calB\subseteq \Sigma$ and for each $\mu\in\Dist(S)$, it holds that $\SD(\mu,\calB)$ implies $\fair(\mu,\calB)$, and $\SD(\calB)$ implies $\fair(\calB)$.
\end{lemma}

\begin{proof}
Fix $\calB\subseteq\Sigma$ and $\mu\in\Dist(S)$. Assume that $\calT$ is strongly decisive w.r.t. $\calB$ from $\mu$, that is for each $B\in\calB$, $\Prob_{\mu}^{\calT}(\G\F B\vee \F \Btilde)=1$. We want to prove that for each $B\in\calB$, for each $B'\in\PreProb(B)$ with $\Prob_{\mu}^{\calT}(\G\F B')>0$, we have that $\Prob_{\mu}^{\calT}(\G\F B\mid\G\F B')=1$.

Fix $B\in\calB$ and $B'\in\PreProb(B)$ such that $\Prob_{\mu}^{\calT}(\G\F B')>0$. We can notice that
\begin{equation}\label{eq:FairSD}
\Prob_{\mu}^{\calT}(\G\F B'\wedge \F \Btilde)=0.
\end{equation}
Indeed, towards a contradiction, assume that $\Prob_{\mu}^{\calT}(\G\F
B'\wedge \F \Btilde)>0$. Observe that
\[\ev{\calT}{\G\F B'\wedge \F \Btilde} =\bigcup_{n\ge 0} \bigcap_{m\ge 0} \bigcup_{l\ge m} \Cyl(\overbrace{S,\ldots,S}^{n\text{ times}},\Btilde,\overbrace{S,\ldots,S}^{l \text{ times}}, B'). \]  
Then, there are $n,m\in\IN$ such that
\[\Prob_{\mu}^{\calT}(\Cyl(\overbrace{S,\ldots,S}^{n\text{ times}},\Btilde,\overbrace{S,\ldots,S}^{m \text{ times}}, B'))>0. \]
It follows, from Lemma~\ref{lemma:integration} like seen previously, that there is $\nu\in\Dist(S)$ ($\nu=\Omega^{(n)}_{\calT}(\mu)$), such that
\[\Prob_{\nu}^{\calT}(\Cyl(\Btilde, \overbrace{S,\ldots,S}^{m \text{ times}}, B'))>0.\]
And since $B'\in\PreProb(B)$, we get that
\[\Prob_{\nu}^{\calT}(\Cyl(\Btilde, \overbrace{S,\ldots,S}^{m \text{ times}}, B', B))>0.\]
Hence, $\nu(\Btilde)>0$ and we can apply Lemma~\ref{lem:Btilde}
(second item) to obtain a contradiction. Hence,
equation~\eqref{eq:FairSD} holds.
We then write:
\begin{eqnarray*}
1 & = & \Prob_{\mu}^{\calT}(\G\F B\vee \F \Btilde\mid\G\F B') \quad\text{from strong decisiveness}\notag\\
{} & = & \frac{\Prob_{\mu}^{\calT}((\G\F B\vee \F \Btilde)\wedge \G\F B')}{\Prob_{\mu}^{\calT}(\G\F B')}\notag\\
{} & = & \frac{\Prob_{\mu}^{\calT}((\G\F B \wedge \G\F B')\vee(\F
  \Btilde \wedge \G\F B'))}{\Prob_{\mu}^{\calT}(\G\F B')}\notag\\
{} & = & \frac{\Prob_{\mu}^{\calT}(\G\F B\wedge \G\F B')}{\Prob_{\mu}^{\calT}(\G\F B')} \quad \text{from~\eqref{eq:FairSD}}\notag\\
{} & = & \Prob_{\mu}^{\calT}(\G\F B\mid \G\F B')
\end{eqnarray*}
which proves that $\SD(\mu,\calB)\Rightarrow \fair(\mu,\calB)$. Then,
the implication $\SD(\calB)\Rightarrow\fair(\calB)$ is immediate since
the previous implication holds for any initial distribution
$\mu\in\Dist(S)$.
\end{proof}


This concludes the proof of the proposition.
\end{proof}

\section{Technical results of Section~\ref{sec:abstractions}}

\subsection{Additional technical results for
  Subsection~\ref{subsec:abst}}
\label{app:pos}

We now establish several technical results, which make explicit how STSs
are related through an $\alpha$-abstraction. The relationship is only
qualitative, in the sense that it only relates positive reachability
probabilities, but does not relate almost-sure or lower-bounded
probabilities.

\begin{lemma}
  \label{lemma:PushForwardDelta}
  Let $\alpha:(S_1,\Sigma_1)\to(S_2,\Sigma_2)$ be a measurable
  function. Then for every $s\in S_2$ and every
  $\mu\in\Dist(\alpha^{-1}(\lbrace s\rbrace))$,
  $\alpha_{\#}(\mu)=\delta_s$.
\end{lemma}

\begin{proof}
  Fix $s\in S_2$ and $\mu\in\Dist(\alpha^{-1}(\lbrace s\rbrace))$. For
  each $A\in\Sigma_2$, we have that
  $(\alpha_{\#}(\mu))(A)=\mu(\alpha^{-1}(A))$. If $s\in A$, then
  $\alpha^{-1}(\lbrace s\rbrace)\subseteq\alpha^{-1}(A)$ and thus
  $\mu(\alpha^{-1}(A))=1$. Otherwise, if $s\notin A$, then
  $\alpha^{-1}(\lbrace s\rbrace)\cap \alpha^{-1}(A)=\emptyset$ and
  thus $\mu(\alpha^{-1}(A))=0$. This directly implies that
  $\alpha_{\#}(\mu)=\delta_s$.
\end{proof}

\begin{lemma}
  \label{lemma:iterative}
  Assume that $\calT_2$ is an $\alpha$-abstraction of $\calT_1$. Then,
  for every $i \in \nats$, for every $\mu \in \Dist(s_1)$,
  $\alpha_{\#}(\Omega^{(i)}_{\calT_1}(\mu))$ is equivalent to
  $\Omega^{(i)}_{\calT_2}(\alpha_{\#}(\mu))$.
\end{lemma}

\begin{proof}
  We show this by induction on $i$. Case $i=1$ is by definition.  Fix
  some $i\ge 1$ and assume that the statement holds true for each
  $1\leq j\leq i$. By induction hypothesis, we have that
  $\alpha_{\#}(\Omega^{(i)}_{\calT_1}(\mu))$ is equivalent to
  $\Omega^{(i)}_{\calT_2}(\alpha_{\#}(\mu))$. We want to show that
  $\alpha_{\#}(\Omega^{(i+1)}_{\calT_1}(\mu))$ is equivalent to
  $\Omega^{(i+1)}_{\calT_2}(\alpha_{\#}(\mu))$.

  We first notice that
  $\Omega_{\calT_2}(\alpha_{\#}(\Omega^{(i)}_{\calT_1}(\mu)))$ is
  equivalent to $\Omega^{(i+1)}_{\calT_2}(\alpha_{\#}(\mu))$. Indeed
  write $\nu=\alpha_{\#}(\Omega^{(i)}_{\calT_1}(\mu))$ and
  $\nu'=\Omega^{(i)}_{\calT_2}(\alpha_{\#}(\mu))$. From the induction
  hypothesis, we know that $\nu$ and $\nu'$. Following a similar
  argument as in the proof of Lemma~\ref{lemma:equiv} and from the
  definition of $\Omega_{\calT_2}$, we can deduce that
  $\Omega_{\calT_2}(\nu)$ is equivalent to
  $\Omega_{\calT_2}(\nu')$. So it remains to show that
  $\Omega_{\calT_2}(\alpha_{\#}(\mu'))$ is equivalent to
  $\alpha_{\#}(\Omega_{\calT_1}(\mu'))$, when $\mu' =
  \Omega^{(i)}_{\calT_1}(\mu)$. This is by definition of an
  $\alpha$-abstraction.
\end{proof}

In other words, the above lemma states that for each $A\in\Sigma_2$
and for each $i\in \nats$,
\[
\Prob_{\mu}^{\calT_1}(\F[=i] \alpha^{-1}(A))>0\Longleftrightarrow
\Prob_{\alpha_{\#}(\mu)}^{\calT_2}(\F[=i] A)>0 \enspace.
\]
This can even be generalized to cylinders:

\begin{lemma}
  \label{coro:equiv}
  Assume that $\calT_2$ is an $\alpha$-abstraction of $\calT_1$. Then
  for every $\mu \in \Dist(S_1)$, for every $(A_i)_{0 \le i \le n} \in
  \Sigma_2^{n+1}$,
  \[
  \Prob^{\calT_1}_{\mu}(\Cyl(\alpha^{-1}(A_0),\dots
  ,\alpha^{-1}(A_n))) > 0 \Longleftrightarrow
  \Prob^{\calT_2}_{\alpha_{\#}(\mu)}(\Cyl(A_0,\dots,A_n)) >0 \enspace.
  \]
\end{lemma}

\begin{proof}
  We do the proof by induction on $n$. The case $n=0$ is obvious from
  the definition of $\alpha_{\#}$. Now fix $n\ge 1$ and assume that
  for each $0\leq k\leq n-1$, for each $\mu\in\Dist(S_1)$ and for each
  $(A_i)_{0\leq i\leq k}\in\Sigma_2^{k+1}$,
  \[
  \Prob^{\calT_1}_{\mu}(\Cyl(\alpha^{-1}(A_0),\dots \alpha^{-1}(A_k)))
  > 0 \Leftrightarrow
  \Prob^{\calT_2}_{\alpha_{\#}(\mu)}(\Cyl(A_0,\dots,A_k)) >0. 
  \] 
  We show that it is still the case for $n$. Fix $\mu\in\Dist(S_1)$
  and $(A_i)_{i\ge n+1}\in\Sigma_2^{n+2}$. We let $\nu_0 =
  \mu_{\alpha^{-1}(A_0)}$ and $\nu'_0 =
  (\alpha_{\#}(\mu))_{A_0}$. Note that we hence assume that
  $\mu(\alpha^{-1}(A_0))>0$. We first realize that $\nu'_0 =
  \alpha_{\#}(\nu_0)$. Indeed for each $A\in\Sigma_2$,
  \[
  (\alpha_{\#}(\nu_0))(A)= \nu_0(\alpha^{-1}(A))
  =\frac{\mu(\alpha^{-1}(A\cap A_0))}{\mu(\alpha^{-1}(A_0))} =
  \frac{(\alpha_{\#}(\mu))(A\cap A_0)}{(\alpha_{\#}(\mu))(A_0)}
  =\nu'_0(A). 
  \] 
  Then, applying Lemma~\ref{lemma:integration}, we get:
  \begin{alignat}{7}
  {} & \Prob^{\calT_1}_\mu(\Cyl(\alpha^{-1}(A_0),\alpha^{-1}(A_1),\dots,\alpha^{-1}(A_n)))\notag\\
  {} & \quad=\mu(\alpha^{-1}(A_0))\cdot\Prob^{\calT_1}_{\Omega_{\calT_1}(\nu_{0})}(\Cyl(\alpha^{-1}(A_1),\dots,\alpha^{-1}(A_n)))\notag
  \end{alignat}
  and 
  \[
  \Prob^{\calT_2}_{\alpha_{\#}(\mu)}(\Cyl(A_0,A_1,\dots,A_n)) =
  (\alpha_{\#}(\mu))(A_0)\cdot\Prob^{\calT_1}_{\Omega_{\calT_2}(\nu'_0)}(\Cyl(A_1,\dots,A_n)).\
  \]
  By definition of an $\alpha$-abstraction, the measures
  $\Omega_{\calT_2}(\nu'_0)$ and
  $\alpha_{\#}(\Omega_{\calT_1}(\nu_0))$ are equivalent. Hence from Lemma~\ref{lemma:equiv},
  \[\Prob^{\calT_2}_{\Omega_{\calT_2}(\nu'_0)}(\Cyl(A_1,\dots,A_n)) >0\Leftrightarrow
  \Prob^{\calT_2}_{\alpha_{\#}(\Omega_{\calT_1}(\nu_0))}(\Cyl(A_1,\dots,A_n))
  >0.\]
  From the hypothesis of induction, we get that
  \[
  \Prob^{\calT_2}_{\alpha_{\#}(\Omega_{\calT_1}(\nu_0))}(\Cyl(A_1,\dots,A_n))
  >0\Leftrightarrow\Prob^{\calT_1}_{\Omega_{\calT_1}(\nu_{0})}(\Cyl(\alpha^{-1}(A_1),\dots,\alpha^{-1}(A_n)))
  >0.
  \]
  Since $(\alpha_{\#}(\mu))(A_0)= \mu(\alpha^{-1}(A_0))$, we
  conclude: 
  \[
  \Prob^{\calT_1}_\mu(\Cyl(\alpha^{-1}(A_0),\alpha^{-1}(A_1),\dots,\alpha^{-1}(A_n)))>0\Leftrightarrow
  \Prob^{\calT_2}_{\alpha_{\#}(\mu)}(\Cyl(A_0,A_1,\dots,A_n)) >0.
  \] 
  We still have to consider the case where
  $\mu(\alpha^{-1}(A_0)=0$. In that case, $(\alpha_{\#}(\mu))(A_0)=0$
  and thus
  \[
  \Prob^{\calT_1}_\mu(\Cyl(\alpha^{-1}(A_0),\alpha^{-1}(A_1),\dots,\alpha^{-1}(A_n)))=0=\Prob^{\calT_2}_{\alpha_{\#}(\mu)}(\Cyl(A_0,A_1,\dots,A_n))
  \]
  which terminates the proof.
\end{proof}

As an immediate consequence, the positivity of properties with bounded
witnesses are preserved through $\alpha$-abstractions:

\begin{corollary}
  \label{coro:Until}
  Assume that $\calT_2$ is an $\alpha$-abstraction of $\calT_1$. Then
  for every $\mu \in \Dist(S_1)$, for every $A,B \in \Sigma_2$:
  \[
  \Prob_\mu^{\calT_1}(\ev{\calT_1}{\alpha^{-1}(A) \U \alpha^{-1}(B)})
  > 0 \Longleftrightarrow
  \Prob_{\alpha_{\#}(\mu)}^{\calT_2}(\ev{\calT_2}{A \U B}) > 0 \enspace.
  \]
\end{corollary}

\paragraph*{Soundness and completeness of abstractions}

When the abstract system $\calT_2$ is a DMC, soundness and
completeness have a simpler characterization, which will be useful in
the proofs.

\begin{restatable}{lemma}{dmc}
\label{lem:dmc}
  Assume $\calT_2$ is a DMC. Then:
  \begin{itemize}
  \item $\calT_2$ is an $\alpha$-abstraction of $\calT_1$ iff for every
    $s, s'\in S_2$,
    \[
    \kappa_2(s,\{s'\}) >0 \Longleftrightarrow \forall \mu \in
    \Dist(\alpha^{-1}(\lbrace s\rbrace)), \
    \Prob_{\mu}^{\calT_1}(\Cyl(S_1,\alpha^{-1}(\{s'\})))>0 \enspace.
    \]
  \item $\calT_2$ is sound iff for every $s \in S_2$ and 
    every $B\in\Sigma_2$,
\[
      \Prob^{\calT_2}_{\delta_s}(\F B) = 1 \Longrightarrow \forall
      \mu \in \Dist(\alpha^{-1}(\lbrace s\rbrace)), \
      \Prob^{\calT_1}_{\mu}(\F \alpha^{-1}(B)) = 1 \enspace.
\]
  \item $\calT_2$ is complete iff for every $s \in S_2$ and 
    every $B\in\Sigma_2$,
    \[
    \forall \mu \in \Dist(\alpha^{-1}(\lbrace s\rbrace)), \
    \Prob^{\calT_1}_{\mu}(\F \alpha^{-1}(B)) = 1 \Longrightarrow
    \Prob^{\calT_2}_{\delta_s}(\F B) = 1 \enspace.
    \] 
  \end{itemize}
\end{restatable}
\begin{proof}
  We handle the case of soundness. Indeed assume that for each $s\in
  S_2$ and for each $B\in\Sigma_2$, the condition presented in the
  statement (second item) 
  holds
  true. Then fix $\mu\in\Dist(S_1)$, $B\in\Sigma_2$ and assume that
  $\Prob_{\alpha_{\#}(\mu)}^{\calT_2}(\F B)=1$ and show that
  $\Prob_{\mu}^{\calT_1}(\F \alpha^{-1}(B))=1$. Towards a
  contradiction, assume that $\Prob_{\mu}^{\calT_1}(\F
  \alpha^{-1}(B))<1$. Then, since $\calT_2$ is a DMC, there is $s\in
  S_2$ such that $\mu(\alpha^{-1}(s))>0$ and
  \[
  \Prob_{\mu_{\alpha^{-1}(s)}}^{\calT_1}(\F \alpha^{-1}(B)) <1. 
  \]
  From the hypothesis,
  it follows that
  $\Prob_{\delta_s}^{\calT_2}(\F B)<1$. Observe that since
  $\mu(\alpha^{-1}(s))>0$, we have that
  $(\alpha_{\#}(\mu))(s)>0$. Hence we get a contradiction by noticing:
  \[
  \Prob_{\alpha_{\#}(\mu)}^{\calT_2}(\F B)\leq
  (\alpha_{\#}(\mu))(s)\cdot \Prob_{\delta_s}^{\calT_2}(\F B) <1.
  \]
\end{proof}

\subsection{Missing proofs in Subsection~\ref{subsec:transfer}}

\noindent\fbox{\begin{minipage}{\linewidth}
    \MuDecisiveAbstr* \end{minipage}}
 
\medskip
In order to prove Proposition~\ref{thm:MuDecisiveAbstr}, we first show
the following technical lemma, which relates avoid-sets in $\calT_1$
and in $\calT_2$.

\begin{lemma}\label{lemma:BtildeAbstr}
  Let $\calT_2$ be an $\alpha$-abstraction of $\calT_1$. Then, for
  every $B\in\Sigma_2$:
  \(\widetilde{\alpha^{-1}(B)}=\alpha^{-1}(\Btilde). \)
\end{lemma}

\begin{proof}
Fix $B\in\Sigma_2$. We have the series of equivalences:
\begin{alignat}{7}
s\in \widetilde{\alpha^{-1}(B)} & \Longleftrightarrow \Prob_{\delta_s}^{\calT_1}(\F\alpha^{-1}(B))=0\notag\\
{} & \Longleftrightarrow \Prob_{\alpha_{\#}(\delta_s)}^{\calT_2}(\F
B)=0\quad \text{(Corollary~\ref{coro:Until}).}\notag
\end{alignat}
Now from Lemma~\ref{lemma:PushForwardDelta}, one can show that $\alpha_{\#}(\delta_s)=\delta_{\alpha(s)}$ by noticing that $\delta_s\in\Dist(\alpha^{-1}(\alpha(s)))$.
Hence $s\in \widetilde{\alpha^{-1}(B)}$ iff $\alpha(s)\in\Btilde$
(\emph{i.e.}\ $s\in \alpha^{-1}(\Btilde)$), which concludes the
  proof.
\end{proof}

We are now ready to prove Proposition~\ref{thm:MuDecisiveAbstr}.

\begin{proof}[Proof of Proposition~\ref{thm:MuDecisiveAbstr}]
  Fix $B\in\Sigma_2$ and assume that $\calT_2$
  is $\D(\alpha_{\#}(\mu),B)$, \emph{i.e.}\
  \begin{equation}\label{eq:hypoDecAbstr}
  \Prob^{\calT_2}_{\alpha_{\#}(\mu)} (\F B \vee \F \Btilde_2)= 1 \enspace.
  \end{equation}
  To show that $\calT_1$ is $\D(\mu,\alpha^{-1}(B))$, by
  Lemma~\ref{lemma:BtildeAbstr}, it suffices to prove that
  \[\Prob_{\mu}^{\calT_1}(\F \alpha^{-1}(B) \vee \F
  \alpha^{-1}(\Btilde_2))=1 \enspace.\]
The latter is immediate by~\eqref{eq:hypoDecAbstr} since $\calT_2$ is $\mu$-sound.
\end{proof}







\bigskip
\noindent\fbox{\begin{minipage}{\linewidth}
\attractorSound* \end{minipage}}
\label{app-attractorsound}

\begin{proof}
  Fix $B\subseteq S_2$ and $\mu\in\Dist(S_1)$. We want to show that
  $\calT_1$ is $\mu$-decisive w.r.t. $\alpha^{-1}(B)$. We therefore
  have to show that $\Prob_{\mu}^{\calT_1}(\F \alpha^{-1}(B) \vee \F
  \alpha^{-1}(\Btilde))=1$.  Towards a contradiction we assume that
  $\Prob_{\mu}^{\calT_1}(\G (\neg \alpha^{-1}(B)) \wedge \G (\neg
  \alpha^{-1}(\Btilde)))>0$, i.e. $\Prob_{\mu}^{\calT_1}(\G
  \alpha^{-1}(B^c) \wedge \G \alpha^{-1}((\Btilde)^c))>0$.
  Since $A_1 =\alpha^{-1}(A_2)$ is an attractor of $\calT_1$, we
  deduce from Lemma~\ref{lemma:attractorGF} that
  $\Prob_\mu^{\calT_1}(\G \F \alpha^{-1}(A_2))=1$, hence:
  \begin{equation}
    \Prob_{\mu}^{\calT_1}(\G \alpha^{-1}(B^c) \wedge \G
    \alpha^{-1}((\Btilde)^c) \wedge \G \F \alpha^{-1}(A_2))>0\enspace. \label{tutu}
  \end{equation}
  We let $A'_2 \subseteq A_2$ be the subset of states $s$ of $A_2$
  such that:
  \[
  \Prob_{\mu}^{\calT_1}(\G \alpha^{-1}(B^c) \wedge \G
  \alpha^{-1}((\Btilde)^c) \wedge \G \F \alpha^{-1}(\{s\}))>0\enspace.
  \]
  Due to equation~\eqref{tutu}, $A'_2$ is non-empty, and furthermore
  every such $s$ belongs to $B^c$ and $(\widetilde{B})^c$. We set
  $A'_1 = \alpha^{-1}(A'_2)$. 

  In particular, $A'_1\subseteq\alpha^{-1}((\Btilde)^c)$, hence from
  Lemma~\ref{lemma:BTildeComplementaire} (third item) we get that for
  every $\nu\in\Dist(A'_1)$, $\Prob_{\nu}^{\calT_1}(\F
  \alpha^{-1}(B))>0$.  According to hypothesis $(\dag)$, for every $s
  \in A'_2$, we can find $p_s>0$ and $k_s \in \IN$ such that for every
  $\nu_s \in \Dist(\alpha^{-1}(s))$,
  \[\Prob_{\nu_s}^{\calT_1}(\F[\leq k_s] \alpha^{-1}(B))\ge p_s. \]
  Then taking $p=\min\lbrace p_s\mid s \in A'_2\rbrace>0$ and
  $k=\max\lbrace k_s\mid s\in A'_2\rbrace\in\IN$ (since $A'_2$ is
  finite), it holds that for every $\nu\in\Dist(A'_1)$,
  \begin{equation}
    \Prob_{\nu}^{\calT_1}(\F[\leq k] \alpha^{-1}(B))\ge p \qquad
    \text{hence} \qquad \Prob_{\nu}^{\calT_1}(\G[\leq k]
    \alpha^{-1}(B^c))\leq 1-p. \label{eq:upBound}
  \end{equation}
  From~\eqref{tutu}, we deduce that:
  \[
  0<\Prob_{\mu}^{\calT_1}(\G \alpha^{-1}(B^c) \wedge \G
  \alpha^{-1}((\widetilde{B})^c)\wedge\G\F A'_1) \le
  \Prob_{\mu}^{\calT_1}(\G \alpha^{-1}(B^c) \wedge\G\F A'_1)
  \]
  Standardly in the literature (see \emph{e.g.}~\cite[Lemma
  3.4]{ABM07}), one infers immediately from~\eqref{eq:upBound} that
  \[
  \Prob_{\mu}^{\calT_1}(\G \alpha^{-1}(B^c) \wedge\G\F A'_1) \le
  \lim_{n\rightarrow \infty} (1-p)^n=0
  \]
  However we believe this is not so immediate, especially in our
  general setting, and we develop a complete proof below. Note that
  with this result, we exhibit a contradiction, which will conclude the
  proof.

\medskip
It remains to show the last inequality.
%
First we introduce some useful notations. Observe that from the
definition of $A'_1$, it holds that $A'_1\subseteq
\alpha^{-1}(B^c)$. Then for each $j\in\IN$, we will write $B^c_{[j]}$
for the finite sequence $\alpha^{-1}(B^c),\ldots,\alpha^{-1}(B^c)$
where $\alpha^{-1}(B^c)$ occurs exactly $j$ times, and similarly we
will write $(B^c\setminus A'_1)_{[j]}$ for the finite sequence
$\alpha^{-1}(B^c)\setminus A'_1,\ldots,\alpha^{-1}(B^c)\setminus A'_1$
where $\alpha^{-1}(B^c)\setminus A'_1$ occurs exactly $j$ times. Then
observe that
\begin{multline}\label{eq:charCylAFsound}
  \ev{\calT_1}{\G\F A'_1 \wedge \G \alpha^{-1}(B^c)} = \\
  \bigcap_{n\in\IN} \bigcup_{j_0\in\IN}
  \bigcup_{(j_1,\ldots,j_n)\in\IN^{n}_{\ge k}} \Cyl(B^c_{[j_0]}, A'_1,
  B^c_{[j_1]},A'_1, B^c_{[j_2]},\ldots, B^c_{[j_{n-1}]}, A'_1,
  B^c_{[j_n]}),
\end{multline}
where $\IN_{\ge k}$ denotes the set of natural numbers larger than or
equal to $k$. We depict such a cylinder and what we can infer on the
probabilities on Figure~\ref{Figure:schemeAttrSound}. As all
behaviours are always in $\alpha^{-1}(B^c)$, the big rectangle
represents this set, while the small one represents $A'_1\subseteq
\alpha^{-1}(B^c)$ which we know is reached infinitely often with
probability~$1$. The behaviours are thus decomposed according to each
visit in $A'_1$ followed by at least $k$ moves (while staying in
$A'_1$). The dashed arrows represent the $k$ first steps. Note that
within those $k$ steps, $A'_1$ could be reached but it has no
importance. What matters here is the fact that from $A'_1$, the
probability of the next $k$ steps within $\alpha^{-1}(B^c)$ is upper
bounded by $1-p$. The curled arrows hold for the next visit to $A'_1$
which we hence know that it will happen with probability $1$.

\begin{figure}[!h]
\centering
\begin{tikzpicture}
\tikzstyle{ptt}=[scale=1]
\tikzstyle{loc}=[ptt,draw,circle,minimum size =0.8cm];
\tikzstyle{locRec}=[ptt,draw,rectangle,minimum height =0.7cm, minimum width =0.9cm, rounded corners = 3pt];
\tikzstyle{locRecInv}=[ptt,rectangle,minimum height =0.7cm, minimum width =0.9cm, rounded corners = 3pt];
\tikzstyle{inv}=[ptt,circle,minimum size =1cm];
\tikzstyle{fleche}=[->,>=stealth', rounded corners=1pt];
\tikzstyle{int}=[->,>=stealth', rounded corners=1pt, dashed];
\tikzstyle{visit}=[->, >=stealth', decorate, decoration={snake,amplitude=.1cm,post length=1mm}]

\draw[rounded corners, thick] (0, 0) rectangle (13, 5);
\draw[rounded corners, thick] (0.5, 0.3) rectangle (12.5, 1.8);

\node[loc] (b1) at (0.75, 4.25) {};
\node[loc] (b2) at (3.25, 4.25) {};
\node[loc] (b3) at (5.75, 4.25) {};
\node[loc] (b4) at (8.25, 4.25) {};
\node[loc] (b5) at (10.75, 4.25) {};

\node[loc] (a1) at (2, 1.05) {};
\node[loc] (a2) at (4.5, 1.05) {};
\node[loc] (a3) at (9.5, 1.05) {};

\draw[visit] (b1) -- (a1) node[midway, below, sloped] {$\scriptstyle \text{1st}\text{ visit to}$} node[midway, above, sloped] {$\scriptstyle \leq 1$};
\draw[visit] (b2) -- (a2) node[midway, below, sloped] {$\scriptstyle \text{2nd}\text{ visit to}$} node[midway, above, sloped] {$\scriptstyle \leq 1$};
\draw[visit] (b4) -- (a3) node[midway, below, sloped] {$\scriptstyle n\text{-th}\text{ visit to}$} node[midway, above, sloped] {$\scriptstyle \leq 1$};

\draw[int] (a1) -- (b2) node[midway, above, sloped] {$\scriptstyle k \text{ steps}$} node[midway, below, sloped] {$\scriptstyle \leq 1-p$};
\draw[int] (a2) -- (b3) node[midway, above, sloped] {$\scriptstyle k \text{ steps}$} node[midway, below, sloped] {$\scriptstyle \leq 1-p$};
\draw[int] (a3) -- (b5) node[midway, above, sloped] {$\scriptstyle k \text{ steps}$} node[midway, below, sloped] {$\scriptstyle \leq 1-p$};

\draw[dotted, very thick] (6.25 ,2.5) -- (7.75, 2.5);
\draw[dotted, very thick] (11.25 ,2.5) -- (12.5, 2.5);

\node (Bc) at (12.2, 4.6) {$\alpha^{-1}(B^c)$};
\node (A) at (12.1, 0.65) {$A'_1$};

\end{tikzpicture} \caption{Scheme for the proof of Proposition~\ref{prop:attractorSound}.}
\label{Figure:schemeAttrSound}
\end{figure}
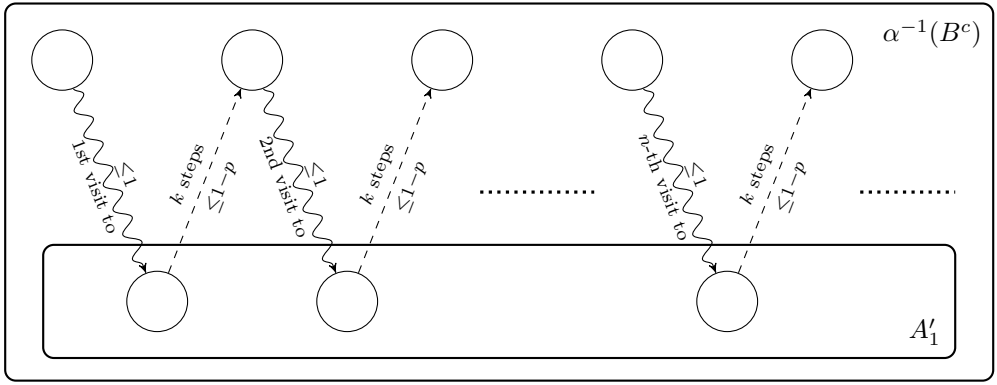

We will prove by induction over $n$ that for each $n\ge 0$ and for
each $\nu\in\Dist(S_1)$,
\begin{equation*}
  \Prob_{\nu}^{\calT_1} \Big( \bigcup_{j_0\in\IN}
  \bigcup_{(j_1,\ldots,j_n)\in\IN^{n}_{\ge k}} \Cyl(B^c_{[j_0]}, A'_1,
  B^c_{[j_1]},A'_1, B^c_{[j_2]},\ldots, B^c_{[j_{n-1}]}, A'_1,
  B^c_{[j_n]}) \Big)\leq (1-p)^{n}. 
\end{equation*}
Observe that for each $n\ge 0$, it holds that
\begin{multline*}
  \bigcup_{j_0\in\IN} \bigcup_{(j_1,\ldots,j_n)\in\IN^{n}_{\ge k}}
  \Cyl(B^c_{[j_0]}, A'_1, B^c_{[j_1]},A'_1, B^c_{[j_2]},\ldots,
  B^c_{[j_{n-1}]}, A'_1, B^c_{[j_n]}) \subseteq \\
  \bigcup_{j_0\in\IN} \bigcup_{(j_1,\ldots,j_{n-1})\in\IN^{n-1}_{\ge
      k}} \Cyl(B^c_{[j_0]}, A'_1, B^c_{[j_1]},A'_1, B^c_{[j_2]},\ldots,
  B^c_{[j_{n-1}]}, A'_1, B^c_{[k]}).
\end{multline*}
Hence it is enough to demonstrate that for each $n\ge 0$ and for each
$\nu\in\Dist(S_1)$,
\begin{equation}\label{eq:IndArgInf}
  \Prob_{\nu}^{\calT_1} \Big( \bigcup_{j_0\in\IN}
  \bigcup_{(j_1,\ldots,j_{n-1})\in\IN^{n-1}_{\ge k}} \Cyl(B^c_{[j_0]},
  A'_1, B^c_{[j_1]},A'_1, B^c_{[j_2]},\ldots, B^c_{[j_{n-1}]}, A'_1, B^c_{[k]})
  \Big)\leq (1-p)^{n}. 
\end{equation}
First fix $n=0$ and $\nu\in\Dist(S_1)$. It corresponds to the two
first arrows on Figure~\ref{Figure:schemeAttrSound}. We will show that
for each $m\ge 0$,
\[
\Prob_{\nu}^{\calT_1} \Big( \bigcup_{j=0}^m \Cyl(B^c_{[j]}, A'_1, B^c_{[k]})
\Big)\leq 1-p,
\]
that is we decompose Figure~\ref{Figure:schemeAttrSound} according to
the length of the first curled arrow. We first prove cases $m=0$ and
$m=1$ in order to illustrate what is happening, and then we will make
the general case. If $m=0$, it then holds that
\begin{alignat}{7}
  \Prob^{\calT_1}_\nu(\Cyl(A'_1, B^c_{[k]})) & =
  \nu(A'_1)\cdot\Prob^{\calT_1}_{\nu_{A'_1}}(\Cyl(A'_1, B^c_{[k]}))\notag\\
  {} & \leq \Prob^{\calT_1}_{\nu_{A'_1}}(\Cyl(B^c_{[k+1]})) \leq 1-p
\end{alignat}
where the first inequality holds from the fact that
$A'_1\subseteq\alpha^{-1}(B^c)$, and the second one
from~\eqref{eq:upBound}. Note that we assumed here that $\nu(A'_1)>0$,
but it has no importance since if $\nu(A'_1)=0$, then the inequality
trivially holds. Now if $m=1$, first observe that
\[
\Cyl(A'_1, B^c_{[k]}) \cup \Cyl(B^c, A'_1, B^c_{[k]}) = \Cyl(A'_1,
B^c_{[k]})\cup \Cyl( B^c\setminus A'_1, A'_1, B^c_{[k]})
\]
where in the second member of the equality, the union is disjoint. It
follows that, writting $\nu'_0=\nu_{B^c\setminus A'_1}$ and
$\nu_1=(\Omega_{\calT_1}(\nu'_0))_{A'_1}$:
\begin{alignat}{7}
  {} & \Prob^{\calT_1}_\nu\big( \Cyl(A'_1, B^c_{[k]}) \cup \Cyl(B^c, A'_1, B^c_{[k]}) \big) \notag\\
  {} & = \Prob^{\calT_1}_\nu( \Cyl(A'_1, B^c_{[k]}))+ \Prob^{\calT_1}_\nu( \Cyl(B^c\setminus A'_1, A'_1, B^c_{[k]}))\notag\\
  {} & \leq \nu(A'_1)\cdot (1-p) + \nu(B^c\setminus A'_1)
  \cdot(\Omega_{\calT_1}(\nu'_0))(A'_1) \cdot \Prob^{\calT_1}_{\nu_1}(
  \Cyl(A'_1, B^c_{[k]}))\ \text{from Lemma~\ref{lemma:integration}}\notag\\
  {} & \leq \nu(A'_1)\cdot (1-p) + \nu(B^c\setminus A'_1)
  \cdot(\Omega_{\calT_1}(\nu'_0))(A'_1) \cdot (1-p) \leq (1-p).\notag
\end{alignat}
Note that we again assumed here that $\nu(B^c\setminus A'_1)>0$ and $(\Omega_{\calT_1}(\nu'_0))(A'_1)>0$, which has again no importance since otherwise, the probability of one of the cylinders would be equal to $0$ and which would thus not interfere on the above inequality. We now prove the general case for $m\ge 2$. Again, we can decompose the union of the cylinders into a disjoint one as follows:
\[
\bigcup_{j=0}^m \Cyl(B^c_j, A'_1, B^c_{[k]}) = \bigcup_{j=0}^m
\Cyl((B^c\setminus A'_1)_j, A'_1, B^c_{[k]})). 
\]
We use the following notations: $\nu'_0 = \nu_{B^c\setminus A'_1}$, $\nu_0=\nu_{A'_1}$, and
\begin{itemize}
\item for each $1\leq i\leq m-1$,
  $\nu'_i=(\Omega_{\calT_1}(\nu'_{i-1}))_{B^c\setminus A'_1}$ and 
\item for each $1\leq i\leq m$,
  $\nu_i=(\Omega_{\calT_1}(\nu'_{i-1}))_{A'_1}$. 
\end{itemize}
Note that we assume again that the conditional probability are well-defined, but like in cases $m=0$ and $m=1$, we can make this supposition w.l.o.g. Then using Lemma~\ref{lemma:integration} and the observation~\eqref{eq:upBound}, we get that:
\begin{alignat}{7}
{} & \Prob^{\calT_1}_\nu\Big(\bigcup_{j=0}^m \Cyl(B^c_{[j]}, A'_1, B^c_{[k]})
\Big) = \sum_{j=0}^m \Prob^{\calT_1}_\nu (\Cyl(B^c_{[j]}, A'_1, B^c_{[k]}))
\notag\\ 
{} & = \nu(A'_1)\cdot \Prob^{\calT_1}_{\nu_0}(\Cyl(A'_1,
B^c_{[k]}))\notag\\ 
{} & + \sum_{j=1}^m \big(\nu(B^c\setminus A'_1) \cdot
\prod_{i=1}^{j-1}(\Omega_{\calT_1}(\nu'_i))(B^c\setminus A'_1) \cdot
(\Omega_{\calT_1}(\nu'_{j-1})(A'_1) \cdot
\overbrace{\Prob^{\calT_1}_{\nu_j}(\Cyl(A'_1, B^c_{[k]}))}^{\leq
  1-p}\big)\notag\\
{} & \leq (1-p) \cdot \Big( \nu(A'_1) + \sum_{j=1}^m \big(\nu(B^c\setminus A'_1) \cdot \prod_{i=1}^{j-1}(\Omega_{\calT_1}(\nu'_0))(B^c\setminus A'_1) \cdot (\Omega_{\calT_1}(\nu'_{j-1})(A'_1)\big)\Big)\notag\\
{} & = (1-p)\cdot \Prob^{\calT_1}_\nu\Big(\bigcup_{j=0}^m
\Cyl(B^c_{[j]}, A'_1) \Big)\leq 1-p\notag 
\end{alignat}
where the last equality comes again from
Lemma~\ref{lemma:integration}, but in the other sense this
time. Finally through the limit over $m$, we obtain
that~\eqref{eq:IndArgInf} is true when $n=0$. 

Now fix $n\ge 0$ and assume that for $0\leq l\leq n$ and for each
$\nu\in\Dist(S_1)$, the inequality~\eqref{eq:IndArgInf} holds true. We
get in particular that for each $\nu\in\Dist(S_1)$,  
\[
\Prob^{\calT_1}_\nu\Big( \bigcup_{j_0\in\IN}
\bigcup_{(j_1,\ldots,j_{n-1})\in\IN^{n-1}_{\ge k}} \Cyl(B^c_{[j_0]},
A'_1, B^c_{[j_1]},A'_1, B^c_{[j_2]},\ldots, B^c_{[j_{n-1}]}, A'_1, B^c_{[k]})
\Big)\leq (1-p)^n. 
\]
We want to show that~\eqref{eq:IndArgInf} is still satisfied for
$n+1$. Like in case $n=0$, we will show that for each $m\ge 0$,
\[
\Prob^{\calT_1}_\nu\Big( \bigcup_{j=0}^m
\bigcup_{(j_1,\ldots,j_{n})\in\IN^{n}_{\ge k}} \Cyl(B^c_{[j]}, A'_1,
B^c_{[j_1]},A'_1, B^c_{[j_2]},\ldots, B^c_{[j_{n}]}, A'_1, B^c_{[k]})
\Big)\leq (1-p)^{n+1}.
\]
We thus again decompose the scheme of
Figure~\ref{Figure:schemeAttrSound} according to the length of the
first arrow. In fact the proof is very similar to the case $n=0$ as
once you hit for the second time $\alpha^{-1}(B^c)$ in the scheme
(\emph{i.e.} after the first dashed arrow), the induction hypothesis
can be applied. What happens before is the exact same behaviour as in
the case for $n=0$. For each $m\ge 0$ this finite union of cylinders
can be decomposed into a finite union of disjoint sets as follows:
\begin{multline*}
\bigcup_{j=0}^m \bigcup_{(j_1,\ldots,j_{n})\in\IN^{n}_{\ge k}}
\Cyl(B^c_{[j]}, A'_1, B^c_{[j_1]},A'_1, B^c_{[j_2]},\ldots, B^c_{[j_{n}]},
A'_1, B^c_{[k]}) = \\ 
\bigcup_{j=0}^m \bigcup_{(j_1,\ldots,j_{n})\in\IN^{n}_{\ge k}}
\Cyl((B^c\setminus A'_1)_{[j]}, A'_1, B^c_{[j_1]},A'_1,
B^c_{[j_2]},\ldots, B^c_{[j_{n}]}, A'_1, B^c_{[k]}).
\end{multline*}
Then using Lemma~\ref{lemma:integration} and this decomposition into a disjoint union, it holds that
\begin{multline*}
\Prob^{\calT_1}_\nu\Big( \bigcup_{j=0}^m
\bigcup_{(j_1,\ldots,j_{n})\in\IN^{n}_{\ge k}} \Cyl(B^c_{[j]}, A'_1,
B^c_{[j_1]},A'_1, B^c_{[j_2]},\ldots, B^c_{[j_{n}]}, A'_1, B^c_{[k]}) \Big)=\\ 
\sum_{j=0}^m \alpha_j \cdot \Prob^{\calT_1}_{\mu_j}\Big(
\bigcup_{j'\in\IN} \bigcup_{(j_2,\ldots,j_{n})\in\IN^{n-1}_{\ge k}}
\Cyl(B^c_{[j']},A'_1, B^c_{[j_2]},\ldots, B^c_{[j_{n}]}, A'_1, B^c_{[k]})
\Big), 
\end{multline*}
where for each $0\leq j\leq m$, $0<\alpha_j<1$ and
$\mu_j\in\Dist(S_1)$ are given by Lemma~\ref{lemma:integration}, where
$\alpha_j$ corresponds to: 
\[
\alpha_j = \Prob^{\calT_1}_\nu(\Cyl((B^c\setminus A'_1)_{[j]}, A'_1,
B^c_{[k]})). 
\]
Note that this is possible due to the fact that we look at the union
of all $j_1\ge k$. Using the induction hypothesis and this last
equality, we get that 
\begin{alignat}{7}
{} & \Prob^{\calT_1}_\nu\Big( \bigcup_{j=0}^m
\bigcup_{(j_1,\ldots,j_{n})\in\IN^{n}_{\ge k}} \Cyl(B^c_{[j]}, A'_1,
B^c_{[j_1]},A'_1, B^c_{[j_2]},\ldots, B^c_{[j_{n}]}, A'_1, B^c_{[k]})
\Big)\notag\\ 
{} & \leq (1-p)^n \cdot \Prob^{\calT_1}_\nu\Big(\bigcup_{j=0}^m
\Cyl(B^c_{[j]}, A'_1,B^c_{[k]})\Big)\notag\\ 
{} & \leq (1-p)^{n+1}\notag
\end{alignat}
where the last inequality stands from what we have done in case $n=0$. Through the limit over $m$, we can thus deduce that~\eqref{eq:IndArgInf} is still true for $n+1$.
 
Finally coming back to~\eqref{eq:charCylAFsound}, through the limit
over $n$ this time, we conclude that 
\[
\Prob_{\mu}^{\calT_1}(\G\F A'_1 \wedge \G \alpha^{-1}(B^c))\leq
\lim_{n\to\infty} (1-p)^n =0.
\]
This concludes the proof.
\end{proof}

\bigskip
\noindent\fbox{\begin{minipage}{\linewidth}
\propfairness* \end{minipage}}
\label{app-fairness}

    \begin{proof}
      As $\calT_2$ is a finite Markov chain, it can be viewed as a
      graph. We can therefore speak of the \emph{bottom strongly
        connected components (BSCC)} of $\calT_2$ (a BSCC is a subset
      $C \subseteq S_2$ such that for all $s,s' \in C$, if $s'$ is
      reachable from $s$, then $s$ is reachable from $s'$ as well).
      We write $\mathsf{BSCC}(\calT_2)$ for the set of BSCCs of
      $\calT_2$. We define $\calC= \lbrace s\in S_2\mid\exists
      C\in\mathsf{BSCC}(\calT_2), \ s\in C\rbrace$. We first prove
      that $\Prob_\mu^{\calT_1}(\F \alpha^{-1}(\calC))=1$. In order to
      establish this, we show that for each $s\in S_2$,
      $\Prob_\mu^{\calT_1}(\G\F \alpha^{-1}(s))>0$ implies that $s\in
      \calC$.
Indeed, pick $s \in S_2$ such that:
      \[
      \Prob^{\calT_1}_\mu(\G \F \alpha^{-1}(\{s\})) >0.
      \]
      We can state that for each $k\ge 1$ and for each $s_0, s_1,\ldots, s_k\in S_2$ with $s_0=s$ and such that for each $0\leq i <k$, $\kappa_2(s_i,s_{i+1})>0$, it holds that 
      \[\Prob_\mu^{\calT_1}(\G\F \alpha^{-1}(s_k)\mid\G\F \alpha^{-1}(s))=1. \] 
      We prove this by induction over $k$. First fix $k=1$ and let $s_1\in S_2$ such that $\kappa_2(s, s_1)>0$. Then for every
      $\nu \in \Dist(\alpha^{-1}(s))$, $\Prob^{\calT_1}_\nu
      (\Cyl(\alpha^{-1}(\{s\}),\alpha^{-1}(\{s_1\}))) > 0$. Hence $\alpha^{-1}(s)
      \in \PreProb^\calT(\{\alpha^{-1}(s_1)\})$. And since $\calT_1$ is fair
      w.r.t. $\alpha$-closed sets, we get that 
      \[
      \Prob^{\calT_1}_\mu(\G \F \alpha^{-1}(\{s_1\}) \mid \G \F
      \alpha^{-1}(\{s\})) =1.
      \]
      Now fix $k> 1$ and assume that for each $1\leq j < k$ and for each $s_0,\ldots, s_j\in S_2$ with $s_0=s$ and such that for each $0\leq i <j$, $\kappa_2(s_i,s_{i+1})>0$, it holds that 
      \[\Prob_\mu^{\calT_1}(\G\F \alpha^{-1}(s_j)\mid\G\F \alpha^{-1}(s))=1. \]
      We want to show that it is still the case for $k$. Fix $s_0,s_1,\ldots, s_k\in S_2$ satisfying all the desired hypotheses. Using the induction hypothesis, we know that $\Prob_\mu^{\calT_1}(\G\F\alpha^{-1}(s_{k-1})\mid\G\F\alpha^{-1}(s_0))=1$ and $\Prob_\mu^{\calT_1}(\G\F\alpha^{-1}(s_{k})\mid\G\F\alpha^{-1}(s_{k-1}))=1$. We can then compute:
\begin{alignat}{7}
{} & \Prob_\mu^{\calT_1}(\G\F\alpha^{-1}(s_{k})\mid\G\F\alpha^{-1}(s_0))\notag\\
{} & \quad = \Prob_\mu^{\calT_1}(\G\F\alpha^{-1}(s_k) \wedge \G\F\alpha^{-1}(s_{k-1})\mid\G\F\alpha^{-1}(s_0))\notag\\
{} & \quad = \Prob_\mu^{\calT_1}(\G\F\alpha^{-1}(s_k) \mid \G\F\alpha^{-1}(s_{k-1})\wedge\G\F\alpha^{-1}(s_0))\cdot\Prob_\mu^{\calT_1}(\G\F\alpha^{-1}(s_{k-1})\mid\G\F\alpha^{-1}(s_0))\notag\\
{} & \quad = 1\notag
\end{alignat}
from the induction hypothesis. 
       This shows that for every state $s'$ which is reachable
      from $s$ in $\calT_2$, 
      \[
      \Prob^{\calT_1}_\mu(\G \F \alpha^{-1}(\{s'\}) \mid \G \F
      \alpha^{-1}(\{s\})) =1.
      \]
      
      Then fix $s'$ reachable from $s$ in $\calT_2$. We can show that $s$ is also reachable from $s'$. Towards a contradiction, assume that it is not the case. It follows that
      \[
      \Prob^{\calT_1}_\mu(\G \F \alpha^{-1}(\{s'\}) \wedge \G \F
      \alpha^{-1}(\{s\})) =0
      \]
      which is a contradiction with $\Prob^{\calT_1}_\mu(\G \F \alpha^{-1}(\{s'\}) \mid \G \F\alpha^{-1}(\{s\})) =1$ and $\Prob^{\calT_1}_\mu(\G \F \alpha^{-1}(\{s\})) >0$. We deduce thus that $s$ belongs to a BSCC of $\calT_2$. 
      
      We can now prove that $\Prob_{\mu}^{\calT_1}(\F \alpha^{-1}(\calC))=1$. Indeed  observe first that from the finiteness of $\calT_2$, it holds that for every paths $\rho=t_0 t_1 t_2 \ldots\in\Paths(\calT_1)$, there is $s\in S_2$ such that $\lbrace i\in\IN\mid t_i\in\alpha^{-1}(s)\rbrace$ is infinite. Keeping this in mind, we write $S_2=\lbrace s_1, \ldots, s_k, s_{k+1}, \ldots, s_n\rbrace$ where $k\ge 1$ and $\lbrace s_1,\ldots, s_k\rbrace=\calC$. Then we can write
\begin{alignat}{7}
\Paths(\calT_1) = & \ev{\calT_1}{\G\F\alpha^{-1}(s_1)} \cup \ev{\calT_1}{\G\F\alpha^{-1}(s_2) \wedge \F\G \neg \alpha^{-1}(s_1)}\notag\\
{} & \cup \cdots \cup \ev{\calT_1}{\G\F \alpha^{-1}(s_n) \wedge \bigwedge_{i=1}^{n-1} \F\G\neg \alpha^{-1}(s_i)}.\notag
\end{alignat}
From what we have shown previously, we now get that for each $j\ge k+1$,
\[0=\Prob_\mu^{\calT_1}(\G\F \alpha^{-1}(s_j)) \ge \Prob_\mu^{\calT_1}(\G\F \alpha^{-1}(s_j)\wedge \bigwedge_{i=1}^{j-1} \F\G\neg \alpha^{-1}(s_i)). \]
And we conclude that
\begin{alignat}{7}
1 & = \Prob_\mu^{\calT_1}(\Paths(\calT_1))\notag\\
{} & = \sum_{j=1}^k \Prob_\mu^{\calT_1}(\G\F \alpha^{-1}(s_j)\wedge \bigwedge_{i=1}^{j-1} \F\G\neg \alpha^{-1}(s_i))\notag\\
{} & \leq \Prob_\mu^{\calT_1}(\F \alpha^{-1}(\calC)).\notag
\end{alignat}
We are now able to prove that $\calT_1$ is $\D(\mu,\calB)$. Fix $B\subseteq S_2$, we want to show that $\Prob_{\mu}^{\calT_1}(\F\alpha^{-1}(B)\vee\F\alpha^{-1}(\Btilde))=1$. We have that
\begin{alignat}{7}
{} & \Prob_{\mu}^{\calT_1}(\F\alpha^{-1}(B)\vee\F\alpha^{-1}(\Btilde))\notag\\
{} & \quad =\sum_{\begin{array}{c} {\scriptstyle C\in\mathsf{BSCC}(\calT_2) \text{ s.t.}}  \\[-.1cm]  {\scriptstyle \Prob_\mu^{\calT_1}(\F\alpha^{-1}(C))>0} \end{array}} \Prob_\mu^{\calT_1}(\F\alpha^{-1}(C))\cdot\Prob_{\mu}^{\calT_1}(\F\alpha^{-1}(B)\vee\F\alpha^{-1}(\Btilde)\mid \F\alpha^{-1}(C)).\notag
\end{alignat}
Now we fix some $C\in\mathsf{BSCC}(\calT_2)$ such that $\Prob_{\mu}^{\calT_1}(\F\alpha^{-1}(C))>0$. There are two cases:
\begin{itemize}
\item first if there is $s\in C$ such that $s\in B$, then $\alpha^{-1}(s)\subseteq \alpha^{-1}(B)$ and thus $\Prob_{\mu}^{\calT_1}(\F\alpha^{-1}(B)\vee\F\alpha^{-1}(\Btilde)\mid \F\alpha^{-1}(C))=1$;
\item or for each $s\in C$, $s\in\Btilde$ which implies that $\alpha^{-1}(C)\subseteq\alpha^{-1}(\Btilde)$ and it that case again $\Prob_{\mu}^{\calT_1}(\F\alpha^{-1}(B)\vee\F\alpha^{-1}(\Btilde)\mid \F\alpha^{-1}(C))=1$.
\end{itemize}
We finally conclude that
\begin{alignat}{7}
 \Prob_{\mu}^{\calT_1}(\F\alpha^{-1}(B)\vee\F\alpha^{-1}(\Btilde)) & =\sum_{\begin{array}{c} {\scriptstyle C\in\mathsf{BSCC}(\calT_2) \text{ s.t.}}  \\[-.1cm]  {\scriptstyle \Prob_\mu^{\calT_1}(\F\alpha^{-1}(C))>0} \end{array}} \Prob_\mu^{\calT_1}(\F\alpha^{-1}(C))\notag\\
{} & = \Prob_{\mu}^{\calT_1}(\F\alpha^{-1}(\calC))=1.\notag
\end{alignat}

    \end{proof}

\section{Technical results of Section~\ref{sec:main}}

\label{appendix:main}

\noindent\fbox{\begin{minipage}{\linewidth}
\attractorproduit* \end{minipage}}
\label{app-attractorproduit}

We first prove the following lemma.

\begin{lemma}\label{lemma:soundproduct1}
  Fix $\mu\in\Dist(S)$ and assume that $A\in\Sigma$ is a
  $\mu$-attractor for $\calT$. Then for each $q\in Q$, $A\times Q$ is
  a $(\mu \times \delta_q)$-attractor for $\calT\ltimes\calM$.
\end{lemma}

\begin{proof}
  Fix $\mu\in\Dist(S)$ and $A\in\Sigma$ such that
  $\Prob_{\mu}^{\calT}(\F A)=1$. Fix $q\in Q$. We know that
\[\ev{\calT\ltimes\calM}{\F A\times Q} = \ev{\calT\ltimes\calM}{\bigcup_{n\in\IN} \Cyl(\overbrace{S',\ldots,S'}^{n\text{ times}}, A\times Q)}. \]
Then from Lemma~\ref{lemma:ProbProduct}, we know that for each $n\in\IN$
\begin{alignat}{7}
\Prob_{\mu \times \delta_q}^{\calT\ltimes\calM}(\Cyl(\overbrace{S',\ldots,S'}^{n\text{ times}}, A\times Q)) & = \sum_{u_1,\ldots,u_n\in 2^\AP} \Prob_{\mu}^{\calT}(\Cyl(\calL^{-1}(u_1),\ldots,\calL^{-1}(u_n), A))\notag\\
{} & = \Prob_{\mu}^{\calT}(\Cyl(\overbrace{S,\ldots,S}^{n\text{ times}}, A\times Q).\notag
\end{alignat}
As this holds true for each $n\ge 0$, we thus get that $\Prob_{\mu
  \times \delta_q}(\F A\times Q)=\Prob_{\mu}^{\calT}(\F A)=1$ from the
hypothesis. This concludes the proof.
\end{proof}

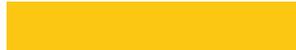
\begin{proof}[Proof of Lemma~\ref{attractorproduit}]
Fix $A\in\Sigma$ such that for each $\mu\in\Dist(S)$, $\Prob_{\mu}^{\calT}(\F A)=1$. We want to prove that for each $\nu\in\Dist(S\times Q)$, $\Prob_{\nu}^{\calT\ltimes \calM}(\F A\times Q)=1$. Fix $\nu\in\Dist(S\times Q)$ and compute:
\[\Prob_{\nu}^{\calT\ltimes\calM}(\F A\times Q) = \sum_{q\in Q} \nu(S\times\lbrace q\rbrace)\cdot \Prob_{\nu_{S\times\lbrace q\rbrace}}^{\calT\ltimes \calM}(\F A\times Q). \]
Note that $\nu_{S\times\lbrace q\rbrace}$ induces a distribution
$\nu_q\in\Dist(S)$ as follows: for each $B\in \Sigma$,
$\nu_q(B)=\nu_{S\times\lbrace q\rbrace}(B\times\lbrace
q\rbrace)$. Writing $\mu=\nu_q$ it then holds that
$\nu_{S\times\lbrace q\rbrace}=\mu \times \delta_q$. We then get, from the hypothesis and Lemma~\ref{lemma:soundproduct1}, that $\Prob_{\nu_{S\times\lbrace q\rbrace}}^{\calT\ltimes \calM}(\F A\times Q)=1$ for each $q\in Q$. Hence, $\Prob_{\nu}^{\calT\ltimes\calM}(\F A\times Q) = \sum_{q\in Q} \nu(S\times\lbrace q\rbrace)=1$ which concludes the proof.
\end{proof}

\bigskip
\noindent\fbox{\begin{minipage}{\linewidth}
\alphabar* \end{minipage}}
\label{app-alphabar}

\begin{proof}

We first show that $\calT_2\ltimes\calM$ is an $\alpha_\calM$-abstraction of $\calT_1\ltimes\calM$. It suffices to show that for each $\mu\in\Dist(S_1)$, for each $q, q'\in Q$ and for each $B_{q'}\in\Sigma_2$,
\begin{equation}\label{eq:abstrProduct}
\Prob_{\mu \times \delta_q}^{\calT_1\ltimes\calM}(\Cyl(S_1\times Q,
\alpha_\calM^{-1}(B_{q'}\times\lbrace q'\rbrace)))>0\Leftrightarrow
\Prob_{(\alpha_\calM)_{\#}(\mu \times \delta_q)}^{\calT_2\ltimes\calM}(\Cyl(S_2\times Q, B_{q'}\times\lbrace q'\rbrace))>0.
\end{equation}
Fix $\mu\in\Dist(S_1)$, $q, q'\in Q$ and $B_{q'}\in \Sigma_2$. Write
$u\in 2^\AP$ for the unique label such that $(q,u,q')\in E$. In order
to prove~\eqref{eq:abstrProduct}, we will use the fact that $\calT_2$
is an $\alpha$-abstraction of $\calT_2$. And in order to make the link
with the wanted equivalence, we will use
Lemma~\ref{lemma:ProbProduct}. We can establish that
$(\alpha_\calM)_{\#}(\mu \times \delta_q) = \alpha_{\#}(\mu) \times
\delta_q$. Indeed given $p\in Q$ and $C_p\in\Sigma_2$, it holds that
\begin{alignat}{7}
(\alpha_\calM)_{\#}(\mu \times \delta_q)(C_p \times\lbrace p\rbrace) &
= (\mu \times \delta_q)(\alpha^{-1}(C_p)\lbrace p\rbrace)\notag\\
{} & = \mu(\alpha^{-1}(C_p))\cdot \delta_q(p)\notag\\
{} & = \alpha_{\#}(\mu)(\alpha^{-1}(C_p))\cdot \delta_q(p) =
(\alpha_{\#}(\mu) \times \delta_q)(C_p\times\lbrace p\rbrace).\notag
\end{alignat}
Hence we get that
\begin{alignat}{7}
\Prob_{(\alpha_\calM)_{\#}(\mu \times \delta_q)}^{\calT_2\ltimes\calM}(\Cyl(S_2\times Q, B_{q'}\times\lbrace q'\rbrace))>0 & \Leftrightarrow \Prob_{\alpha_{\#}(\mu)}^{\calT_2}(\Cyl(\calL_2^{-1}(u), B_{q'}))>0\notag\\
{} & \Leftrightarrow \Prob_{\mu}^{\calT_1}(\Cyl(\calL_1^{-1}(u), \alpha^{-1}(B_{q'})))>0\notag\\
{} & \Leftrightarrow \Prob_{\mu \times \delta_q}^{\calT_1\ltimes\calM}(\Cyl(S_1\times Q, \alpha_\calM^{-1}(B_{q'}\times\lbrace q'\rbrace)))>0\notag
\end{alignat}
where the first and third equivalences hold from Lemma~\ref{lemma:ProbProduct}, and the second equivalence holds from the fact that $\calT_2$ is an $\alpha$-abstraction of $\calT_1$.

Finally, since $\calT_1\ltimes\calM$ is decisive w.r.t $\alpha_\calM^{-1}(B)$ for each $B\in\Sigma'_2$ and since $\calT_2\ltimes\calM$ is an $\alpha_\calM$-abstraction of $\calT_1\ltimes\calM$, Proposition~\ref{coro:DecSound} allows us to conclude that $\calT_2\ltimes\calM$ is a sound $\alpha_\calM$-abstraction of $\calT_1\ltimes\calM$.
\end{proof}

\bigskip
\noindent\fbox{\begin{minipage}{\linewidth}
We give here the (partial) counter-example mentioned in
Remark~\ref{rk:produit}.  \end{minipage}}
\label{app:cex}

\begin{example}
\label{ex:cex-sound}
We illustrate Remark~\ref{rk:produit} by exhibiting an example where
soundness (w.r.t. a fixed distribution) as well as decisiveness
properties do not transfer to the product with a deterministic Muller
automaton.

Consider the DMC $\calT_1$ depicted on the left of
Figure~\ref{Figure:SoundProd} which corresponds to the random walk
over $\IN$ from Example~\ref{Example:DMCRandomWalk}, when
$p=2/3$. Consider also the finite MC $\calT_2$ on the right of the
same figure. Clearly enough, $\calT_2$ is an $\alpha$-abstraction of
$\calT_1$ for the mapping $\alpha:\IN \to \lbrace s_0,s_1,s_2\rbrace$
defined as follows: $\alpha(0)=s_0$, $\alpha(1)=s_1$ and
$\alpha(i)=s_2$ for any $i\ge 2$.

Define $\mu=\delta_0$ as the initial distribution in $\calT_1$.  For
any $B\subseteq \IN$, $\Prob_\mu^{\calT_1}(\F B)=1$ and it follows that
$\calT_2$ is a $\mu$-sound $\alpha$-abstraction of $\calT_1$. It
should be noted that it is however not sound when considering
$\mu'=\delta_1$ as initial distribution. Indeed,
$\Prob_{\mu'}^{\calT_1}(\F \lbrace 0\rbrace)<1$ though
$\Prob_{\delta_{s_1}}^{\calT_2}(\F \lbrace s_0\rbrace)=1$ (and
$\delta_{s_1} = \alpha_{\#}(\mu')$).

\begin{figure}[!h]
  \centering
  \begin{tikzpicture}
    \tikzstyle{ptt}=[scale=1] \tikzstyle{loc}=[ptt,draw,circle,minimum
    size =1cm]; \tikzstyle{inv}=[ptt,circle,minimum size =1cm];
    \tikzstyle{fleche}=[->,>=stealth', rounded corners=1pt];
    \begin{scope}[scale=.8]
    \node[loc] (lzero) at (0,0) {$0$};
    \node[loc] (lun) at (2.2,0) {$1$};
    \node[loc] (ldeux) at (4.4,0) {$2$};
    \node[inv] (lint) at (6.6, 0) {$\cdots$};
    \draw[fleche] (lzero) to[bend left=30] node[ptt,midway, above] {$1$} (lun);
    \draw[fleche] (lun) to[bend left=30] node[ptt,midway, below]
    {$\frac 1 3$} (lzero);
    \draw[fleche] (lun) to[bend left=30] node[ptt,midway, above]
    {$\frac 2 3$} (ldeux);
    \draw[fleche] (ldeux) to[bend left=30] node[ptt,midway, below]
    {$\frac 1    3$} (lun);
    \draw[fleche] (ldeux) to[bend left=30] node[ptt,midway, above]
    {$\frac 2 3$} (lint);
    \draw[fleche] (lint) to[bend left=30] node[ptt,midway, below]
    {$\frac 1 3$} (ldeux);
    \end{scope} 
    \node[loc] (szero) at (7.8, 0) {$s_0$};
    \node[loc] (sun) at (10, 0) {$s_1$};
    \node[loc] (sdeux) at (12.2, 0) {$s_2$};
    \draw[fleche] (sdeux.15)..controls +(45:1) and
    +(315:1)..(sdeux.345) node[ptt,midway, right] {$\frac 2 3$};
    \draw[fleche] (szero) to[bend left=30] node[ptt,midway, above] {$1$} (sun);
    \draw[fleche] (sun) to[bend left=30] node[ptt,midway, below]
    {$\frac 1 3$} (szero);
    \draw[fleche] (sun) to[bend left=30] node[ptt,midway, above]
    {$\frac 2 3$} (sdeux);
    \draw[fleche] (sdeux) to[bend left=30] node[ptt,midway, below]
    {$\frac 1 3$} (sun);
  \end{tikzpicture} \caption{Left, $\calT_1$ a random walk over $\IN$
    and right, its sound finite abstraction $\calT_2$.}
  \label{Figure:SoundProd}
\end{figure}

Consider now the Muller automaton of Section~\ref{sec:prelim} on the left of
Figure~\ref{Figure:MullerAutomaton}. 
As stated in
Lemma~\ref{lemma:alphabar}, it holds that $\calT_2\ltimes\calM$ is an
$\alpha_\calM$-abstraction of $\calT_1\ltimes\calM$ where for each
$n\in\IN$ and each $q\in Q$,
$\alpha_\calM((n,q))=(\alpha(n),q)$. Consider
$\mu \times \delta_{q_0}=\delta_{(0,q_0)}$ and $B=\lbrace (s_0,
q_2)\rbrace$. It then holds that
$(\alpha_{\calM})_{\#}(\mu \times \delta_{q_0})=\delta_{(s_0, q_0)}$ and
that $\alpha_\calM^{-1}(B)=\lbrace (0, q_2)\rbrace$. It is easily
observed that starting in state $(0, q_0)$ (resp. $(s_0, q_0)$) in
$\calT_1\ltimes\calM$ (resp. $\calT_2\ltimes\calM$), then if we visit
in the future a state $(0, q)$ (resp. $(s_0,q)$) we will necessarily
get that $q=q_2$. Keeping this in mind, one can see that
$\Prob_{\delta_(s_0, q_0)}^{\calT_2\ltimes\calM}(\F B
)=1$ while
\[
\Prob_{\mu \times \delta_{q_0}}^{\calT_1\ltimes\calM}(\F
\alpha_\calM^{-1}(B)) = \Prob_{\delta_{(1,
    q_1)}}^{\calT_1\ltimes\calM}(\F \alpha_\calM^{-1}(B))
=\Prob_{\mu'}^{\calT_1}(\F \lbrace 0\rbrace)<1
\]
where the first equality holds from Lemma~\ref{lemma:integration} and
the second equality holds from Lemma~\ref{lemma:ProbProduct}. This
proves that $\calT_2\ltimes\calM$ is not
$(\mu \times \delta_{q_0})$-sound for $\calT_1\ltimes\calM$.

Now, observe that $\calT_1$ is decisive w.r.t. any set of states
$B\subseteq \IN$ from $\mu$ as we have seen that
$\Prob_{\mu}^{\calT_1}(\F B)= 1$ for any set of states $B$. It should
be noted that $\calT_1$ is not decisive by considering $\mu'$ as the
initial distribution and $B=\lbrace 0\rbrace$. In this case,
$\widetilde{\lbrace 0\rbrace}=\emptyset$ and thus
$\Prob_{\mu'}^{\calT_1}(\F \lbrace 0\rbrace \vee \F \widetilde{\lbrace
  0\rbrace}) = \Prob_{\mu'}^{\calT_1}(\F \lbrace
0\rbrace)<1$. Consider now $\calT_1\ltimes\calM$, we have already
shown that $\Prob_{\mu \times \delta_{q_0}}^{\calT_1\ltimes\calM}(\F
\lbrace (0,q_2)\rbrace)<1$. It can be established that
$\widetilde{\lbrace (0, q_2)\rbrace} = (2\IN +1)\times\lbrace q_0,
q_2\rbrace\cup 2\IN\times\lbrace q_1\rbrace$ which are states not
reachable from $(0, q_0)$. We deduce that $\Prob_{\mu \times
  \delta_{q_0}}^{\calT_1\ltimes\calM}(\F \lbrace (0,q_2)\rbrace \vee
\F \widetilde{\lbrace(0, q_2)\rbrace})=\Prob_{\mu \times
  \delta_{q_0}}^{\calT_1\ltimes\calM}(\F \lbrace
(0,q_2)\rbrace)<1$. This shows that $\calT_1\ltimes\calM$ is not
decisive w.r.t. $\lbrace (0, q_2)\rbrace$ from $\mu \times
\delta_{q_0}$.
\end{example}

\section{Technical results of Section~\ref{sec:qualitative}}

\label{appendix:qualitative}

\subsection{Additional technical results for
  Subsection~\ref{subsec:qual-basic}}
\label{app:add-qual}

\begin{lemma}
\label{lemma:QualTechnical}
For every $\mu\in\Dist(S)$
\begin{enumerate}[(i)]
\item $\Prob_{\mu}^{\calT}(\F B \wedge (\neg B\U \Btilde))=0$;
\item $\Prob_{\mu}^{\calT}(\G\F B \wedge \F \Btilde)=0$.
\end{enumerate}
\end{lemma}


\begin{proof}
  We first prove point (i). Since $B$ cannot be reached while we are
  in $\neg B$, it holds that
  \[\Prob_{\mu}^{\calT}(\F B \wedge (\neg B\U \Btilde))=\Prob_{\mu}^{\calT}(\neg B \U (\Btilde\wedge\F B)). \]
  Relaxing the constraint on the until, we get $\Prob^\calT_\mu(\neg B
  \U (\widetilde{B} \wedge \F B)) \leq \Prob^\calT_\mu(\F
  (\widetilde{B} \wedge \F B))$, and the latter is null by definition
  of $\widetilde{B}$. This proves the first item.

  Point (ii) is straightforward from the definition of $\Btilde$ by
  observing that $\Prob_{\mu}^{\calT}(\G\F B \wedge \F
  \Btilde)\leq\Prob^\calT_\mu(\F (\widetilde{B} \wedge \F B))=0$.
\end{proof}

\begin{lemma}
  \label{lemma:QualTechnicalPD}
  For every $\mu\in\Dist(S)$, if $\calT$ is $\PD(\mu, B)$, then
  $\Prob_{\mu}^{\calT}(\F \Btilde \wedge \F \widetilde{\Btilde})=0$.
\end{lemma}


\begin{proof}
  Assume that $\calT$ is $\PD(\mu, B)$, i.e. for each $p\ge 0$,
  $\Prob_{\mu}^{\calT}(\F[\ge p B] \vee \F[\ge p]
  \Btilde)=1$. Towards a contradiction, we suppose that
  $\Prob_{\mu}^{\calT}(\F \Btilde \wedge \F
  \widetilde{\Btilde})>0$. Since
  \[\ev{\calT}{\F \Btilde \wedge \F \widetilde{\Btilde}} = \bigcup_{n\ge 0} \bigcup_{m\ge 0} \ev{\calT}{\F[=n] \Btilde} \cap \ev{\calT}{\F[=m] \widetilde{\Btilde}}, \]
  we deduce that there are $n,m\ge 0$ such that
  $\Prob_{\mu}^{\calT}(\F[=n] \Btilde \wedge \F[=m]
  \widetilde{\Btilde})>0$. We write $e$ for the event
  $e=\ev{\calT}{\F[=n] \Btilde \wedge \F[=m]
    \widetilde{\Btilde}}$. We can show that
  $\Prob_{\mu}^{\calT}(\F[\ge n] B\mid e)=0$ and
  $\Prob_{\mu}^{\calT}(\F[\ge m] \Btilde\mid e)=0$. Indeed we get
  that:
  \begin{alignat}{7}
    \Prob_{\mu}^{\calT}(\F[\ge n] B\mid e) & = \frac{\Prob_{\mu}^{\calT}((\F[\ge n] B)\wedge e)}{\Prob_{\mu}^{\calT}(e)}\notag\\
    {} & \leq  \frac{\Prob_{\mu}^{\calT}(\F[\ge n] B\wedge \F[=n] \Btilde )}{\Prob_{\mu}^{\calT}(e)}\notag\\
    {} & = 0\notag
\end{alignat}
from the definition of $\Btilde$. The equality
$\Prob_{\mu}^{\calT}(\F[\ge m] \Btilde\mid e)=0$ is proved
similarly. Writing $q=\max(m,n)$, it follows that
\[\Prob_{\mu}^{\calT}(\F[\ge q] B\vee \F[\ge q] \Btilde \mid e)=0. \]
And since $\Prob_{\mu}^{\calT}(e)>0$, this contradicts the fact that
$\calT$ is $\PD(\mu, B)$, which concludes the proof.
\end{proof}




\subsection{Qualitative analysis of simple properties}

\begin{proposition}
  \label{app-qualsimple}
  Let $\mu \in \Dist(S)$. Then we have the following
    implications, yielding various characterizations for the
    qualitative analysis of STSs (under specified assumptions):
  \begin{description}
  \item[Almost-sure reachability]~
    \begin{itemize}
    \item if $\Prob^\calT_\mu(\F B) = 1$ then $\Prob^\calT_\mu(\neg B
      \U \widetilde{B}) =0$;
    \item if $\calT$ is $\D(\mu,B)$ and $\Prob^\calT_\mu(\neg B \U
      \widetilde{B}) = 0$, then $\Prob^\calT_\mu(\F B) = 1$.
    \end{itemize}
  \item[Almost-sure repeated reachability]~
    \begin{itemize}
    \item if $\Prob^\calT_\mu(\G \F B) = 1$ then $\Prob^\calT_\mu(\F
      \widetilde{B}) =0$;
    \item if $\calT$ is $\SD(\mu,B)$ and $\Prob^\calT_\mu(\F
      \widetilde{B}) =0$, then $\Prob^\calT_\mu(\G \F B) = 1$.
    \end{itemize}
  \item[Positive repeated reachability]~
    \begin{itemize}
    \item if $\calT$ is $\D(\mu,\widetilde{B})$ and if
      $\Prob^\calT_\mu(\G\F B) >0$, then $\Prob^\calT_\mu(\F
      \widetilde{\widetilde{B}})>0$;
    \item if $\calT$ is $\PD(\mu,B)$ and if $\Prob^\calT_\mu(\F
      \widetilde{\widetilde{B}})>0$, then $\Prob^\calT_\mu(\G\F B)
      >0$.
    \end{itemize}
  \end{description}
\end{proposition}

\begin{proof}
  We start with almost-sure reachability.  We start with the first
  implication. Since the event $\ev{\calT}{\F B}$ is almost-sure, we
  have
  \[ 
  \Prob^\calT_\mu(\neg B \U \widetilde{B}) = \Prob^\calT_\mu((\neg B
  \U \widetilde{B}) \wedge \F B)
  \]
  and then it is straightforward from point (i) of
  Lemma~\ref{lemma:QualTechnical}.

  In order to prove the other implication, we need the assumption that
  $\calT$ is $\D(\mu,B)$. We have that:
  \begin{alignat}{7}
    1 = \Prob_{\mu}^\calT(\F B \vee \F\Btilde) & = \Prob_{\mu}^\calT(\F B \vee (\neg B\U \Btilde)) \quad \text{from Lemma~\ref{lem:Btilde} (fifth item)}\notag\\
    {} & = \Prob_\mu^\calT(\F B) + \Prob_{\mu}^\calT(\neg B\U\Btilde)
    \quad \text{from Lemma~\ref{lemma:QualTechnical} (point
      (i)).}\notag
  \end{alignat}
  Then from $\Prob^\calT_\mu(\neg B \U \widetilde{B})=0$, we derive
  that $\Prob^\calT_\mu(\F B) =1$.

  \medskip We now consider almost-sure repeated reachability.  Since
  the event $\ev{\calT}{\G\F B}$ is almost-sure, we have
  \[
  \Prob^\calT_\mu(\F \widetilde{B}) = \Prob^\calT_\mu(\F \widetilde{B}
  \wedge \G\F B)
  \]
  and then it is straightforward from point (ii) of
  Lemma~\ref{lemma:QualTechnical}.

  In order to prove the second item, we assume that $\calT$ is
  $\SD(\mu,B)$, i.e. $\Prob^{\calT}_\mu(\G \F B \vee \F \widetilde{B})
  = 1$. By assumption, the event $\ev{\calT}{\F \widetilde{B}}$ has
  probability $0$, and thus $\ev{\calT}{\G \F B}$ is almost-sure.

  \medskip We now consider positive repeated reachability.  For the
  first item, we only require $\calT$ to be $\D(\mu,\widetilde{B})$,
  that is $\Prob^\calT_\mu(\F \widetilde{B} \vee \F
  \widetilde{\widetilde{B}})=1$.  Since the event
  $\ev{\calT}{\F\widetilde{B} \vee \F \widetilde{\widetilde{B}}}$ is
  almost-sure, we derive the equality:
  \[
  \Prob^\calT_\mu(\G\F B) = \Prob^\calT_\mu(\G\F B \wedge
  (\F\widetilde{B} \vee \F \widetilde{\widetilde{B}})) \enspace.
  \]
  Now from point (ii) of Lemma~\ref{lemma:QualTechnical}, we get that
  $\Prob^\calT_\mu(\G\F B \wedge (\F\widetilde{B} \vee \F
  \widetilde{\widetilde{B}}))=\Prob^\calT_\mu(\G\F B \wedge \F
  \widetilde{\widetilde{B}})$.  Therefore $\Prob^\calT_\mu(\G\F B
  \wedge \F\widetilde{\widetilde{B}}) =\Prob^\calT_\mu(\G\F B) >0$,
  and thus $\Prob^\calT_\mu(\F\widetilde{\widetilde{B}}) >0$.

  Assume now that $\calT$ is $\PD(\mu,B)$ and that
  $\Prob_{\mu}^{\calT}(\F
  \widetilde{\Btilde})>0$. Lemma~\ref{lemma:QualTechnicalPD} implies
  that $\Prob_{\mu}^{\calT}(\F \Btilde)<1$. Since $\PD(\mu, B)$
  implies $\SD(\mu,B)$, it follows that $\Prob_{\mu}^{\calT}(\G\F
  B\vee \F\Btilde) =1$ and thus, $\Prob_{\mu}^{\calT}(\G\F B)>0$.
\end{proof}

\section{Technical results of Section~\ref{sec:quantitative}}
\label{appendix:quantitative}

\noindent\fbox{\begin{minipage}{\linewidth}
\approxreach* \end{minipage}}
\label{app:approxreach}

\begin{proof}
We have that:
\begin{alignat}{7}
\lim_{n \to +\infty} p_n^{\mathsf{Yes}} + p_n^{\mathsf{No}} & = \Prob_{\mu}^{\calT}(\F B) + \Prob_{\mu}^{\calT}(\neg B \U \Btilde )\notag\\
{} & = \Prob_{\mu}^{\calT}(\F B \vee (\neg B \U \Btilde)) \quad\text{from point (i) of Lemma~\ref{lemma:QualTechnical}}\notag\\
{} & = \Prob_{\mu}^{\calT}(\F B\vee \F \Btilde) \quad \text{from Lemma~\ref{lem:Btilde} (fifth item)}\notag\\
{} & = 1 \enspace.\notag
\end{alignat}
The last equality comes from the decisiveness assumption.
\end{proof}

\noindent\fbox{\begin{minipage}{\linewidth}
\quantrepreach* \end{minipage}}
\label{app:quantrepreach}
 
\begin{proof}
  Since $\calT$ is $\D(\mu, \Btilde)$, it holds that
  $\Prob^\calT_\mu(\F\Btilde \vee \F\widetilde{\Btilde})=1$. Since
  $\calT$ is $\PD(\mu, B)$, one derives from
  Lemma~\ref{lemma:QualTechnicalPD} that
  \[\Prob^\calT_\mu(\F \widetilde{\Btilde}) = 1 - \Prob^\calT_\mu(\F \Btilde).\]
We can now show that
\[1 - \Prob^\calT_\mu(\F \Btilde) =\Prob^\calT_\mu(\G\F B). \] It
comes from the fact that $\PD(\mu, B)$ is equivalent to $\SD(\mu, B)$
and from point (ii) of Lemma~\ref{lemma:QualTechnical}. This
proves the first part of the corollary.

Finally, we can directly establish from
Lemma~\ref{lemma:QualTechnicalPD}
and from the hypothesis $\D(\mu, \Btilde)$, that $\lim_{n\to +\infty}
q^{\mathsf{Yes}}_n + q^{\mathsf{No}}_n=1$.
\end{proof}

\section{Technical results of Section~\ref{sec:appli}}
\label{app:appli}
  In this section, we argue why GSMPs with no cycle of immediate
  events are almost-surely non-zeno. We call \emph{immediate event} a
  fixed-delay event with delay $0$.

\begin{lemma}
  Let $\calG=(Q, \calE, \ell, u, f, \bfE, \Succ)$ be a GSMP with no
  cycle with immediate fixed-delay events. Fix $q_0 \in Q$ an initial
  state, and $\mu$ the measure assigning probability $1$ to
  $q_0$. Then:
  \[
  \Prob_\mu^{\calT_\calG}(\{\rho \in \Paths(\calT_\calG) \mid \rho\
  \text{is zeno}\}) = 0
  \]
 \end{lemma}

 \begin{proof}[Sketch]
   Let $d>0$ be smaller than any constant appearing in the
   non-immediate events of $\calG$. There is $\lambda_0>0$ such that
   for every non-immediate event $e$, $\int_{t=d}^\infty f_e(t) \ud t
   \ge \lambda_0$.

   We consider a non-stochastic interpretation of $\calG$, where
   delays of events are selected non-deterministically in the supports
   of the distributions. Pick a finite run $\rho$ that can be
   generated that way from an initial configuration, and let $\gamma =
   (q,\nu)$ be its last configuration.  In any firable sequence of
   transitions $q \xrightarrow{E_1} q_1 \ldots \xrightarrow{E_N} q_N$
   of length $N>|Q| \cdot |\calE|$ from $\gamma$, there is an event
   which is newly enabled along that sequence, and
   there is $1 \le i < k \le N$ with $e \in \bfE(q_j)$ for every $i
   \le j < k$ and $e \in E_k$.

   Towards a contradiction assume it is not the case, then this means
   that each event $e$ in $E_i$ for some $1 \le i\le N$ is either an
   immediate event or an event in $\bfE(q) \cap \bigcap_{j<i}
   \bfE(q_j) \cap \bigcap_{j<i}E_j^c$ (that is, $e$ was already
   enabled in $q$, it is fired by $E_i$, and was not disabled
   inbetween). There can be at most $|\bfE(q)| \le |\calE|$ such
   events which are not immediate. Furthermore, by assumption, there
   is no cycle with only immediate events. Hence as soon as $N > |Q|
   \cdot |\calE|$, this is not possible. Hence, this implies the above
   claim.

   Hence with probability lower-bounded by $\lambda_0$, the duration
   of a continuation of $\rho$ along that sequence of edges will be
   larger than $d$. Hence, providing more details here, we deduce that
   almost-surely, runs will diverge.
 \end{proof}

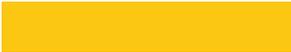
\end{document}